\documentclass[12pt,fleqn]{elsarticle}

\usepackage{amssymb}
\usepackage{amsfonts}
\usepackage{amssymb}
\usepackage{amsmath}
\usepackage{amsthm}

\usepackage[usenames]{color}
\usepackage{colortbl}

\usepackage[czech,russian,ukrainian,english]{babel}

\usepackage{epsfig}
\usepackage{graphicx}
\usepackage{epstopdf}

\usepackage[utf8]{inputenc}

\providecommand{\B}{\mathbf}
\providecommand{\BS}[1]{\boldsymbol{#1}}
 
\providecommand{\R}{\mathrm}

\renewcommand{\Re}{\operatorname{Re}}

\newtheorem{teor}{Theorem}[section]
\newtheorem{lem}{Lemma}[section]
\newtheorem{com}{Comment}[section]

\newcommand{\e}{\mathrm{e}}
\renewcommand{\i}{\textit{i}}
\newcommand{\ii}{\textit{i}}
\newcommand{\w}{\mathrm{w}}
\renewcommand{\d}{\mathrm{d}}
\renewcommand{\c}{\hat c}

\newcommand{\ds}{\displaystyle}
\newcommand{\dsfrac}{\ds\frac}

\renewcommand{\(}{\left(}
\renewcommand{\)}{\right)}
\newcommand{\ol}{\overline}

\renewcommand{\i}{\textit{i}}
\renewcommand{\d}{\mathrm{d}}
\renewcommand{\(}{\left(}
\renewcommand{\)}{\right)}
\renewcommand{\r}{\mathfrak{r}}
\renewcommand{\R}{\mathbb{R}}
\renewcommand{\l}{\mathfrak{l}}
\renewcommand{\Im}{\mbox{Im}}
\renewcommand{\Im}{\mathrm{Im}}

\newenvironment{Proof}
    {\par\noindent{\bf Proof }}
    {\hfill$\scriptstyle\blacksquare$\vskip1cm}

\newenvironment{SketchProof}
    {\par\noindent{\it Sketch of the proof. }}
    {\hfill$\scriptstyle\square$\vskip1cm}

\begin{document}

\begin{frontmatter}



\title{Asymptotics of step-like solutions for the Camassa-Holm equation }


\author{Alexander Minakov}
\address{Mathematical Division, B. Verkin Institute for
Low Temperature Physics\\ and Engineering
of the National Academy of Science of Ukraine\\
47 Lenin Avenue, Kharkiv, 61103, Ukraine
\\
\vskip3mm
Doppler Institute for Mathematical Physics and Applied Mathematics,\\ Czech Technical University in Prague\\
Brehova 7, 11519 Prague, Czech Republic
\\
\vskip3mm Department of
Physics, Faculty of Nuclear Science and Physical
Engineering, \\Czech Technical University in Prague\\
Pohranicni 1288/1, D$\check{e}\check{c}\acute{i}$n, Czech
Republic
\\\vskip3mm
SISSA (International School for Advanced Studies) \\
via Bonomea 265, 34136 Trieste, Italy
\\
\vskip3mm E-mail: minakov.ilt@gmail.com 
}

\begin{abstract}
%%% Text of abstract
%The Cauchy problem for the Camassa -- Holm equation with step-like
%initial conditions is reformulated as a Riemann -- Hilbert
%problem. Then the initial value problem solution is obtained then
%in a parametric form from the Riemann -- Hilbert problem solution.

We study the long-time asymptotics of solution of the Cauchy problem 
for the Camassa-Holm equation with a step-like initial datum.

By using the nonlinear steepest descent method and the so-called $g$-function approach,
we show that the Camassa-Holm equation exhibits a rich
structure of sharply separated regions in the $x,t$-half-plane with qualitatively
different asymptotics, which can be described in terms of a sum of modulated
finite-gap hyperelliptic or elliptic functions and a finite number of solitons.
\end{abstract}

\begin{keyword}
Camassa-Holm equation \sep Riemann-Hilbert problem \sep step-like
initial data \sep $g$-function approach
% \MSC[2010] 37K10 \sep 37K15 \sep 37K40 \sep 35B40 \sep 37K05
%

\end{keyword}

\end{frontmatter}
%We suppose everywhere $\omega>0$, $c>0.$

\section{Introduction}\label{sect: Introduction}
We consider the Cauchy problem for the Camassa -- Holm (CH) equation on the
line with a step-like initial datum
\begin{equation}\label{CH}u_t-u_{txx}+2\omega
u_x+3uu_x=2u_xu_{xx}+uu_{xxx},\quad -\infty<x<\infty, \ t>0,
\end{equation}
\vskip-10mm
\begin{equation}\label{init_cond}u(x,0)=u_0(x), \textrm{where}
\end{equation}
% ---------------- 
\vskip-10mm
\begin{equation}\label{initi_cond_2}
\begin{cases} u_0(x)\to c>0,\quad
x\to-\infty, \qquad \omega>0,\\ u_0(x)\to 0,\quad  x\to+\infty.
\end{cases}
\end{equation}
% --------------- 
We consider real-valued solutions $u(x,t)$.
% -----------
% ----
Based on the Riemann-Hilbert problem formalism \cite{M15}, we study here the long-time asymptotics for the solution of the initial value problem.

The Camassa-Holm equation describes the
unidirectional propagation of shallow water waves over a flat
bottom (\cite {Camassa Holm 1993}, \cite{Camassa Holm Hyman 1994})
as well as axially symmetric waves in a hyperelastic rod
(\cite{Dai}). First, it was found using the method of recursion
operators as a bi-Hamiltonian equation with an infinite number of
conserved functionals (\cite{Fokas Fuchssteiner 1981}).

Study of initial-value problems for integrable equations with step-like initial data originates from the seminal papers \cite{GP}, \cite{Kh2}, 
see also \cite{KK2}, \cite{BikN2}, \cite{Bikb1}, \cite{Bikb2}. However, an implementation of the rigorous 
Riemann-Hilbert problem scheme to step-like initial value problems was done only recently in \cite{BV}, \cite{BIK07}, 
and now this area is actively being developed (see  \cite{Kotlyarov_Minakov_2010}, \cite{KM12}, \cite{EGKT},
\cite{Egorova_Gladka_Kotlyarov_Teschl}, \cite{KM15} and references therein). The importance of such studies is that 
similar methods have been
applied in studies of semiclassical problems \cite{BTVZ07}, Painleve functions \cite{BM14}, random matrix models 
\cite{CV07}, asymptotics of the determinant of Fredholm operator with sine kernel \cite{BDIK15}.

We must mention that an alternative way of asymptotic analysis, which does not employ the fact of fully integrability of 
equation,
can be made by using dispersive methods based on Duhamel's formula and Strichartz estimates (cf. \cite{Pego-Weinstein}, \cite{Tzvetkov},
\cite{Mizumachi-Tzvetkov}, \cite{Germain Pusateri Rousset}). 

However, for the CH equation the only asymptotic analysis was done by using Riemann-Hilbert method 
and nonlinear steepest descent method \cite{Shep2007},
\cite{Shep_2008_halh-line}, \cite{Shep_2008}, \cite{Kostenko_Shepelsky_Teschl_2009}, \cite{Shep_Its_2010}. 
This analysis was restricted to the case of vanishing initial data, and the aim of the present paper is to 
generalize those results to the case of step-like initial data.

%\noindent 
In the sequel we made the following assumptions on the initial datum and solution of the Cauchy problem
(\ref{CH}), (\ref{init_cond}).

\

\noindent\textit{\textbf{Assumption 1. }We restrict ourselves to the
following class of initial conditions:
}
\begin{equation}\label{m+omega} c>0, \omega>0, \quad m_0(x)+\omega>0,\end{equation} where
$m_0(x):=u_0(x)-u_{0,xx}(x)$. 
\begin{com}In \cite[Lemma 2.1]{M15} it is shown that  Assumption 1 provides 
\begin{equation}\label{m+omega>0}
m(x,t)+\omega>0 
\end{equation}
 for all values of time $t\geq0$, where the function $$m(x,t):=u(x,t)-u_{xx}(x,t)$$ is the so-called 
``momentum'' variable. Our method handles only solutions satisfying this property (\ref{m+omega>0}).

\end{com}

\begin{com}The class of initial data (\ref{initi_cond_2}) is equivalent to the following class of initial data:
$$u_0(x)\to \begin{cases}c_{\l},\quad \rm{ as }\quad  x\to-\infty,\qquad \omega\in\R, \\
c_{\r}, \quad\quad\ \ \rm{ as } \quad  x\to+\infty,
\end{cases}$$ 
where
\begin{equation}\nonumber \frac{c_{\l}+\omega}{c_{\r}+\omega}>0,\quad \dsfrac{m_0(x)+\omega}{c_{\r}+\omega}>0.\end{equation}
This can be seen by using the following change of variables, which transforms a solution of the CH equation into solution of the CH equation (\cite{Grunert_Holden_Raynaud_2011}, p.4): $$(\omega,\
u(x,t))\mapsto(\alpha\omega-\beta,\ v(x,t)=\alpha\, u(x-\beta
t,\alpha t)+\beta).$$ This transformation preserves the quantity
$\dsfrac{c_{\l}+\omega}{c_{\r}+\omega}$ and the function
$\dsfrac{m(x,t)+\omega}{c_{\r}+\omega}.$ 
\end{com}

\noindent\textit{
\textbf{Assumption 2. }
We assume that \begin{equation}\label{cond_m<c}\forall x\in\R:\quad c\geq m_0(x)\end{equation}}
with $c$ from (\ref{initi_cond_2}).
\begin{com}
Lemma \ref{Lemma_properties_a,b,r} (cf.~\cite[Lemma 2.6]{M15}) below provides that in this case the corresponding Wronskian
of the associated Jost solutions does not vanish at the edge of the single spectrum: $W\(\i\c\)\neq0$ (see Section \ref{sect: Preliminaries} for details) 
and hence, the (right) transmission coefficient takes continuous boundary values at the edge of one-fold spectrum $k=\i\c.$
\end{com}

\noindent
\textit{\textbf{Assumption 3. } Moreover, we assume that the initial condition tends to its limits 
as $x\to\pm\infty$ exponentially fast,
\begin{equation}\hskip-7mm\hskip-1mm\int\limits_{\R}\hskip-1mm\e^{C_0|x|}
\hskip-1mm\(\hskip-0.5mm|m(\hskip-0.5mmx,0\hskip-0.5mm)\hskip-0.5mm-\hskip-0.5mmc
H\hskip-0.5mm(\hskip-0.5mm-x\hskip-0.5mm) \hskip-0mm
|\hskip-1mm+\hskip-1mm|m_x(\hskip-0.5mmx,0\hskip-0.5mm)|\hskip-1mm+\hskip-1mm|m_{xx}(\hskip-0.5mmx,0\hskip-0.5mm)|\)\mathrm{d}
x\hskip-1mm<\hskip-1mm\infty, \label{assum3_integr_ineq}\end{equation}
where $C_0>\dsfrac{c}{4(c+\omega)}$ with $c$ from (\ref{initi_cond_2}).
}

\begin{com}
This assumption provides analytic continuation of the corresponding spectral functions $a(k)$, $b(k)$ and reflection coefficient $r(k)=\frac{b(k)}{a(k)}$ into some neighborhood of the contour 
of Riemann-Hilbert problem $\Sigma=\R\cup[\i\c,-\i\c]$ (with $\c$ from (\ref{chat})).
\end{com}

The existence of a \textit{weak} solution to step-like initial
value problems for the Camassa-Holm equation as well as for other
nonlinear systems can be established by using PDE techniques
\cite{Grunert_Holden_Raynaud_2011}.
However, our method is suitable to work only with classical solutions, and hence, we have to assume 
the existence of a global classical solution (see Assumption 4 below). Although we must assume the 
existence of a \textit{classical}
solution to this problem, the advantage of our approach is that it
yields a rigorous asymptotic analysis by using an extension of
the nonlinear steepest-descent method \cite{DZ93}.

\

\noindent
\textit{\textbf{Assumption 4.} Suppose that there exists a global real-valued classical
solution $u(x,t)$ of the CH equation (\ref{CH}), which tends rapidly
 to its limits as $x\rightarrow\pm\infty$, that is, for any
$T\geq 0$
\vskip-4mm\begin{equation}\label{conditions u(x,t)}
\max\limits_{0\leq t\leq
T}\int\limits_{-\infty}^{+\infty}\(1+|x|\)\times\left(|m(x,t)-cH(-x)|+|m_x(x,t)|+|m_{xx}(x,t)|\right)\d x<\infty,\end{equation}
% --------------- 
\vskip-10mm \noindent where $H(x)=\left\{\begin{array}{l}1,\quad x\geq0\\0,\quad
x<0\end{array}\right.$ is the Heaviside function.
}
\begin{com}
This assumption establishes the existence of the corresponding Jost solutions for all values of time $t\geq0$ (see Section \ref{sect: Preliminaries} for details).
\end{com}

%\textbf{\textcolor{green}{$\bullet$ Add about existence of solutions, from general methods, from RH}}
The existence of solutions of the Cauchy problem (\ref{CH}), (\ref{init_cond}) satisfying Assumptions 1-4 can 
be established by extensions of the methods, used for the case of vanishing initial data 
(see for example \cite{Constantin_2001}). Since our main goal is the asymptotics, we will not consider this issue 
here. On the other hand, we know that the class of solutions satisfying the Assumptions 1-4 is not empty. Indeed, 
by standard method one can check that each function which is obtained from the RH problem in 
Section \ref{sect: Preliminaries} by formulas (\ref{M1*M2}), is a solution of the CH equation (see for example
\cite{Zakharov_Shabat_74_79}, \cite[Theorem 5.2]{Shep_2008_halh-line}).

\textbf{\textcolor{black}{$\bullet$ Main result.}}

The main result is that the $x,t\geq0$ half-plane is divided into several domains with qualitatively different 
asymptotics of solution. Namely, if we look at asymptotics with accuracy up to a decaying term, then we have 3 
different sectors, and if we look for asymptotics with accuracy up to $t^{-1/2}$, then we have 5 different 
sectors (see Figures \ref{Regions_xt_3_5}, \ref{Regions_xt_13_5}, \ref{Regions_xt_01_5}).

This subdivision considerably depends on whether $\dsfrac{c}{\omega}>3$, $1<\dsfrac{c}{\omega}<3$, 
or $0<\dsfrac{c}{\omega}<1.$

Parameters $\zeta_j, j=1,...,4$ depend on the value of $\dsfrac{c}{\omega}$ and are defined in 
(\ref{xi_j_def_beg})-(\ref{xi_j_def_end}).

The asymptotics in the domains marked in yellow in Figures \ref{Regions_xt_3_5}--\ref{Regions_xt_01_5} is 
described by formulas (\ref{elliptic_asymp_x}), (\ref{elliptic_asymp_u}) in Theorem \ref{Teor_Elliptic}, except for the domain $D_{2a}$ in the case $\dsfrac{c}{\omega}>3;$ 
in the latter case the asymptotics is described by formulas (\ref{hyperelliptic_asymp_x}), 
(\ref{hyperelliptic_asymp_u}) in Theorem \ref{Teor_HyperElliptic}.
The asymptotics in the domains marked in purple is described by Theorem \ref{Teor_constant}. Finally, 
the asymptotics in the domains marked in green is described by Theorem \ref{Teor: soliton_ asympt}.

Briefly, asymptotics in "yellow" domains is described by modulated elliptic (genus 1) or hyperelliptic (genus 2)
functions, 
in "purple" domains the leading asymptotic term is the constant $c$ as in (\ref{initi_cond_2}), and in "green"
domains we have solitonic asymptotics 
on a vanishing background.

\begin{figure}
\vskip-0cm\hskip-.cm\includegraphics
[scale=0.6]
%[width=1.6\textwidth, natwidth=1210, natheight=1642]
{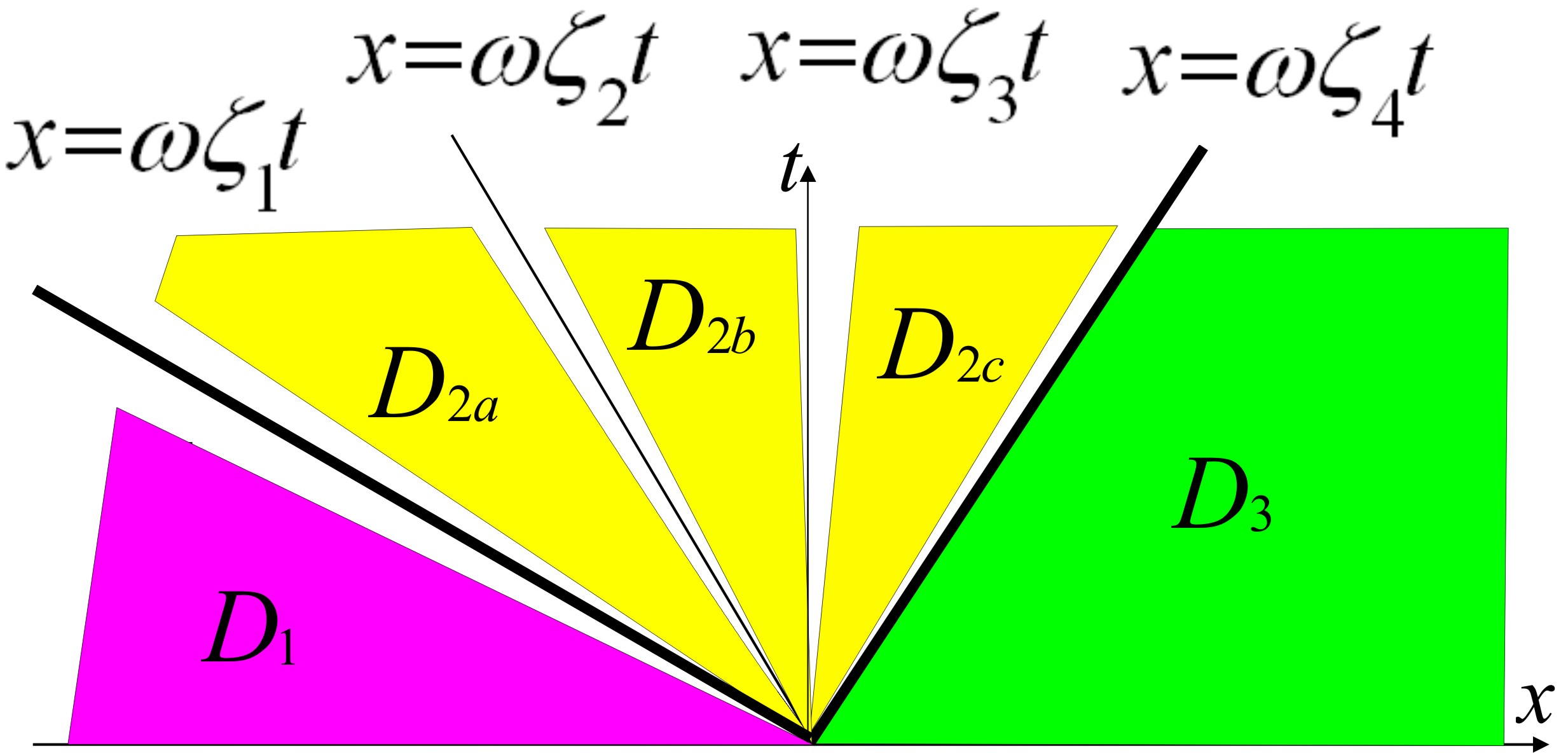}
\caption{Case $\dsfrac{c}{\omega}>3$. Regions in $x,t$-halfplane with qualitatively different asymptotics.}
\label{Regions_xt_3_5}
\end{figure}

\begin{figure}
\vskip-0cm\hskip-.cm\includegraphics
[scale=0.6]
%[width=1.6\textwidth, natwidth=1210, natheight=1642]
{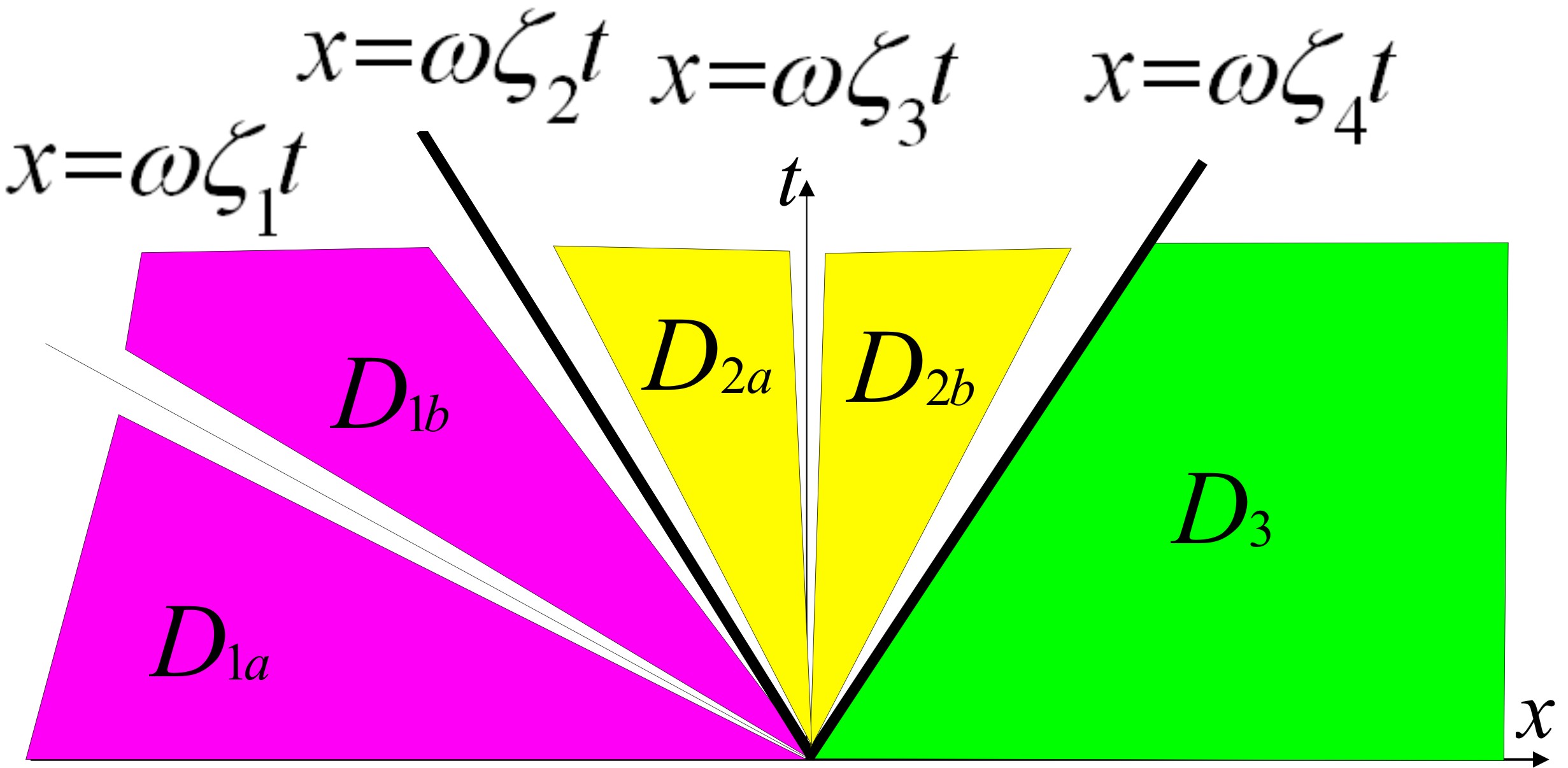}
\caption{Case $1<\dsfrac{c}{\omega}<3$. Regions in $x,t$-halfplane with qualitatively different asymptotics.}
\label{Regions_xt_13_5}
\end{figure}

\begin{figure}
\vskip-0cm\hskip-.cm\includegraphics
[scale=0.6]
%[width=1.6\textwidth, natwidth=1210, natheight=1642]
{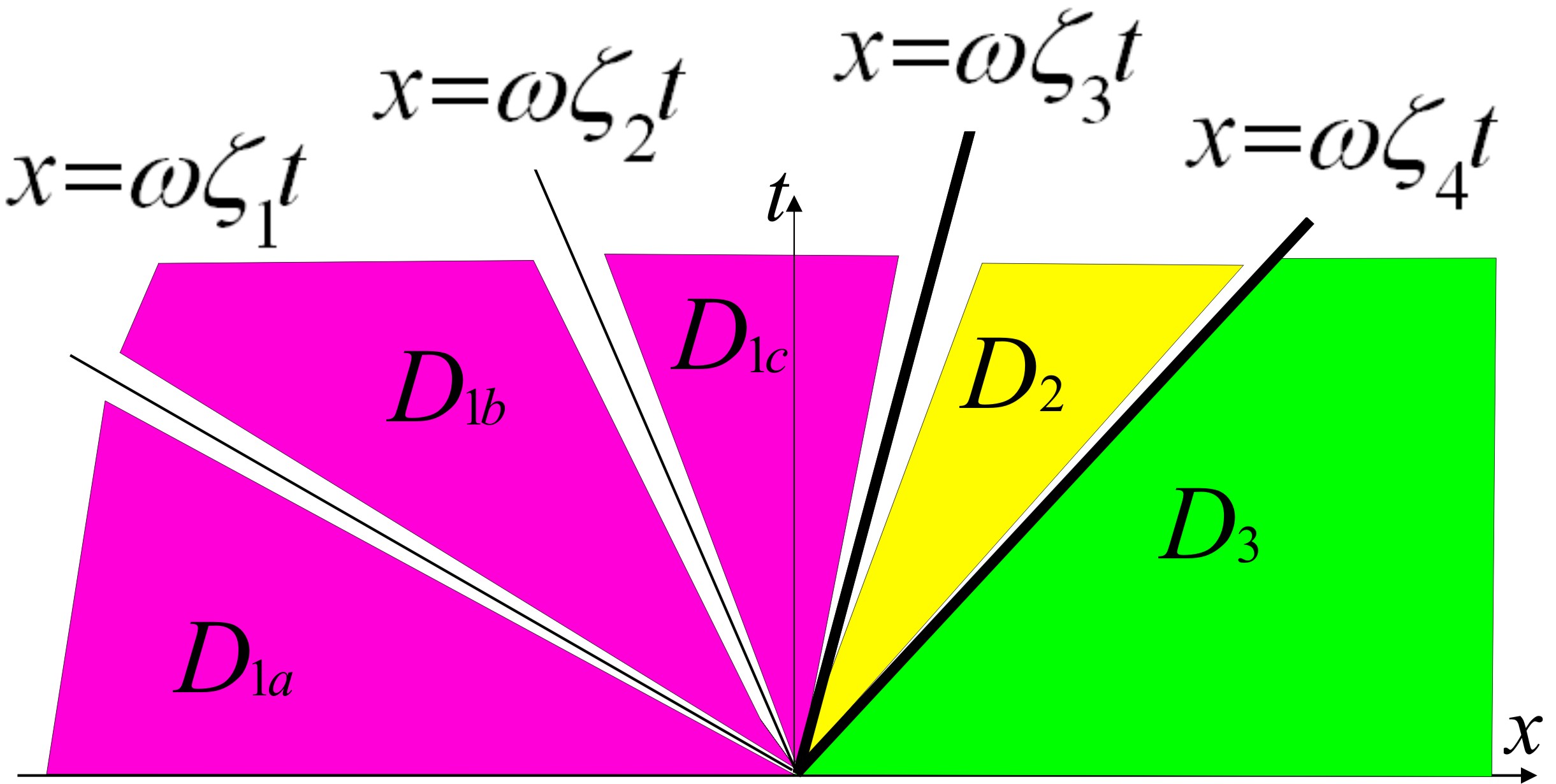}
\caption{Case $0<\dsfrac{c}{\omega}<1$. Regions in $x,t$-halfplane with qualitatively different asymptotics.}
\label{Regions_xt_01_5}
\end{figure}

The paper is organized as follows: in Section \ref{sect: Preliminaries} we list some facts about the Camassa-Holm equation 
and formulate the vector Riemann-Hilbert problem.
In Section \ref{sect_phase function} we describe the corresponding $g$-function approach and the suitable phase functions 
that we use in asymptotic analysis in different regions of the $x,t\geq0$-half-plane. In Sections 
\ref{sect_RH problem transformations}, \ref{sect: Asymptotics} we give a chain of Riemann-Hilbert problem 
transformations, which leads to model problems explicitly solvable in terms of elliptic (genus 1) or 
hyperelliptic (genus 2) functions. This in turn gives us explicit expressions for the asymptotics.

\textbf{Acknowledgements.} 
We thank Iryna Egorova, Dmitry Shepelsky, Robert Buckingham and Svetlana Roudenko 
for useful discussions.

The research has been supported by the project "Support of inter-sectoral mobility and quality enhancement 
of research teams at Czech Technical University in Prague", CZ.1.07/2.3.00/30.0034, sponsored by European Social 
Fund in the Czech Republic."

\section{Preliminaries}\label{sect: Preliminaries}
\subsection{Lax pair, Jost solutions, spectral functions}\label{SubSect: Lax pair, Jost solutions, spectral functions} We begin by recalling
some important results for the Camassa -- Holm equation
\cite{Shep2007}, \cite{Kostenko_Shepelsky_Teschl_2009}, \cite{Constantin_2001}, \cite{Constantin_Gerdikov_Ivanov},
 \cite{M15}.

The starting point for our considerations is the Lax pair
representation: the CH equation is the compatibility condition of
two linear equations
% -------------- %
\begin{subequations}\label{Lax_representation_CH}
\begin{eqnarray}\label{Shturm_x-eq}
-\,\varphi''_{xx}+\dsfrac{1}{4}\ \varphi=\lambda\,
\frac{m+\omega}{\omega}\ \varphi\ ,
\\\label{Shturm_t-eq}
\varphi_t=-\(\dsfrac{\omega}{2\lambda}+u\)\varphi'_x+\frac{u_x}{2}\varphi\
,
\end{eqnarray}
\end{subequations}
% ------------- %
%
\begin{equation}\label{lambda_k_z}\textstyle\hskip-10.mm\textrm{where}
\qquad\qquad\qquad
\lambda=:k^2+\frac{1}{4}=:z^2+\frac{\omega}{4(c+\omega)}
\end{equation}
are the spectral parameters. This means that 
%the mixed derivatives coincide 
$\varphi_{xxt}=\varphi_{txx}$ iff $u$ satisfies the CH equation.

The $x-$ equation is closely related to the spectral problem for the Schrodin\-ger operator with 
a step-like potential, which has been studied in \textit{Buslaev, Fomin} \cite{Buslaev_Fomin},
\textit{Cohen, Kappeler} \cite{Cohen_Kappeler}.

% --------------------
% --------------------
% Таким
%образом, при определенных условиях на функцию $u(x,t)$,
%спек\-траль\-ная задача (\ref{x-eq_Phi}) сводится к спектральной
%задаче, разобранной в \cite{Cohen_Kappeler}. To apply results from
%\cite{Cohen_Kappeler} we need to change there $x$ with $-y$, $k$
%with $\sqrt{k^2+\frac{c}{4(c+\omega)}},$ and $v$ with
%$v-\frac{c}{4(c+\omega)}.$
%
%\vskip2mm\noindent Denote\vskip-16mm
\begin{equation}\label{chat}\hskip-10.5mm\textrm{Denote}
\qquad\qquad\qquad
\c=\frac{1}{2}\sqrt{\frac{c}{c+\omega}}\in\(0,\frac12\).
\end{equation}

%In the following preliminary lemmas we consider 
%\begin{equation}\label{z in k}z=z(k)=\sqrt{k^2+\c^2}\end{equation} as an analytic function of $k$ on the plane cut along the segment $\(\i\c,-\i\c\)$, with the square root fixed by the condition $z\sim
%k\ $ as $\ k\rightarrow\infty.$
% To make the writing more transparent, we use the following notations for subsets of the $k$-plane cut along the segment $\(\i\c,-\i\c\)$:
%% -------------- %
%\begin{subequations}\label{Domain D bar D partial bar D}\begin{eqnarray}\label{Domain D}D:=&&\left\{k: \Im
%k>0\right\}\setminus[0,\i\c],\\ \label{Domain bar D}
%\overline{D}:=&&\(\left\{k: \Im k\geq 0\right\}\setminus\left[0,\i\c\right]\)\cup
%\left[0,\i\c\right]_-\cup \left[0,\i\c\right]_+;\nonumber
%\\
%\partial\overline{D}:=&&\R\cup\left[\i\c,0\right]_-\cup\left[\i\c,0\right]_+;
%\end{eqnarray}
%\end{subequations}
%% ---------- %
%the last two formulas stress that the segment $[0,\i\c]$ has two banks: the left and the right ones.
%% ------------ %
%Let us notice, that $D$, $\overline{D},$ and $\partial\overline{D}$ can be characterized in terms of $z(k)$ (\ref{z in k}) as the domains where $\Im z(k)>0$, $\Im z(k)\geq 0,$ and both  $\Im z(k)=0$ and $\Im k=0$, respectively.

\noindent 
The next Lemma \ref{lem_m+omega} (cf. \cite[Lemma 2.1]{M15}) distinguishes an important class of solutions of the
Camassa-Holm equation (see Assumption 1 in Section \ref{sect: Introduction}).

\

\begin{lem}{\label{lem_m+omega}} (cf.  \cite{Constantin_2001}, \cite[Lemma 2.1]{M15}).\textit{
\\Suppose there exists a global classical solution to the problem
(\ref{CH}), (\ref{init_cond}). Given Assumption 1, we have 
\begin{equation}\label{m+omega_all_time}m(x,t)+\omega>0\end{equation} for all values of time $t\geq0.$}
\end{lem}

\

% --------------------- %
This lemma justifies consideration of solutions of the Camassa-Holm equation ($\ref{CH}$) satisfying
$m(x,t)+\omega>0$ for all $x$ and $t$, so from now on, we consider only solutions 
satisfying (\ref{m+omega_all_time}). Then (\ref{CH}) can be written in a form of a local conservation law

%\begin{equation}\label{CH alternative}
\centerline{$\(\sqrt{\frac{m+\omega}{\omega}}\)_t=-\(u\sqrt{\frac{m+\omega}{\omega}}\)_x\
.$}
%\end{equation}
% --------------- %

\noindent There are indeed infinite number of conservation laws, in the sequel we need two of them:
\\
$H_{-1}=x\(\sqrt{\dsfrac{c+\omega}{\omega}}-1\)-c\,\sqrt{\dsfrac{c+\omega}{\omega}\ }\
t+\int\limits_{-\infty}^x\(\sqrt{\dsfrac{m(\xi,t)+\omega}{\omega}}-\sqrt{\dsfrac{c+\omega}{\omega}}\)\d\xi+$
% -------------- 
\begin{equation}\label{wp}\hfill+\int\limits_x^{+\infty}\(\sqrt{\dsfrac{m(\xi,t)+\omega}{\omega}}-1\)\d\xi
\end{equation}
% --------------
\begin{equation}\label{H_0[u]}
H_0=\int_{-\infty}^x \(c-m(r,t)\)\d
r-\int\limits_x^{+\infty}m(r,t)\d
r-c(x+1)+\frac{c(3c+4\omega)t}{2}.
\end{equation}
% ---------- %

\noindent An important role is played by the quantity
% --------------- %

\begin{equation}\label{y_x_-}
y\hskip-.001mm:=\hskip-.001mmx-\hskip-.001mm\int\limits_x^{+\infty}\hskip-.002mm\(\hskip-.001mm
\sqrt{\frac{m+\omega}{\omega}\
}-\hskip-.001mm1\hskip-.001mm\)\hskip-.001mm\d
r\hskip-.001mm
\end{equation}
$\hskip15mm=-H_{-1}+\sqrt{\hskip-.005mm\frac{c+\omega}{\omega}}\left[\hskip-.001mmx\hskip-0.001mm-\hskip-.001mmct
\hskip-.001mm+\hskip-.002mm
\int\limits_{-\infty}^x\hskip-.002mm\(\hskip-1mm\sqrt{\frac{m+\omega}{c+\omega}}-\hskip-.001mm1\hskip-.001mm\)
\hskip-.001mm\d
r\right]\hskip-0.001mm.
$

\noindent	
For further use let us notice that $y\to\pm\infty$ as $x\to\pm\infty$ and vice versa.

\begin{lem}\label{lem_Jost_solutions} (cf. \cite[Lemma 2.2, 2.3]{M15})
\textit{
Assumption 4 ensures that there exist two Jost solutions $\varphi_{\l}(x,t;k)$, $\varphi_{\r}(x,t;k)$ that solve
both the $x$- and $t$-equations (\ref{Lax_representation_CH}a), (\ref{Lax_representation_CH}b), and satisfy}
\begin{equation}\label{asymptotics varphi - x}
\lim\limits_{x\rightarrow-\infty}\sqrt[4]{\frac{c+\omega}{\omega}}\exp\left\{\i
z\sqrt{\frac{c+\omega}{\omega}}\(x-\(c+\frac{\omega}{2\lambda}\)t\)\right\}\varphi_{\l}(x,t;k)=1,
\end{equation}
\begin{equation}\label{asymptotics varphi + x}
\lim\limits_{x\rightarrow+\infty}\e^{-\i kx+\frac{\i\omega
kt}{2\lambda}}\varphi_{\r}(x,t;k)=1.
\end{equation}

\noindent\textit{The function $\ \exp\left\{\frac{-\i zt\sqrt{(c+\omega)\omega\
}}{2\lambda}\right\}\varphi_{\l}(x,t;k)\ $ is analytic in 
\begin{equation}\label{Domain D bar D partial bar D}\hskip16mmD:=\left\{k: \Im k>0\right\}\setminus[0,\i\c]
\end{equation}
and continuous in (subscripts $``\pm"$ denote the right/left side of the bank)
\begin{equation}\label{Domain bar D}\overline{D}:=\(\left\{k: \Im k\geq 0\right\}\setminus\left[0,\i\c\right]\)
\cup
\left[0,\i\c\right]_-\cup \left[0,\i\c\right]_+;\end{equation}
The last expression means that we distinguish between the left and right banks of the segment $[0,\i\c].$
Here, \\$\bullet\ \varphi_{\l}(x,t;.)$ is  discontinuous across $[0,\i\c],$
\\$\bullet$
the function $\ \exp\left\{\frac{\i
kt\omega}{2\lambda}\right\}\varphi_{\r}(x,t;k)\ $ is analytic for
$\Im k>0$ and continuous for $\Im k\geq0.$ As
$k\rightarrow\infty$, $\Im k\geq0,$
we have}
\begin{eqnarray}\nonumber\bullet\ \varphi_{\l}(x,t;k)&=&\sqrt[4]{\dsfrac{\omega}{m+\omega}}\ \cdot
\\
\nonumber&\cdot&\exp\left\{-\i z
\sqrt{\frac{c+\omega}{\omega}}\(x+\int\limits_{-\infty}^x\(\sqrt{\frac{m+\omega}{c+\omega}}-1\)\d\tilde
x-\(c+\frac{\omega}{2\lambda}\)t\)\right\}\cdot
\\
\label{asymptotics varphi - k}&\cdot&\(1-\frac{1}{2\i k}\int\limits_{-\infty}^y\(v(\tilde y,t)+\frac{c}{4(c+\omega)}\)
\d\tilde y+\mathrm{O}\(k^{-2}\)\),
\end{eqnarray}
% --------------------- %
\begin{eqnarray}\nonumber
\bullet\ \varphi_{\r}(x,t;k)&=&\sqrt[4]{\frac{\omega}{m+\omega}} \exp\left\{\i
k\(x-\int\limits_{x}^{+\infty}\(\sqrt{\frac{m+\omega}{\omega}}-1\)\d\tilde
x\)-\frac{\i\omega kt}{2\lambda}\right\}\cdot 
\\
\label{asymptotics varphi + k}&\cdot&
\(1-\frac{1}{2\i k}\int\limits^{+\infty}_y v(\tilde y,t) \d\tilde
y+\mathrm{O}\(k^{-2}\)\).
\end{eqnarray}

\noindent\textit{ Moreover, the following relations are satisfied for $k\in\(\i\c ,0\)$:\\
$\varphi_{\l,\r}(x,t;k)=\overline{\varphi_{\l,\r}(x,t;-\overline{k})},\
\varphi_{\l}(x,t;k\pm0)=\overline{\varphi_{\l}(x,t;\overline{k}\mp0)},$\\
$\varphi_{\r}(x,t;k)=\overline{\varphi_{\r}(x,t;k)},\
\varphi_{\l}(x,t;k\pm0)=\varphi_{\l}(x,t;\overline{k}\pm0).$
}
\end{lem}

\

\noindent
Let us notice, that $D$, $\overline{D},$ and 
\begin{equation}\label{Domain partial bar D}\partial\overline{D}:=\R\cup\left[0,\i\c\right]_-\cup\left[0,\i\c\right]_+
\end{equation}
 can be characterized in terms of $z(k)$ (\ref{lambda_k_z}) as the domains where $\Im z(k)>0$, $\Im z(k)\geq 0,$ and both  $\Im z(k)=0$ and $\Im k\ge0$, respectively.

% ----------- %
%\vskip 1cm
%\noindent
\noindent
The Jost solutions determine spectral functions via the  scattering relation
\hskip-2cm\begin{subequations}\label{scattering relations varphi}
\hskip-2cm\begin{eqnarray}
\varphi_{\l}(x,t;k)&\hskip-3mm=\hskip-3mm&a_{}(k)\
\overline{\varphi_{\r}(x,t;\overline{k})}+b_{}(k)\
\varphi_{\r}(x,t;k),\ k\in\mathbb{R}\hskip-1mm\setminus\hskip-1mm\left\{k=0\right\},
%\\
%\varphi_{\r}(x,t;k)&\hskip-3mm=\hskip-3mm&a_{\l}(k)\
%\overline{\varphi_{\l}(x,t;\overline{k})}+b_{\l}(k)\
%\varphi_{\l}(x,t;k),\quad k\in\partial\overline{D},
\end{eqnarray}
\end{subequations}
%where \vskip-10.25mm
%$$\partial\overline{D}:=\R\cup\left[\i\c,0\right]_-\cup\left[\i\c,0\right]_+;$$
% ----------
which determines the (right) transmission coefficient $a_{}^{-1}(k),$ and the (right) reflection coefficient $r(k):=\frac{b_{}(k)}{a_{}(k)}.$ 

%Introduce
%also the Wronskian
%\begin{equation}\label{Wronskian definition}W(k):=W\left\{\varphi_{\l}(x,t;k),\varphi_{\r}(x,t;k)\right\}=\varphi_{\l}\cdot\frac{\partial\varphi_{\r}}{\partial
%x}-\varphi_{\r}\cdot\frac{\partial\varphi_{\l}}{\partial x}.
%\end{equation} 

\noindent
In the following preliminary lemma we use the notations (\ref{Domain D bar D partial bar D}), (\ref{Domain bar D}), 
(\ref{Domain partial bar D}) for the domains $D$, $\overline{D}$, and $\partial\overline{D}$.

\

\begin{lem}\label{Lemma_properties_a,b,r}
\noindent\textbf{(cf. \cite[Lemmas 2.4-2.7]{M15})}\textit{
%Suppose that $\forall x\in\mathbb{R}: \quad c>m(x,0).$
\\
$\bullet$ The (right) transmission coefficient $a^{-1}_{}(k)$ is meromorphic in $D$; has a finite number of poles $ \i\kappa_1,...,\i\kappa_N,$ which lie in the interval $\c<\kappa_N<...<\kappa_1<\frac{1}{2}$; function $a^{-1}_{}(k)$ is continuous up to the boundary of its domain of analyticity with (in general) the exception of the edge point $k=\i\c$. Moreover, $a^{-1}_{\pm}(0)=0.$ 
\\$\bullet$ Further, under Assumption 2}
(
$\textstyle\forall x\in\mathbb{R}: \quad c\geq m_0(x)
$),
\textit{the (right) transmision coefficient $a^{-1}(k)$ is
continuous up to the point $k=\i\c;$
\\
$\bullet$ Under Assumption 3, the reflection coefficient
$r(k)\equiv\frac{b(k)}{a(k)}$ is meromorphic in $D_{C_0}:=\left\{k:0<\Im k<\sqrt{C_0}\right\}\setminus\left\{\i/2\right\}$ with simple poles at $\i\varkappa_j$ (those which lie in the domain $D_{C_0}$), and takes continuous boundary values at $(\i\c,0]$ 
(recall that $C_0$ is the constant 
that appears in (\ref{assum3_integr_ineq}), Assumption 3).
Further, under Assumption 2, $r(k)$ is continuous up to the point $k=\i\c;$ 
}

\noindent
$\bullet$ \textit{Asymptotically as $k\rightarrow\infty$ we have
\begin{equation}\label{transmission coeff+- asymptotics}
a^{-1}_{}(k)=\e^{\i H_{-1} k}\(1+\mathrm{O}(k^{-1})\), \
r(k)\equiv\frac{b(k)}{a(k)} = \mathrm{O}(k^{-1}).
\end{equation}
\\
$\bullet$ 
The residues of $a_+^{-1}(k)$ are given by
\begin{equation}\label{transmission coeff+-residues}
\mathrm{Res}_{\, \i\kappa_j}\ a^{-1}_+(k)=\i\mu_j\gamma_{+,j}^2,\
\end{equation}
where $\gamma_{\r,j}^{-2}:=\int\limits_{-\infty}^{+\infty}
\(\varphi_{}(x,t;\i\kappa_j)\)^2\frac{m+\omega}{\omega}\d x>0
$
 and
$\varphi_\r(x,t;\i\kappa_j)=\mu_j\varphi_\l(x,t;\i\kappa_j)$ with
quantities $\gamma_{\r,j}$, $\mu_j$ independent on $t.$
}

\noindent
$\bullet$ 
$\overline{a^{-1}(-\overline{k})}=a^{-1}(k),\, k\in\overline{D};$
$\quad \overline{r_{}(-\overline{k})}=r_{}(k),\
k\in\mathbb{D_{C_0}};$ 
% ---------
\\\\$1-|r_{}(k)|^2=\dsfrac{z(k)}{k}\dsfrac{1}{|a(k)|^2},\quad
k\in\mathbb{R}\setminus\left\{0\right\}.$
\end{lem}

% ------------------- %
% ------------------- %

%--------------------------------------------------------------------------------------------------------------------------------

\subsection{Vector Riemann -- Hilbert problem.}\label{SubSect_RH_problem_1} \noindent 
%Suppose that the condition  $\frac{c-m_0(x)}{\omega}\geq0$ from \cite[Lemma 2.6]{M15}
 %is satisfied. 
We define the vector
Riemann -- Hilbert problem as follows. The sectionally meromorphic
function
$V_{\mathfrak{r}}(y,t;k)=\(V_{\mathfrak{r},\,1}(y,t;k),V_{\mathfrak{r},\,2}(y,t;k)\)$
is defined by
% ------------- 
\begin{equation}\label{M_psi}
\hskip-3mm\left\{\begin{array}{l}
\sqrt[4]{\frac{m+\omega}{\omega}}\(\dsfrac{1}{a(k)}\varphi_\l(x,t;k)\,\e^{\i
g_{\r}(y,t;k)}\ , \ \varphi_\r(x,t;k)\,\e^{-\i g_{\r}(y,t;k)}\), \quad k\in D,
\\
\sqrt[4]{\frac{m+\omega}{\omega}}\(\overline{\varphi_{}(x,t;\overline{k})}\,\e^{\i
g_{\r}(y,t;k)}\ , \
\dsfrac{1}{\overline{a_{}(\overline{k})}}\overline{\varphi_-(x,t;\overline{k})}\,\e^{-\i
g_{\r}(y,t;k)}\), \quad k\in D^*,
\end{array}\right.
\end{equation}
\noindent where $g_{\r}(y,t;k)=k y-\frac{2\omega k t}{4k^2+1}$,  and the domain $D$ is determined by (\ref{Domain D bar D partial bar D}).

We are interested in the jump relations of
$V_{\mathfrak{r}}(y,t;k)$ on the contour
$$\Sigma_{\mathfrak{r}}=\mathbb{R}\cup
\left[\i\c,-\i\c\right].$$ The orientation of the
contour is chosen as follows: from $-\infty$ to $+\infty$ and from
$+\i\c$ to
$-\i\c.$ The positive side of the contour
lies on the left as one moves along the contour in the positive direction, the negative one lies on the right. 
We will use the notation $V_{\pm}(y,t;k)$
for the limit of $V_{\pm}(y,t;k)$ from the positive/negative side of the
contour.

%\begin{figure}
%\vskip-4cm\hskip-2.5cm\includegraphics
%[scale=0.95]
%%[width=1.6\textwidth, natwidth=1210, natheight=1642]
%{Graphics_Sigma_r.eps}
%\centerline{Contour $\Sigma_{\r}$}
%\label{Figure_phase_middle_function_plot}
%\end{figure}

\begin{figure}[ht!]
\begin{center}
\epsfig{width=60mm,figure=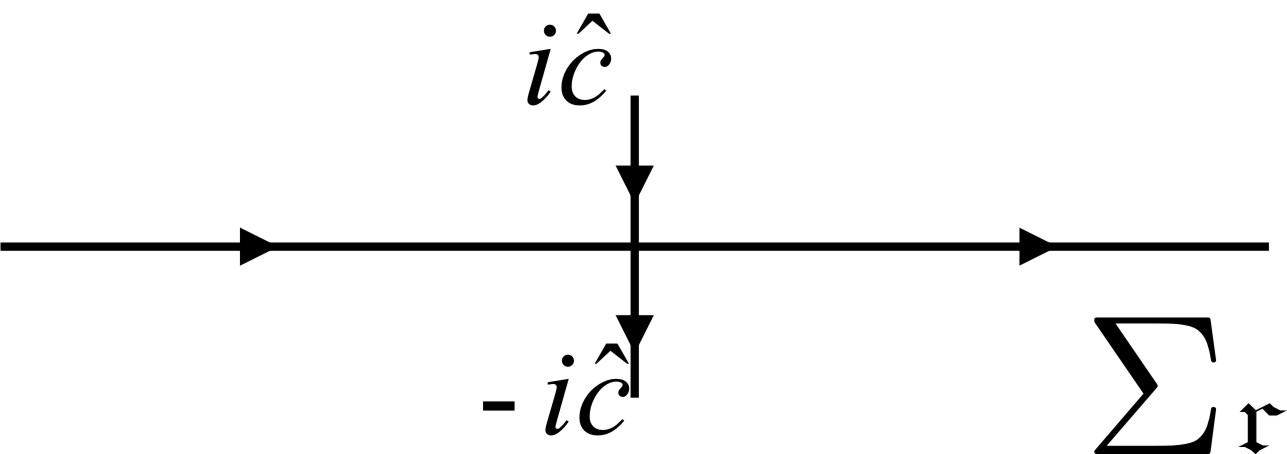}
\end{center}
%\caption{Signature table for $\Im\theta(k,\xi)$ at
%$\dsfrac{-1}{4}<\xi<0.$} \label{Signature table right 4}
% ------------ %
\caption{Contour $\Sigma_{\r}$.}
\label{Graphics_Sigma_r}
\end{figure}
% ---------- %

The scattering relation (\ref{scattering relations varphi}) and
Lemma \ref{Lemma_properties_a,b,r} provide the following
properties of the function $V_{\mathfrak{r}}(y,t;k):$
\begin{lem}\label{lem_RH}(cf. \cite[Section 3]{M15})
\begin{enumerate}
\item Analyticity: \\$V_{\mathfrak{r}}(y,t;k)$ is meromorphic
in $k\in\mathbb{C}\setminus\Sigma_{\mathfrak{r}}$ with simple poles at
$\pm\i\kappa_j$ and continuous up to the boundary;

\item Jump relations: for $k\in\Sigma_{\mathfrak{r}}$\\
$V_{\mathfrak{r}}^-(y,t;k)=V_{\mathfrak{r}}^+(y,t;k)J_{\mathfrak{r}}(y,t;k),$
where
\begin{equation}\label{J_R}
J_{\mathfrak{r}}(y,t;k)=\begin{pmatrix}1&\ol{r_{}\(\ol{k}\)}\,\,\e^{-2\i
g_{\r}(y,t;k)}\\-r_{}(k)\,\e^{2\i
g_{\r}(y,t;k)}&1-|r_{}(k)|^2\end{pmatrix},\quad
k\in\mathbb{R}\setminus\{0\},
\end{equation}
% --------------- %
\begin{equation}\nonumber
J_{\mathfrak{r}}(y,t;k)=\begin{pmatrix}1&0\\f(k)\,\e^{2\i g_{\r}(y,t;k)}&1\end{pmatrix},\quad
k\in\(\i\c,0\),
\end{equation}
% --------------- %
\begin{equation}\label{J_0_-ic}
J_{\mathfrak{r}}(y,t;k)=\begin{pmatrix}1&\ol{f\(\ol{k}\)}\,\e^{-2\i
g_{\r}(y,t;k)}\\0&1\end{pmatrix},\quad
k\in\(0,-\i\c\),
\end{equation}
% --------
where $r_{}(k)=\frac{b(k)}{a(k)},$
\ $f(k):=\frac{z(k+0)}{k\,
a(k-0)a(k+0)},\quad k\in(\i\c,0),$
% ------
\\
$1-|r_{}(k)|^2=\frac{z(k)}{k\,|a(k)|^2},\quad k\in{\R};$
% ------------------- %
\item Residue conditions: for $j=1,...,N$
\begin{subequations}\label{prop_symmetry_a}
\begin{eqnarray}
\mathrm{Res}_{\i\kappa_j}V_{\mathfrak{r}}(y,t;k)&\hskip-3mm=\hskip-3mm&\lim\limits_{k\rightarrow\i\kappa_j}V_{\mathfrak{r}}(y,t;k)\begin{pmatrix}0&0\\\i\gamma_{+,j}^2\e^{2\i
g_{\r}(y,t;\i\kappa_j)}&0\end{pmatrix},\qquad\ \ 
\\
\mathrm{Res}_{-\i\kappa_j}V_{\mathfrak{r}}(y,t;k)&\hskip-3mm=\hskip-3mm&\lim\limits_{k\rightarrow-\i\kappa_j}V_{\mathfrak{r}}(y,t;k)\begin{pmatrix}0&-\i\gamma_{+,j}^2\e^{2\i
g_{\r}(y,t;\i\kappa_j)}\\0&0\end{pmatrix}\hskip-1mm;
\end{eqnarray}
\end{subequations}
% --------------- %
\item Symmetry relations:
\begin{equation}\nonumber
\overline{V_{\mathfrak{r}}(y,t;\overline{k})}=V_{\mathfrak{r}}(y,t;-k)=V_{\mathfrak{r}}(y,t;k)\begin{pmatrix}0&1\\1&0\end{pmatrix},
\quad\overline{V_{\mathfrak{r}}(y,t;-\overline{k})}=V_{\mathfrak{r}}(y,t;k);
\end{equation}
% --------------- %
 \item Asymptotics at infinity: \begin{equation}\label{M infinity}
V_{\mathfrak{r}}(y,t;k) \rightarrow
\begin{pmatrix}1&1\end{pmatrix} \textrm{ as }
k\rightarrow\infty.\end{equation}
\end{enumerate}
\end{lem}
\noindent The solution of the initial value problem (\ref{CH}), (\ref{init_cond}) can be obtained via the solution of the RH problem in a 
parametric form :

\begin{lem}\label{Lemma reconstruction of potential}
\textbf{(cf. \cite[Lemma 5.1]{M15})} 
\\
\textit{ The function $V_{\mathfrak{r}}(x,t;k)=\(V_{\mathfrak{r,1}}(x,t;k), V_{\mathfrak{r,2}}(x,t;k)\)$ defined in
(\ref{M_psi}) satisfies the
following relations:}
% ---------- 
\begin{subequations}\label{M1*M2}\label{M1/M2}
\begin{eqnarray}&&\dsfrac{V_{\mathfrak{r},\,1}(x,t;\frac{\i}{2})}{V_{\mathfrak{r},\,2}(x,t;\frac{i}{2})}=
\e^{x-y},
\\\nonumber
&&V_{\mathfrak{r},\,1}(x,t;k)V_{\mathfrak{r},\,2}(x,t;k)=\sqrt{\frac{m+\omega}{\omega}}
\(1+\frac{2\i}{\omega}u(x,t)\(k-\frac{\i}{2}\)
+O\(k-\frac{\i}{2}\)^2\),\\
\end{eqnarray}
\end{subequations}
\end{lem}

%\[\hfill k\rightarrow\frac{\i}{2},\]
%and

%\begin{equation}\label{M1/M2 tilde}\dsfrac{V_{\mathfrak{l},\,1}(x,t;\frac{\i}{2}\sqrt{\frac{\omega}{c+\omega}})}
%{
%V_{\mathfrak{l},\,2}(x,t;\frac{\i}{2}\sqrt{\frac{\omega}{c+\omega}})}=\e^{\sqrt{\frac{\omega}{c+\omega}}\(y+H_{-1}\)-x+ct}=\e^{\left\{
%\int\limits_{-\infty}^x\(\sqrt{\frac{m+\omega}{c+\omega}}-1\)\d
%r\right\}},\end{equation}
%\[V_{\mathfrak{l},\,1}(x,t;z)V_{\mathfrak{l},\,2}(x,t;z)=\]
%\begin{equation}\label{M1*M2 tilde}=\sqrt{\frac{m+\omega}{c+\omega}}
%\(1+\dsfrac{2\i
%\(u(x,t)-c\)}{\sqrt{\omega(c+\omega)}}\(z-\frac{\i}{2}\sqrt{\frac{\omega}{c+\omega}}\)
%+O\(z-\frac{\i}{2}\sqrt{\frac{\omega}{c+\omega}}\)^2\).\end{equation}
%%\[\hfill z\rightarrow\frac{\i}{2}\sqrt{\frac{\omega}{c+\omega}}.\]
%\end{lem}

\section{Phase functions.}\label{sect_phase function} In this section we construct the necessary phase functions 
and the signature table for their imaginary part. Since the RH problem is formulated in terms of $y(x)$, 
rather than in terms of $x$,
 in 
the sequel we study the problem in terms of the variables $y,t$. We study the asymptotics along the rays 
$y= Ct,$ where 
$C$ is a constant, and hence 
in asymptotic analysis we use a ``slow" variable

$\hskip5cm\xi\equiv \dsfrac{y}{\omega t},$\\
 and a fast variable $t$; hence instead of the phase function $g_{\r}(y,t;k)$ from (\ref{M_psi}) 
it is convenient to introduce a phase function which depends 
only on $\xi\equiv \frac{y}{\omega t}$ and $t$.
\subsection{The right and the left phase function.}
As was shown in \cite{M15}, in asymptotic analysis for large positive $x$ and for large negative $x$ we 
can use the ``right" and the ``left" phase functions,
% ------------------ %
\begin{eqnarray}\label{phase grgl}
g_{\r}(k,\xi)&=&k\omega\xi-\frac{2k\omega}{4k^2+1},\qquad \textrm{ where } \xi=\frac{y}{\omega t},\\
g_{\l}(k,\xi)&=&\omega\xi\sqrt{k^2+\c^2\ }-\frac{2\omega\sqrt{k^2+\c^2}}{(1-4\c^2)(4k^2+1)},
\end{eqnarray}
(here $\c$ is from (\ref{chat})), which naturally appear in asymptotics for the Jost solutions of the Lax pair.
The signature table for $\Im g_{\r}(k,\xi)$ is drawn in Figures \ref{Signature table right 123}, 
\ref{Signature table right 4567}. 
The stationary phase points $\theta'_k(k,\xi)=0$ are given as $\i\mu_{0,\r}, \i\mu_{1,\r}$, where
% ---
%\scriptsize
\\$\qquad\mu_{0,\mathfrak{r}}(\xi)=\frac{1}{2}\sqrt{\frac{\xi+1-\sqrt{1+4\xi}}{-\xi}},\qquad
\mu_{1,\mathfrak{r}}(\xi)=\frac{1}{2}\sqrt{\frac{\xi+1+\sqrt{1+4\xi}}{-\xi}}.
$
\vskip-14mm
\begin{equation}
\label{mur01}
\end{equation}
\vskip-2mm
% -------
\noindent The quantity $d_{0,\mathfrak{r}}$ is defined
by the formula
$d_{0,\mathfrak{r}} = \frac{1}{2}\sqrt{\frac{\xi - 2}{\xi}}.$

\noindent
Once the signature table for $\Im g_{\r}$ is constructed, the one for $\Im g_{\l}$ is obtained by using the relation
% ---------- %
\noindent $g_{\l}(k,\xi)=
\frac{c+\omega}{\omega} g_\r\(\hat z,\hat\xi\), $ 
where
%$z=\sqrt{k^2+\c^2}$, 
$\quad\hat z
=
\dsfrac{\sqrt{k^2+\c^2}}{\sqrt{1-4\c^2}}$, $\ \hat\xi
= \(\frac{\omega}{c+\omega}\)^{3/2}\xi.$
% ---------- %
Under the conformal mapping $\hat z\mapsto k$, the point $\hat z =
\frac{\i}{2}$ comes to the point $k = \frac{\i}{2}$
and the segment $\hat
z\in[-\frac{1}{2}\sqrt{\frac{c}{\omega}},
\frac{1}{2}\sqrt{\frac{c}{\omega}}]$ comes to the interval $k\in
[\frac{\i}{2}\sqrt{\frac{c}{c+\omega}},
-\frac{\i}{2}\sqrt{\frac{c}{c+\omega}}]=[\i\c,-\i\c].$
% ---------- 
Therefore, an important role is played by the mutual location of the quantities
$\frac{1}{2}\sqrt{\frac{c}{\omega}}$, $\frac{\sqrt{3}}{2}$ and
$\frac{1}{2},$ and thus, we distinguish the cases 
\begin{itemize}
\item $\frac c\omega>3,$ 
\item $1<\frac c\omega<3,$ 
\item $0<\frac c\omega<1.$
\end{itemize}
% ---------------------- %
The stationary phase points 
$\ \hat z_{0,\mathfrak{l}}(\hat\xi)=\frac{1}{2}\sqrt{\frac{\hat\xi+1-\sqrt{1+4\hat\xi}}{-\hat\xi}},\
\hat
z_{1,\mathfrak{l}}(\hat\xi)=\frac{1}{2}\sqrt{\frac{\hat\xi+1+\sqrt{1+4\hat\xi}}{-\hat\xi}}$\\
are equal to $\frac{1}{2}\sqrt{\frac{c}{\omega}}$ when $\hat\xi =
\frac{-2(c-\omega)}{(c+\omega)^2}.$ The corresponding quantities
in the $k-$ plane are equal to $\i\mu_{0,\l}$, $\i\mu_{1,\l}$, where
%
%$k_{0,\mathfrak{l}}(\xi)=\frac{1}{2}\sqrt{-1-\frac{\omega}{(c+\omega)\hat\xi}-\frac{\sqrt{1+4\hat\xi\ \ }\ \omega}{(c+\omega)\hat\xi}},
%\quad k_{1,\mathfrak{l}}(\xi)=\frac{1}{2}\sqrt{-1-\frac{\omega}{(c+\omega)\hat\xi}+\frac{\sqrt{1+4\hat\xi\ \ }\
%\omega}{(c+\omega)\hat\xi}}.
%$
\\$\mu_{0,\l}=\frac{1}{2}\sqrt{\frac{\xi\sqrt{1-4\c^2}+1-\sqrt{1+4\(1-4\c^2\)^{3/2}\xi}}{\xi\sqrt{1-4\c^2}}},
\mu_{1,\l}=\frac{1}{2}\sqrt{\frac{\xi\sqrt{1-4\c^2}+1+\sqrt{1+4\(1-4\c^2\)^{3/2}\xi}}{\xi\sqrt{1-4\c^2}}}.$
\vskip-14mm\begin{equation}
\label{mul01}
\end{equation}
%
% -------
The signature table for $\mbox{Im} g_{\l}(k,\xi)$ is drawn in Figures \ref{Signature table left 1234}--\ref{Signature table left(3) 5678}.

%\newpage 

\begin{figure}[ht!]
\begin{minipage}[ht!]{0.3\linewidth}
\begin{center}
\epsfig{width=40mm,figure=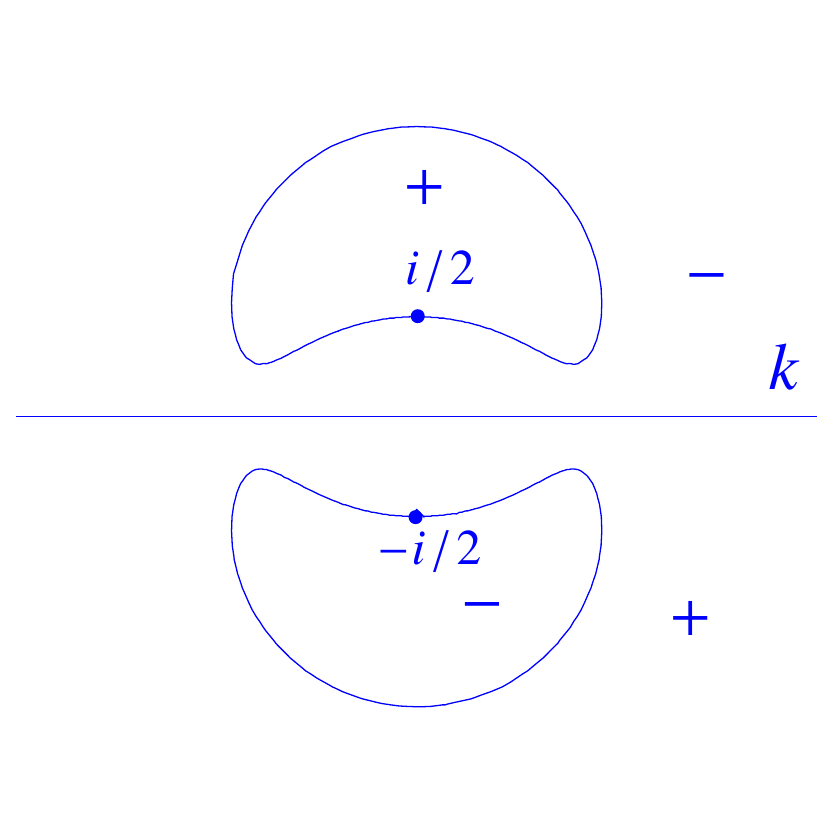}\\
$\xi<\dsfrac{-1}{4}$
\end{center}
%\caption{Signature table for $\Im\theta(k,\xi)$ at
%$\xi<\dsfrac{-1}{4}.$} \label{Signature table right 1}
\end{minipage}
% ---------- %
\begin{minipage}[ht!]{0.3\linewidth}
\begin{center}
\epsfig{width=40mm,figure=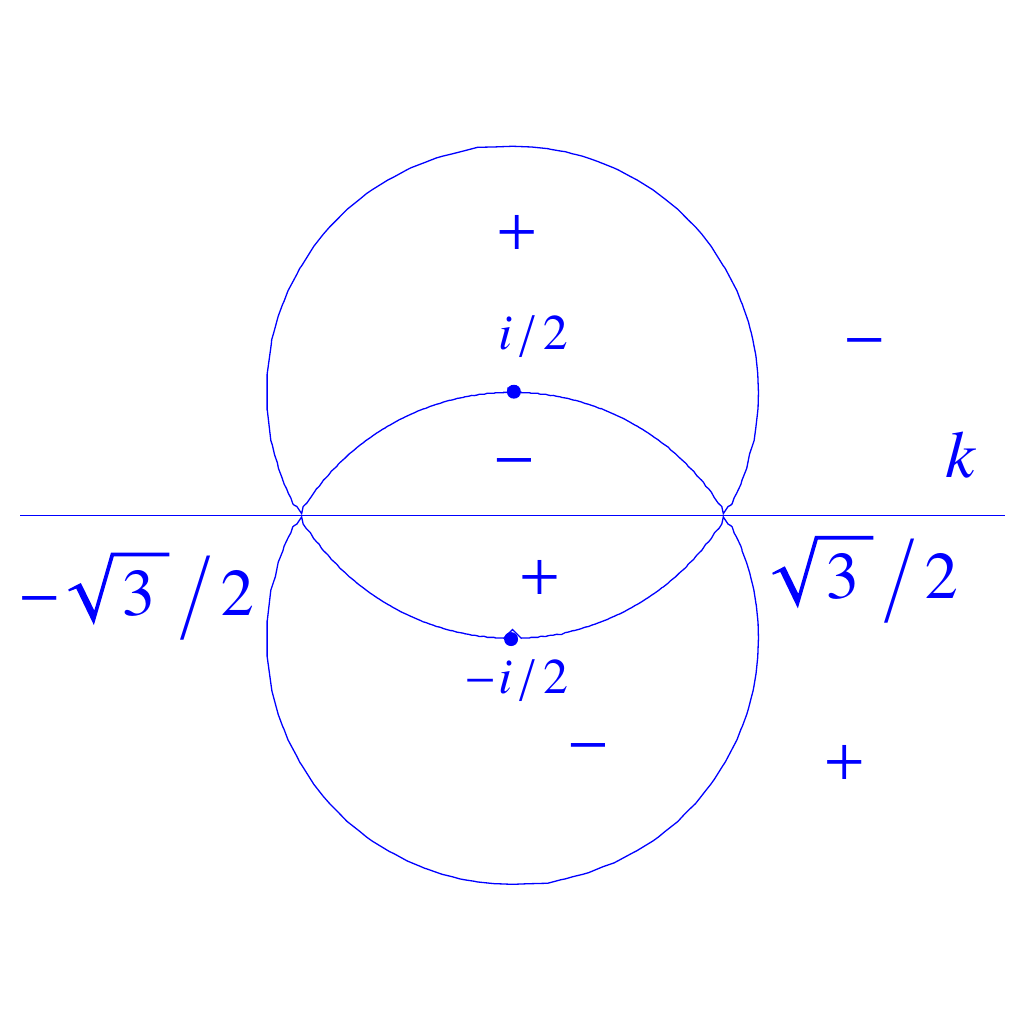}\\
$\xi=\dsfrac{-1}{4}$
\end{center}
%\caption{Signature table for $\Im\theta(k,\xi)$ at
%$\xi=\dsfrac{-1}{4}.$} \label{Signature table right 2}
\end{minipage}
\begin{minipage}[h!]{0.3\linewidth}
\begin{center}
\epsfig{width=40mm,figure=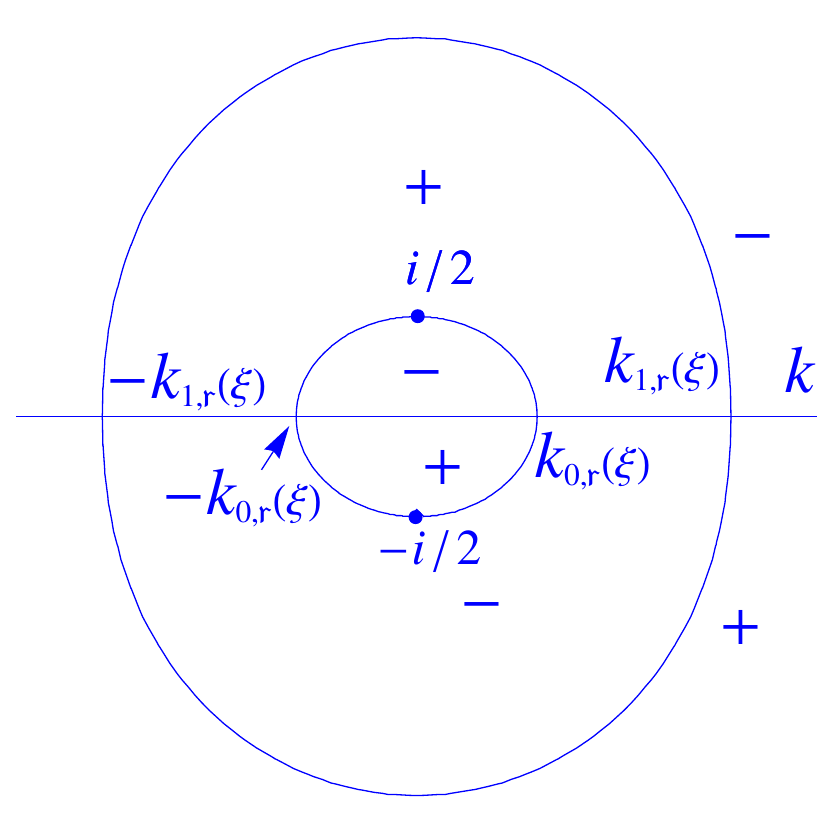}\\
$\dsfrac{-1}{4}<\xi<0$
\end{center}
%\caption{Signature table for $\Im\theta(k,\xi)$ at
%$\dsfrac{-1}{4}<\xi<0.$} \label{Signature table right 3}
\end{minipage}
\caption{Signature table for $\Im g_{\r}(k,\xi).$ The plotted contours are the lines $\Im g_{\r}=0.$}
\label{Signature table right 123}
\end{figure}
% ---------- %
% ---------- %
\begin{figure}[ht!]
\begin{minipage}[h!]{0.24\linewidth}
\begin{center}
\epsfig{width=30mm,figure=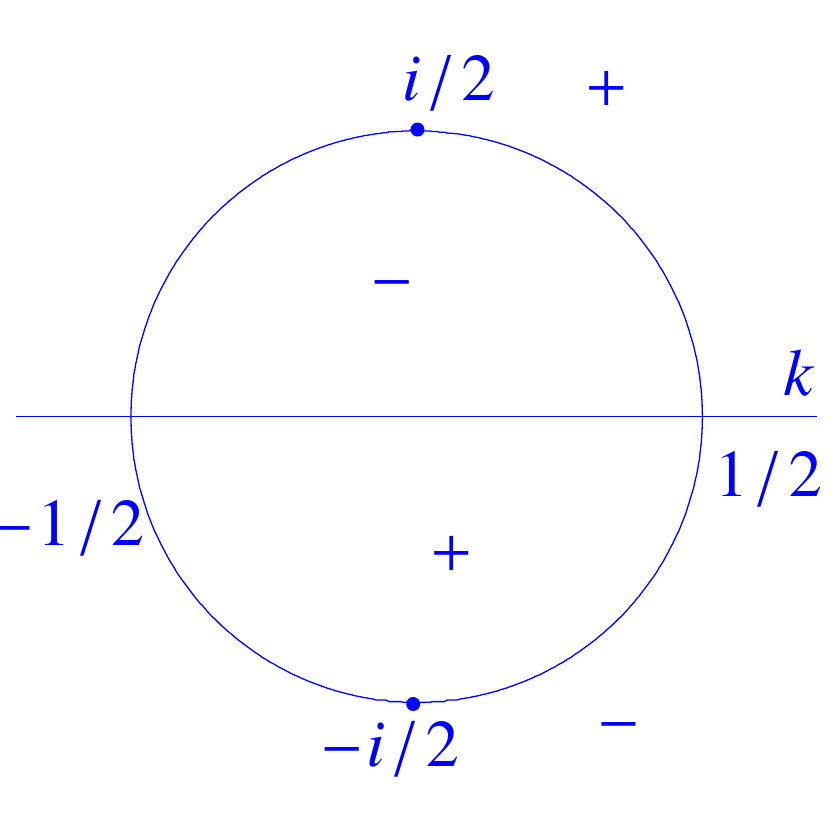}\\
$\xi=0$
\end{center}
%\caption{Signature table for $\Im\theta(k,\xi)$ at
%$\dsfrac{-1}{4}<\xi<0.$} \label{Signature table right 4}
\end{minipage}
% ------------ %
\begin{minipage}[ht!]{0.24\linewidth}
\begin{center}
\epsfig{width=30mm,figure=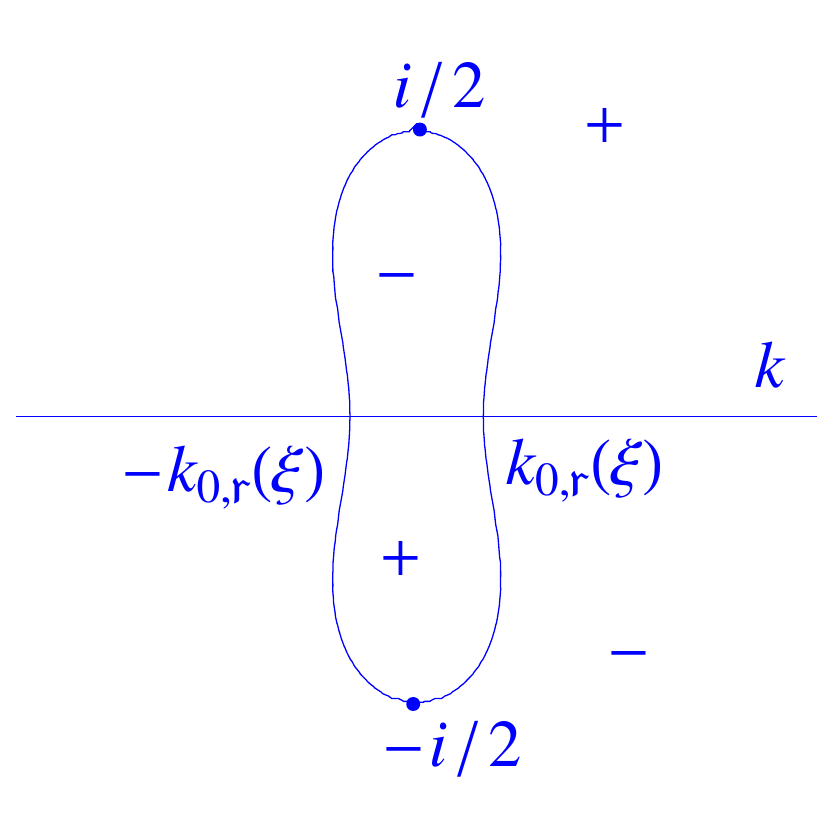}\\
$0<\xi<2$
\end{center}
%\caption{Signature table for $\Im\theta(k,\xi)$ at $\xi=0.$}
%\label{Signature table right 4}
\end{minipage}
% --------- %
\begin{minipage}[ht!]{0.24\linewidth}
\begin{center}
\epsfig{width=30mm,figure=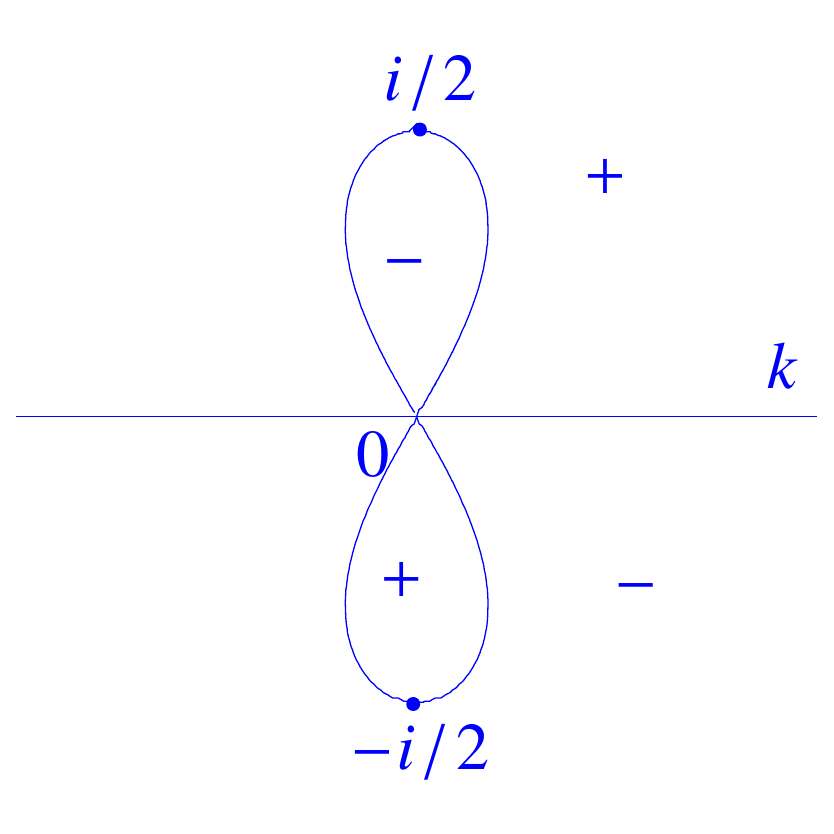}\\
$\xi=2$
\end{center}
%\caption{Signature table for $\Im g_{\r}(k,\xi)$ at $\xi=0.$}
%\label{Signature table right 4}
\end{minipage}
% ---------- %
\begin{minipage}[ht!]{0.24\linewidth}
\begin{center}
\epsfig{width=30mm,figure=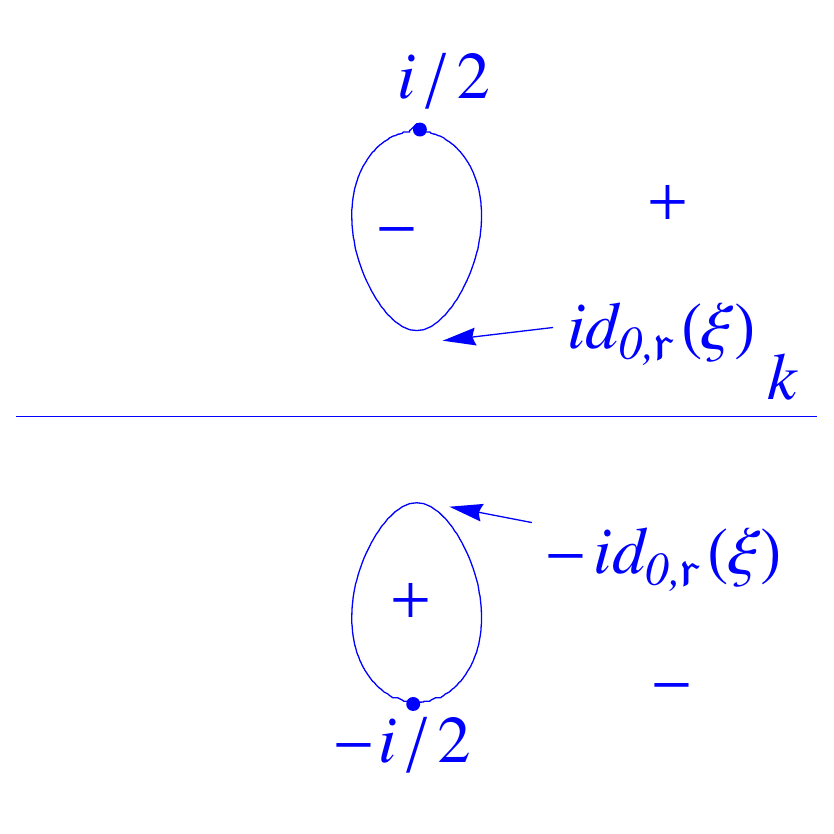}\\
$\xi>2$
\end{center}
%\caption{Signature table for $\Im g_{\r}(k,\xi)$ at $\xi=0.$}
%\label{Signature table right 4}
\end{minipage}
\caption{Signature table for $\Im g_{\r}(k,\xi).$ The plotted contours are the lines $\Im g_{\r}=0.$}
\label{Signature table right 4567}
\end{figure}
% ---------- %
\begin{figure}[ht!]
\begin{minipage}[ht!]{0.24\linewidth}
\begin{center}
\epsfig{width=30mm,figure=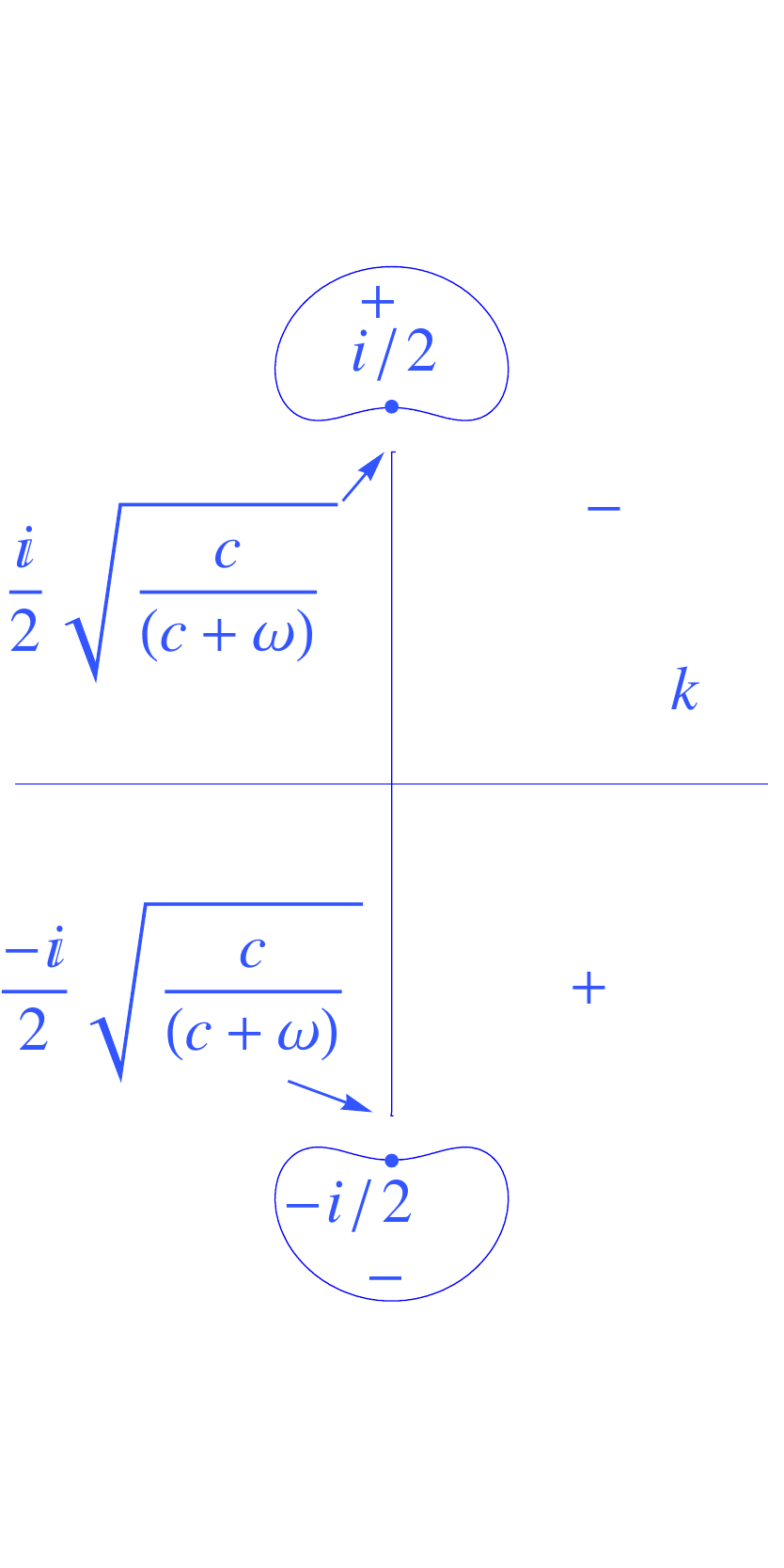}\\
$\hat\xi<\dsfrac{-1}{4}$
\end{center}
%\caption{Signature table for $\Im\theta(k,\xi)$ at $\xi=0.$}
%\label{Signature table right 4}
\end{minipage}
% --------- %
\begin{minipage}[ht!]{0.24\linewidth}
\begin{center}
\epsfig{width=30mm,figure=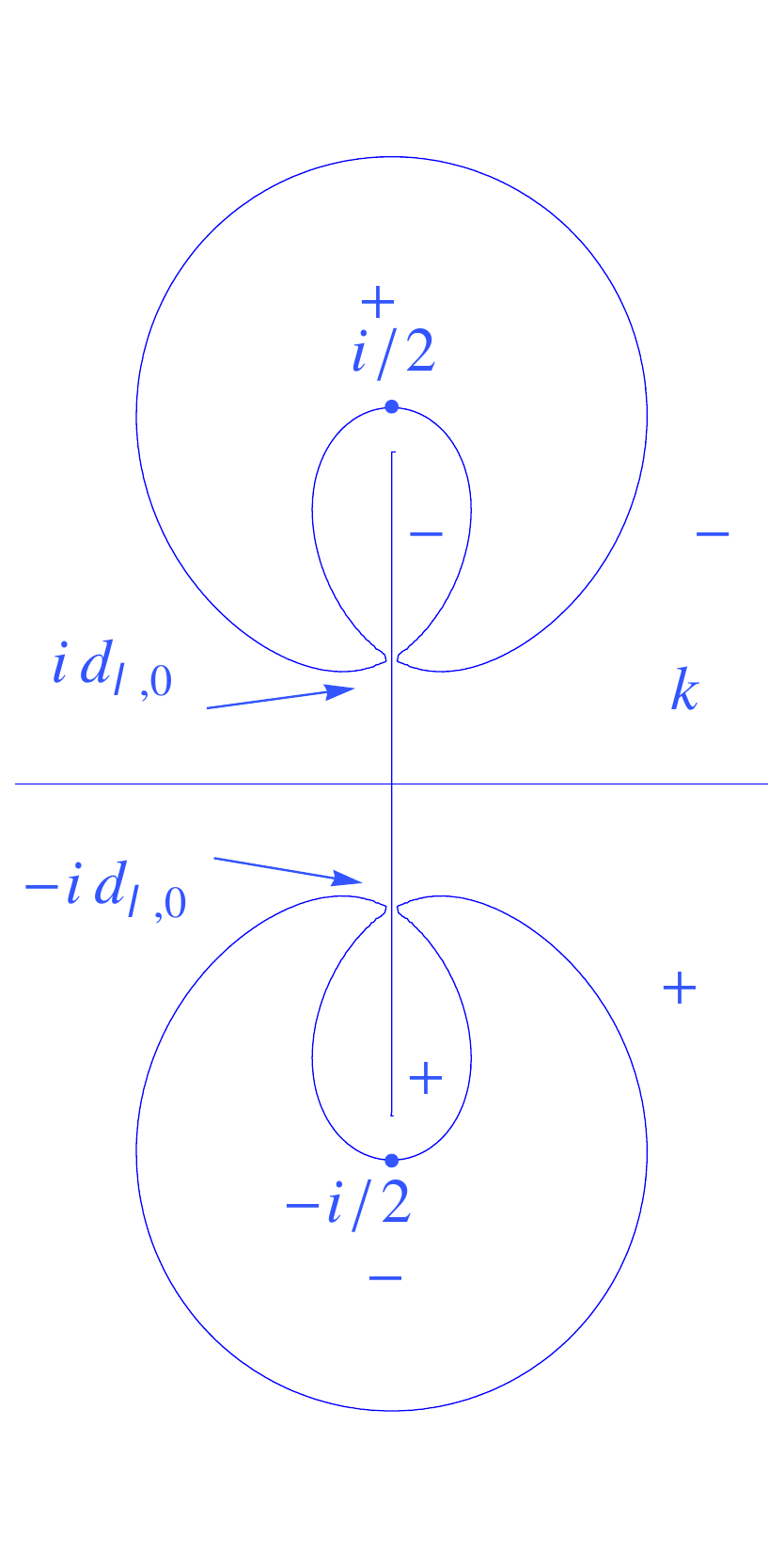}\\
$\hat\xi=\dsfrac{-1}{4}$
\end{center}
%\caption{Signature table for $\Im\theta(k,\xi)$ at $\xi=0.$}
%\label{Signature table right 4}
\end{minipage}
% ---------- %
\begin{minipage}[ht!]{0.24\linewidth}
\begin{center}
\epsfig{width=30mm,figure=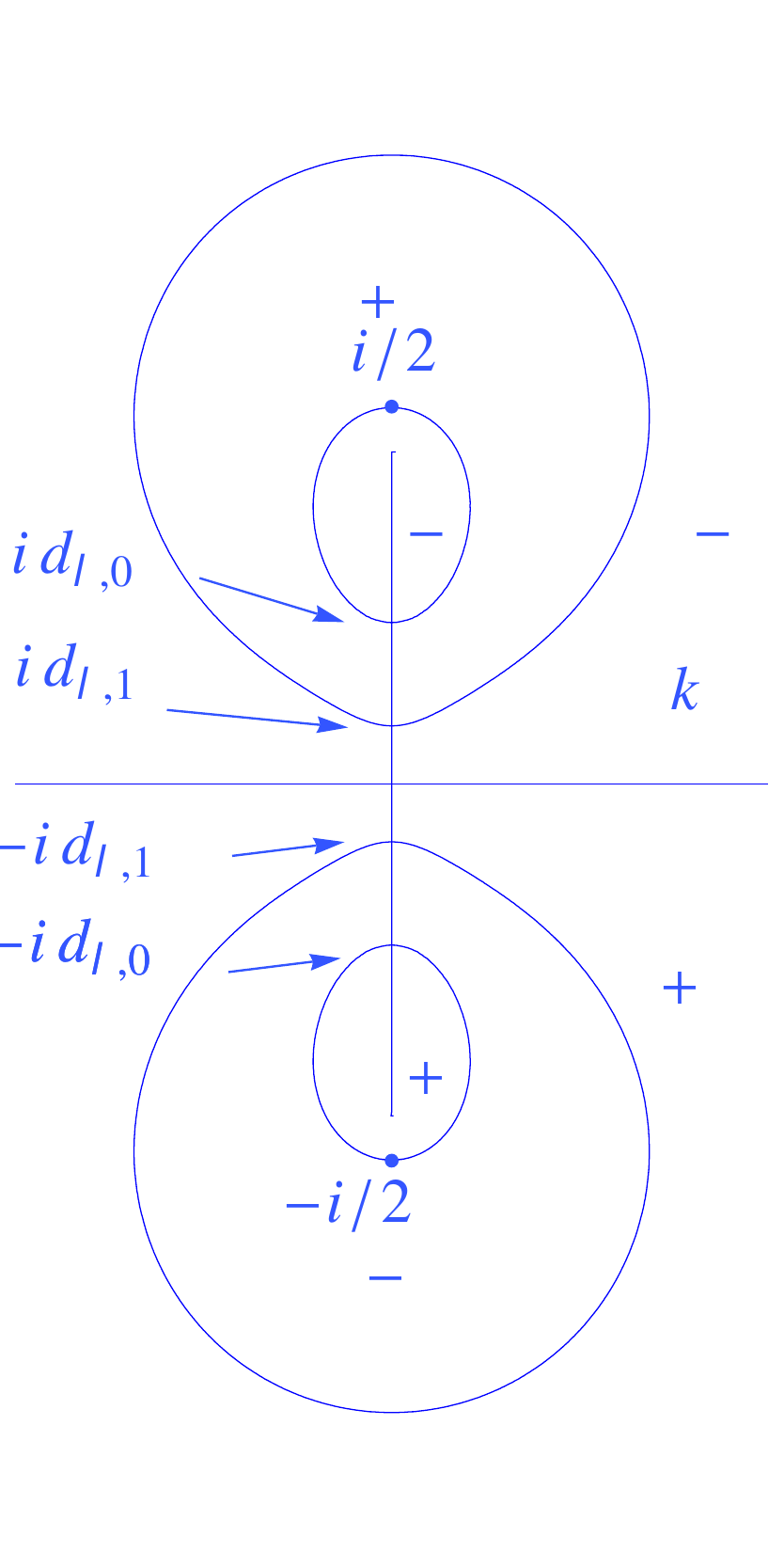}\\
$\frac{-1}{4}<\hat\xi<\frac{-2(c-\omega)\omega}{(c+\omega)^2}$
\end{center}
%\caption{Signature table for $\Im\theta(k,\xi)$ at $\xi=0.$}
%\label{Signature table right 4}
\end{minipage}
\begin{minipage}[ht!]{0.24\linewidth}
\begin{center}
\epsfig{width=30mm,figure=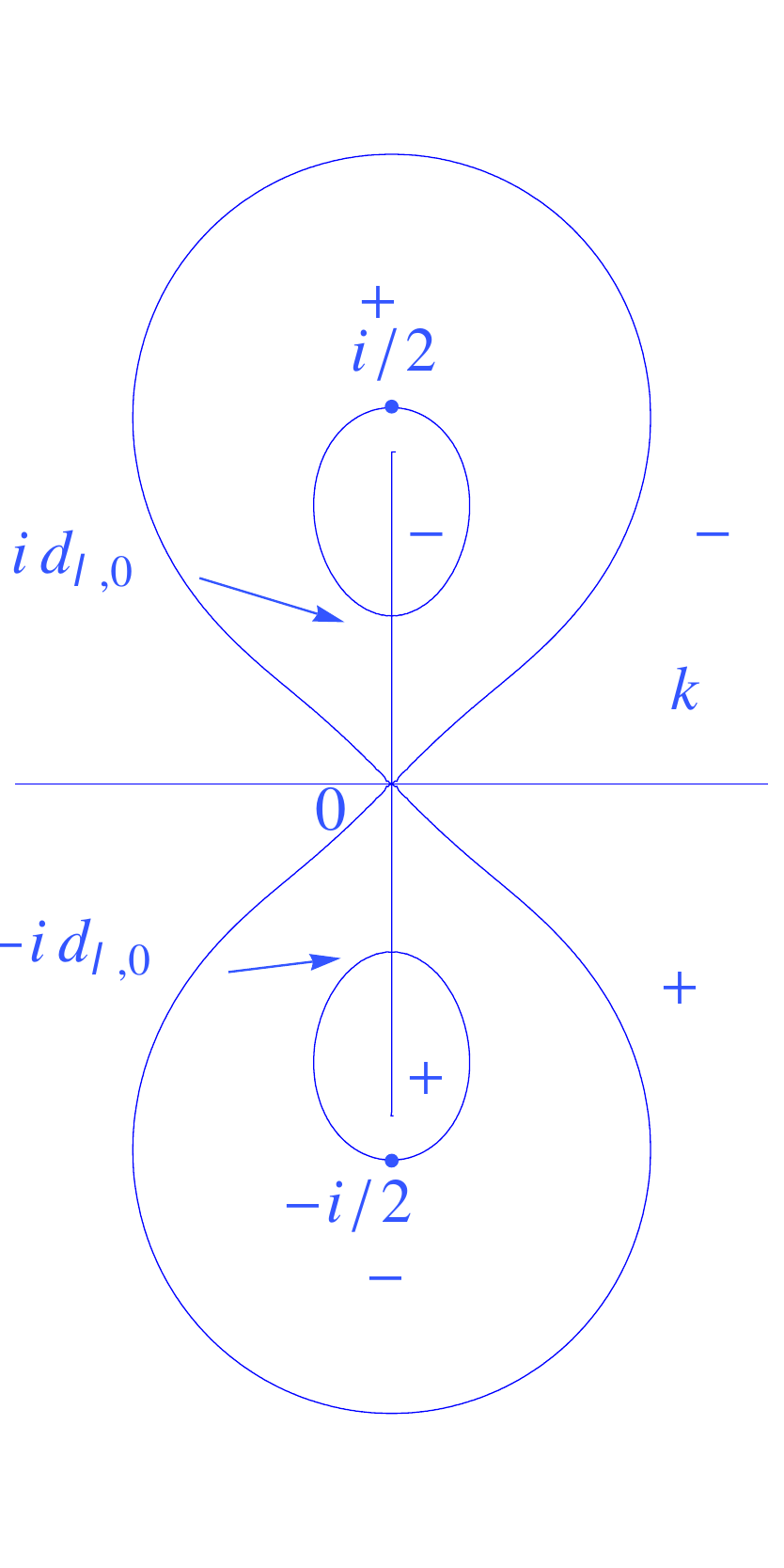}\\
$\hat\xi=\frac{-2(c-\omega)\omega}{(c+\omega)^2}$
\end{center}
%\caption{Signature table for $\Im\theta(k,\xi)$ at $\xi=0.$}
%\label{Signature table right 4}
\end{minipage}
\caption{Signature table for $\Im g_{\l}(k,\xi)$ when
$\frac{c}{\omega}>3$. The plotted contours are the lines $\Im g_{\l}=0.$} \label{Signature table left 1234}
\end{figure}
% -------------- %
\begin{figure}[ht!]
\begin{minipage}[ht!]{0.24\linewidth}
\begin{center}
\epsfig{width=30mm,figure=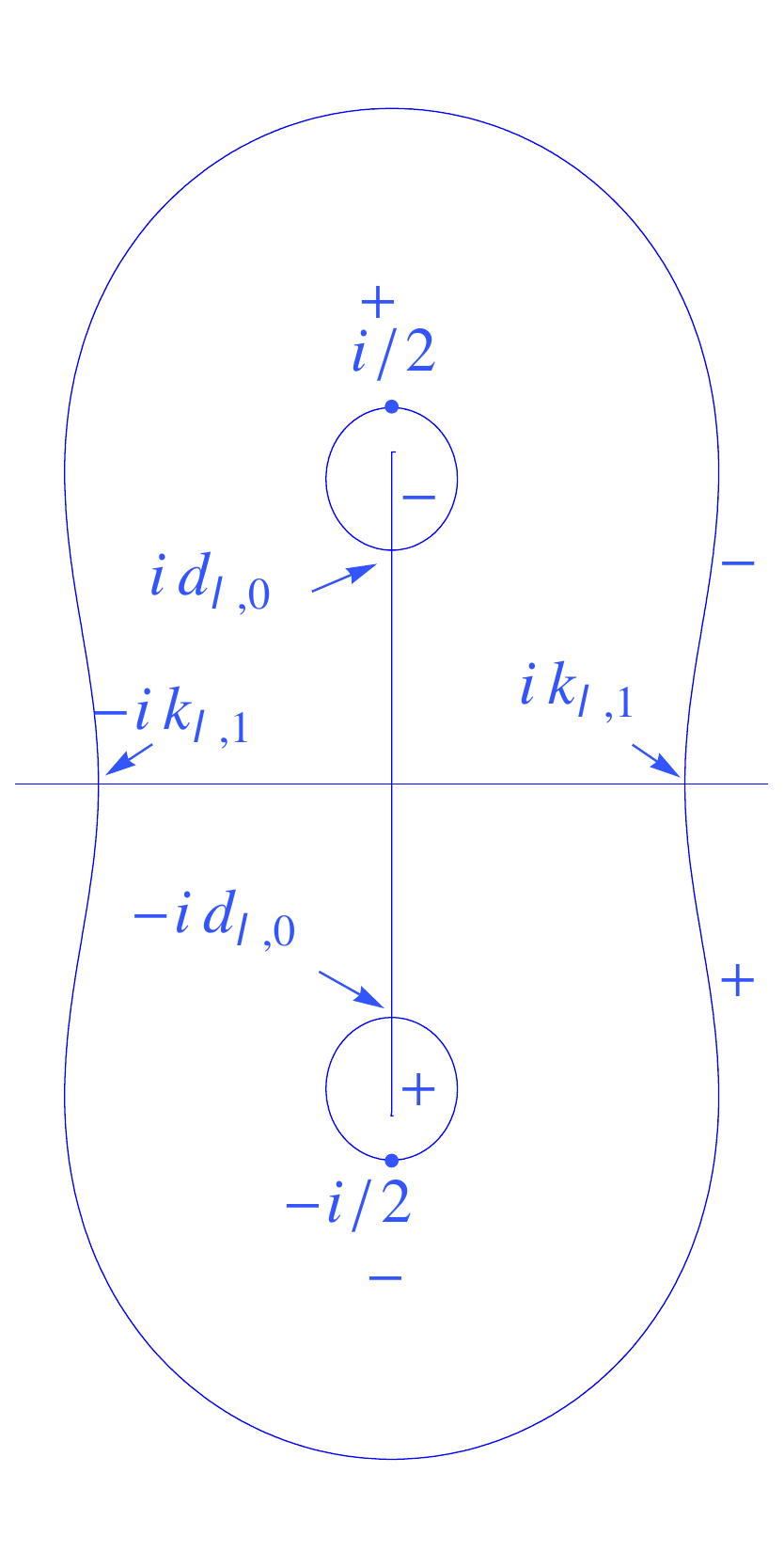}\\
$\frac{-2(c-\omega)\omega}{(c+\omega)^2}<\hat\xi<0$
\end{center}
%\caption{Signature table for $\Im\theta(k,\xi)$ at $\xi=0.$}
%\label{Signature table right 4}
\end{minipage}
% --------- %
\begin{minipage}[ht!]{0.24\linewidth}
\begin{center}
\epsfig{width=30mm,figure=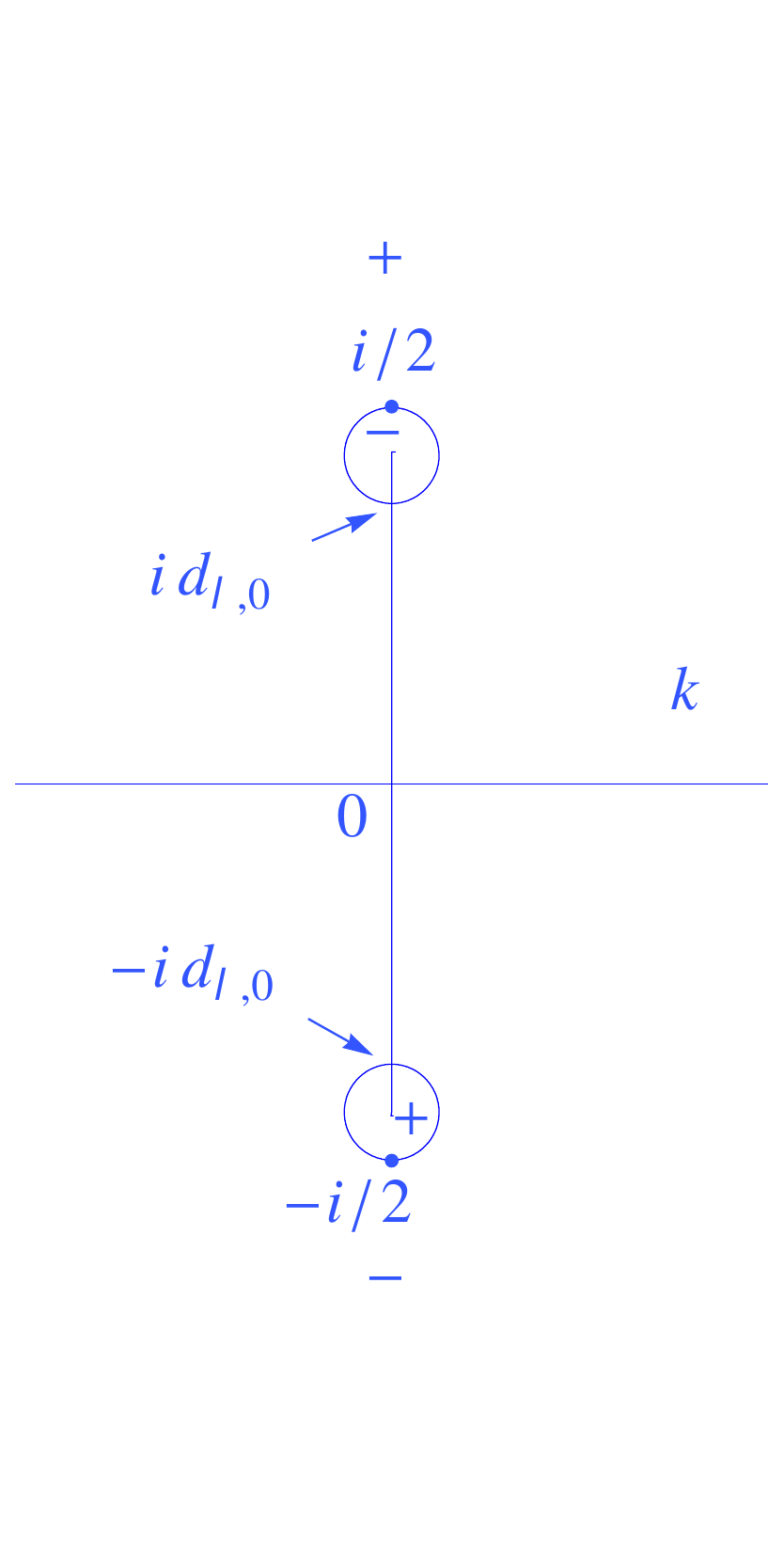}\\
$0\leq\hat\xi<2$
\end{center}
%\caption{Signature table for $\Im\theta(k,\xi)$ at $\xi=0.$}
%\label{Signature table right 4}
\end{minipage}
% ---------- %
\begin{minipage}[ht!]{0.24\linewidth}
\begin{center}
\epsfig{width=30mm,figure=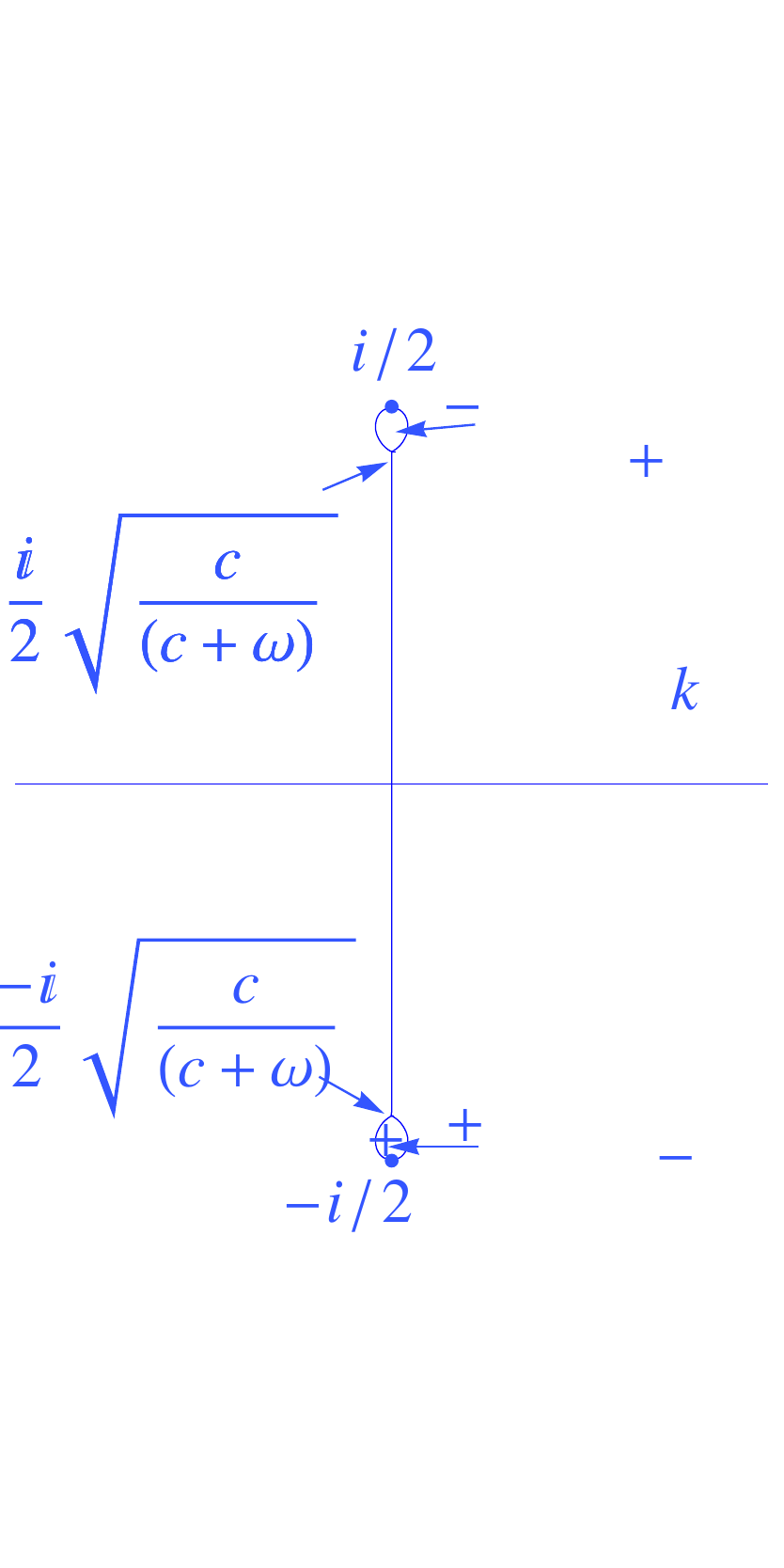}\\
$\hat\xi=2$
\end{center}
%\caption{Signature table for $\Im\theta(k,\xi)$ at $\xi=0.$}
%\label{Signature table right 4}
\end{minipage}
\begin{minipage}[ht!]{0.24\linewidth}
\begin{center}
\epsfig{width=30mm,figure=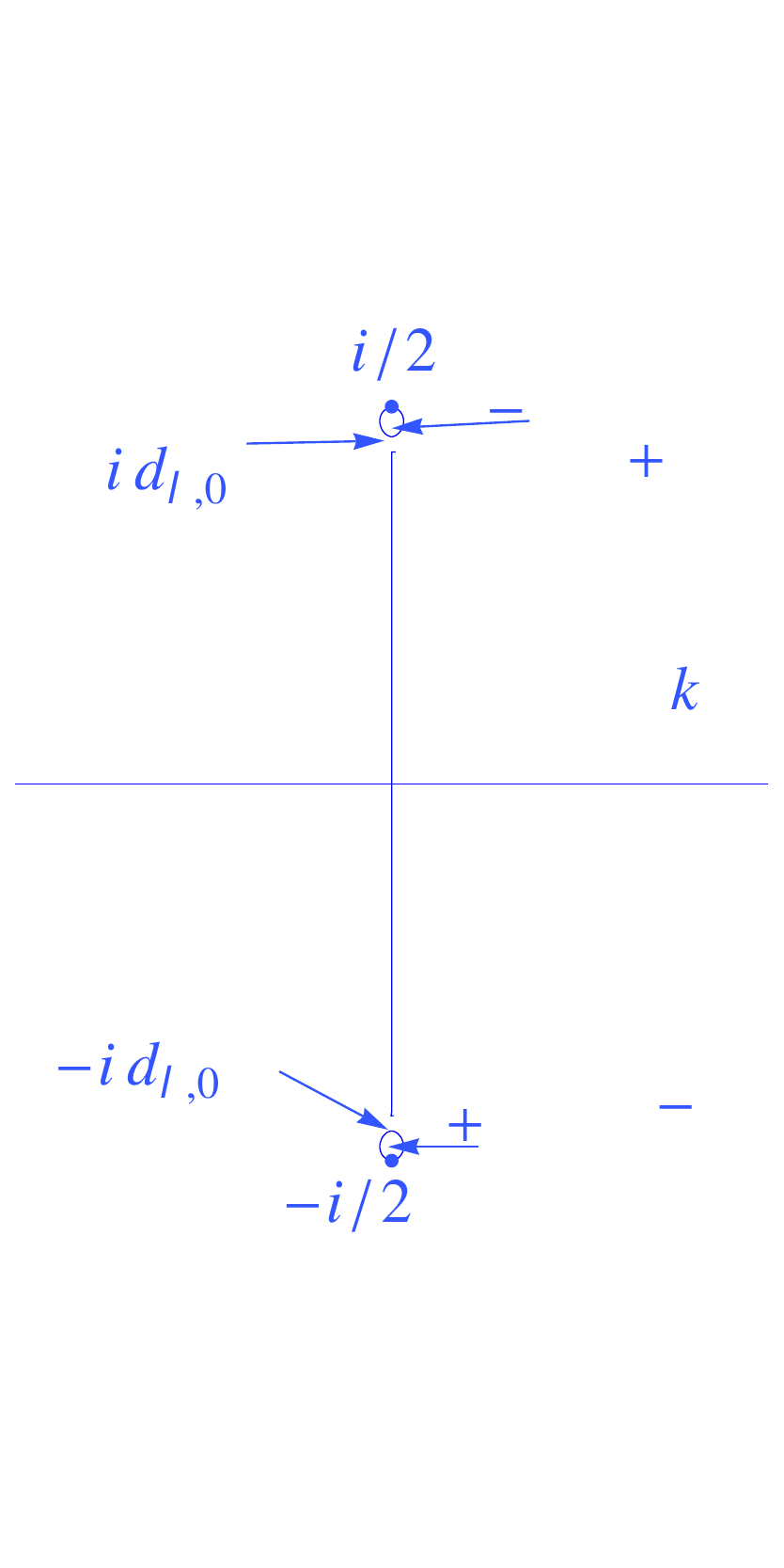}\\
$\hat\xi > 2$
\end{center}
%\caption{Signature table for $\Im\theta(k,\xi)$ at $\xi=0.$}
%\label{Signature table right 4}
\end{minipage}
\caption{Signature table for $\Im g_{\l} (k,\xi)$ when
$\frac{c}{\omega}>3$. The plotted contours are the lines $\Im g_{\l}=0.$} \label{Signature table left 5678}
\end{figure}

% --------------------------------------------------------------

\begin{figure}[ht!]
\begin{minipage}[ht]{0.24\linewidth}
\begin{center}
\epsfig{width=30mm,figure=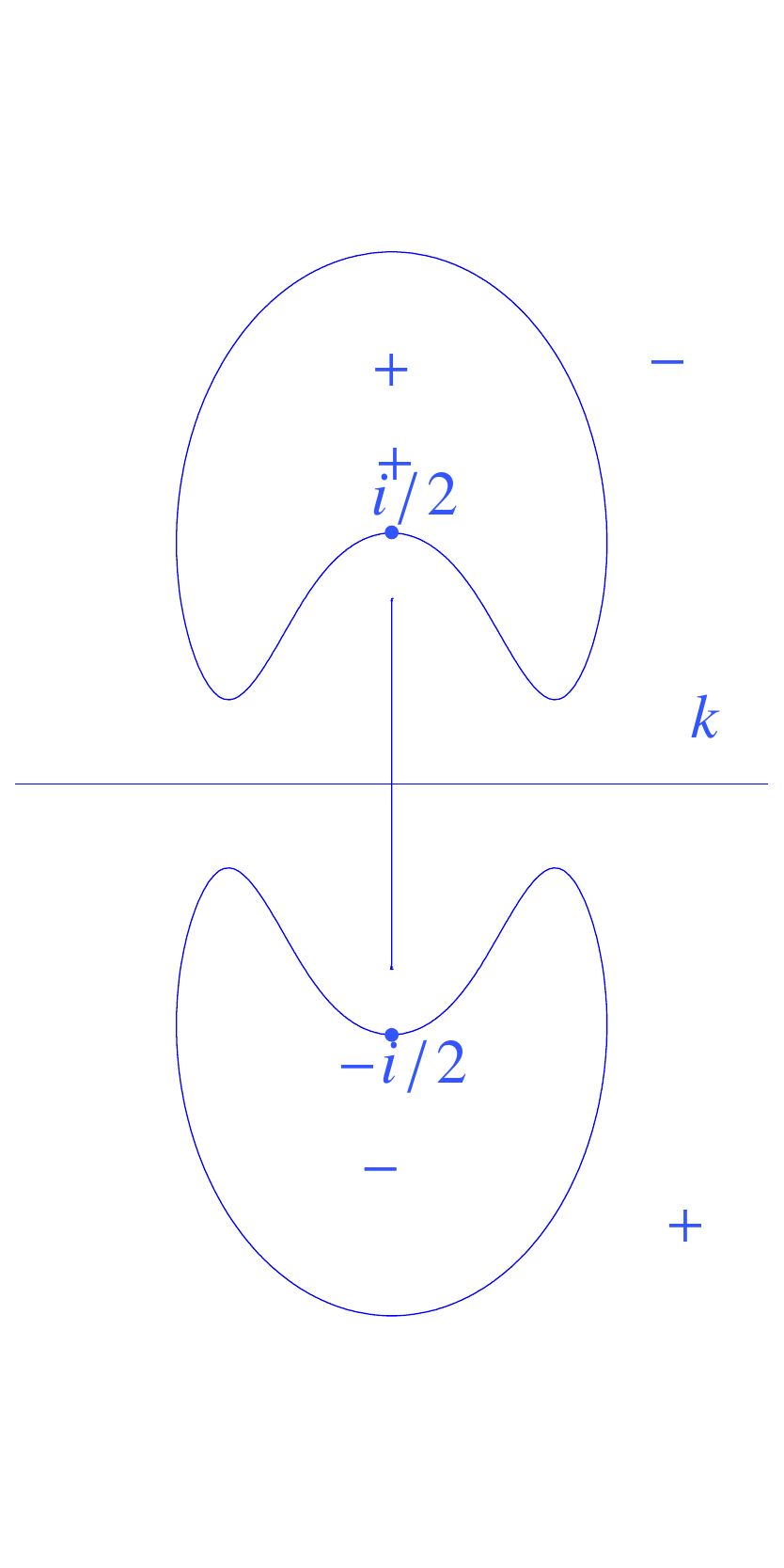}\\
$\hat\xi<\dsfrac{-1}{4}$
\end{center}
%\caption{Signature table for $\Im\theta(k,\xi)$ at $\xi=0.$}
%\label{Signature table right 4}
\end{minipage}
% --------- %
\begin{minipage}[ht!]{0.24\linewidth}
\begin{center}
\epsfig{width=30mm,figure=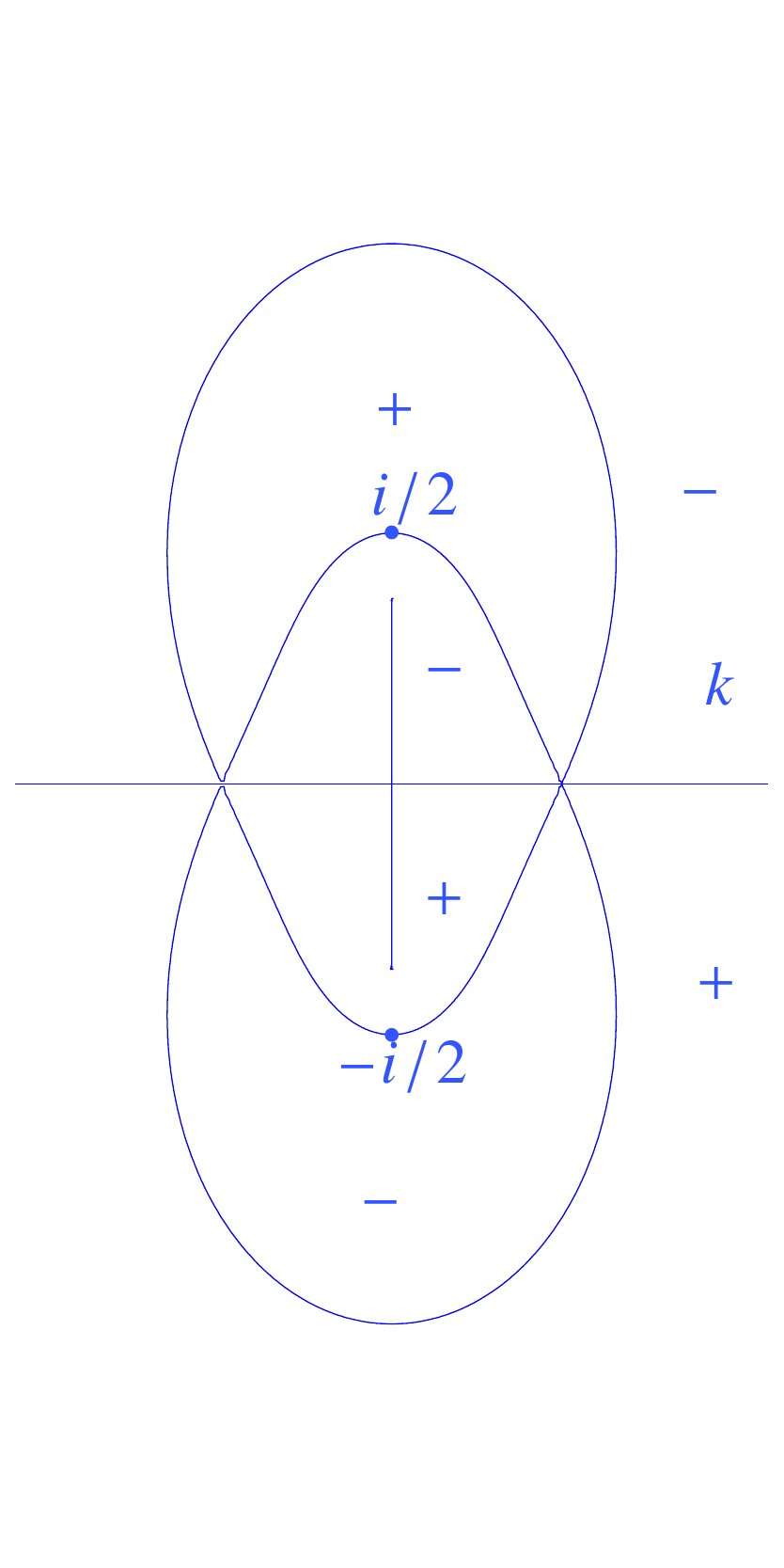}\\
$\hat\xi=\dsfrac{-1}{4}$
\end{center}
%\caption{Signature table for $\Im\theta(k,\xi)$ at $\xi=0.$}
%\label{Signature table right 4}
\end{minipage}
% ---------- %
\begin{minipage}[ht!]{0.24\linewidth}
\begin{center}
\epsfig{width=30mm,figure=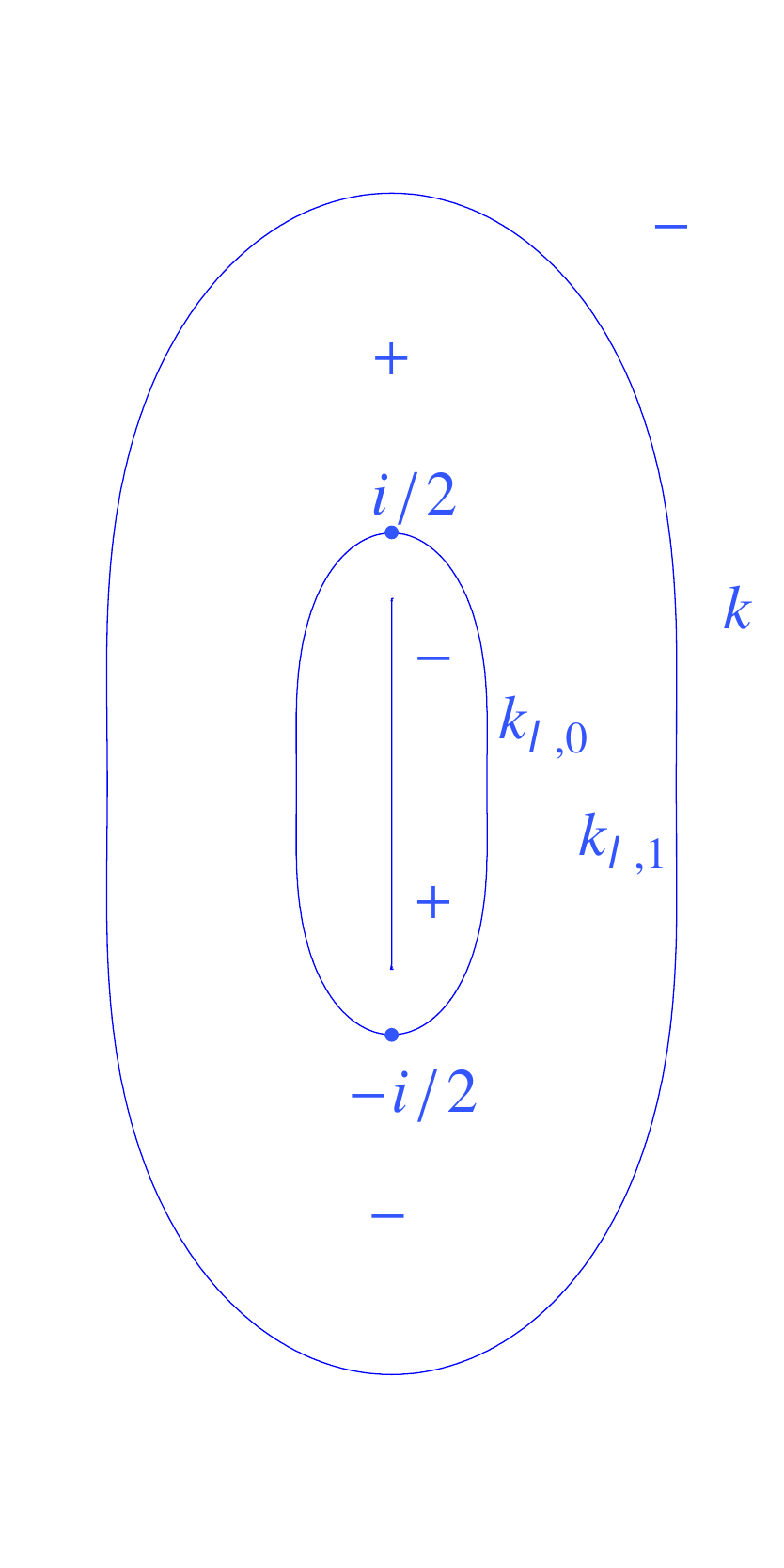}\\
$\frac{-1}{4}<\hat\xi<\frac{-2(c-\omega)\omega}{(c+\omega)^2}$
\end{center}
%\caption{Signature table for $\Im\theta(k,\xi)$ at $\xi=0.$}
%\label{Signature table right 4}
\end{minipage}
\begin{minipage}[ht!]{0.24\linewidth}
\begin{center}
\epsfig{width=30mm,figure=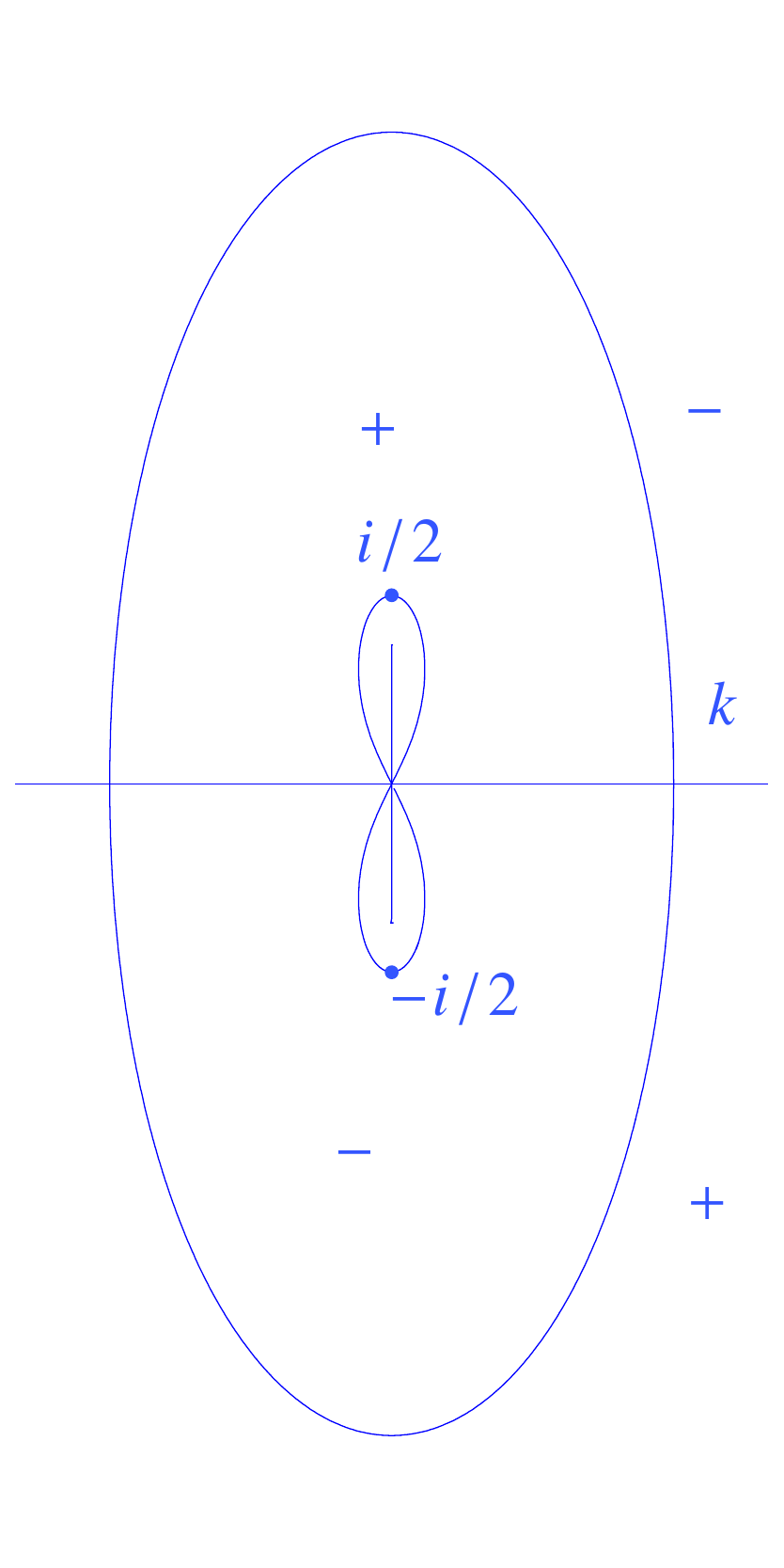}\\
$\hat\xi=\frac{-2(c-\omega)\omega}{(c+\omega)^2}$
\end{center}
%\caption{Signature table for $\Im\theta(k,\xi)$ at $\xi=0.$}
%\label{Signature table right 4}
\end{minipage}
\caption{Signature table for $\Im g_{\l}(k,\xi)$ when
$1<\frac{c}{\omega}<3$. The plotted contours are the lines $\Im g_{\l}=0.$} \label{Signature table left(2) 1234}
\end{figure}
% -------------- %
% --------------
\begin{figure}[ht!]
\begin{minipage}[ht!]{0.24\linewidth}
\begin{center}
\epsfig{width=30mm,figure=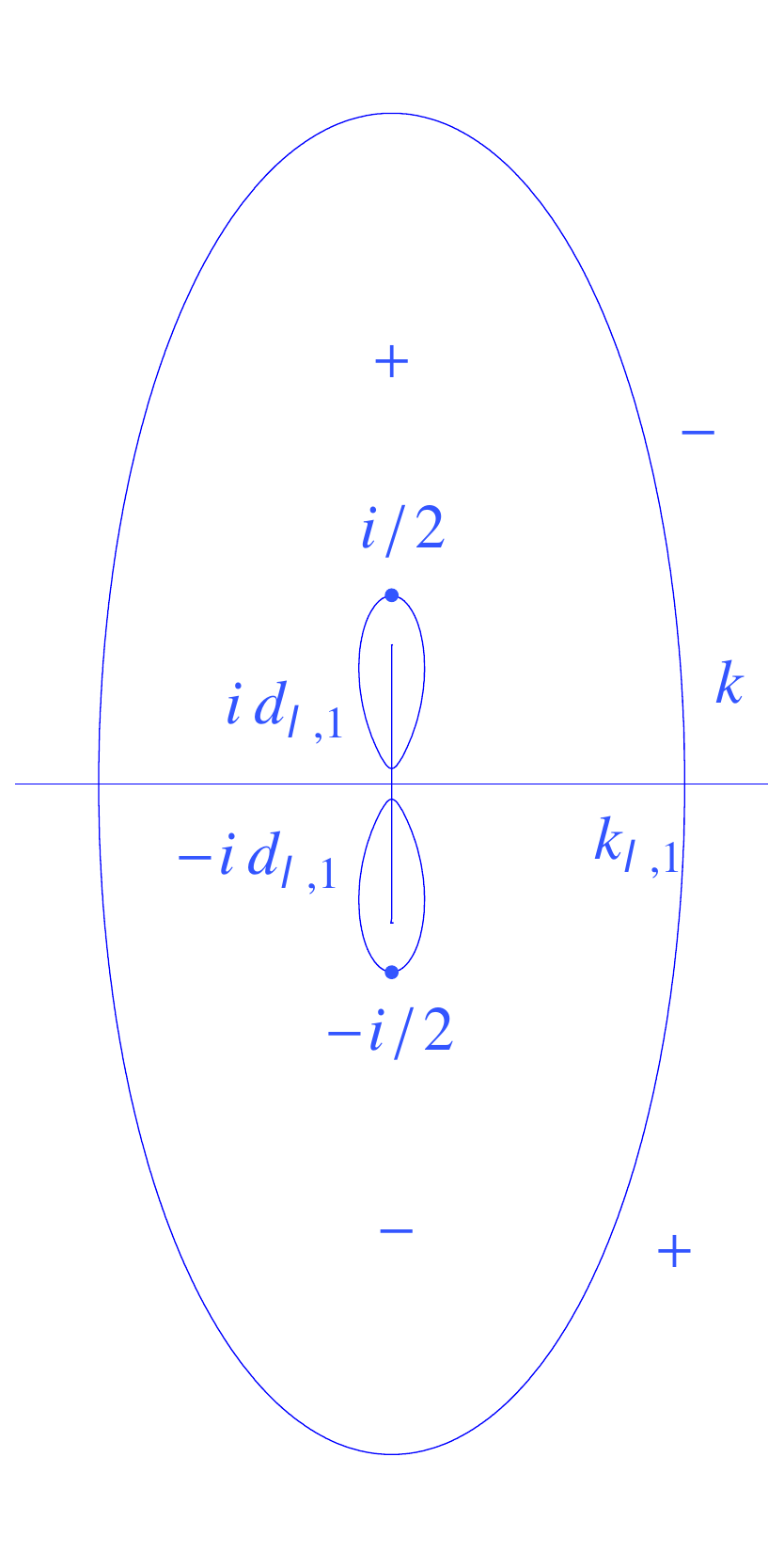}\\
$\frac{-2(c-\omega)\omega}{(c+\omega)^2}<\hat\xi<0$
\end{center}
%\caption{Signature table for $\Im\theta(k,\xi)$ at $\xi=0.$}
%\label{Signature table right 4}
\end{minipage}
% --------- %
\begin{minipage}[ht!]{0.24\linewidth}
\begin{center}
\epsfig{width=30mm,figure=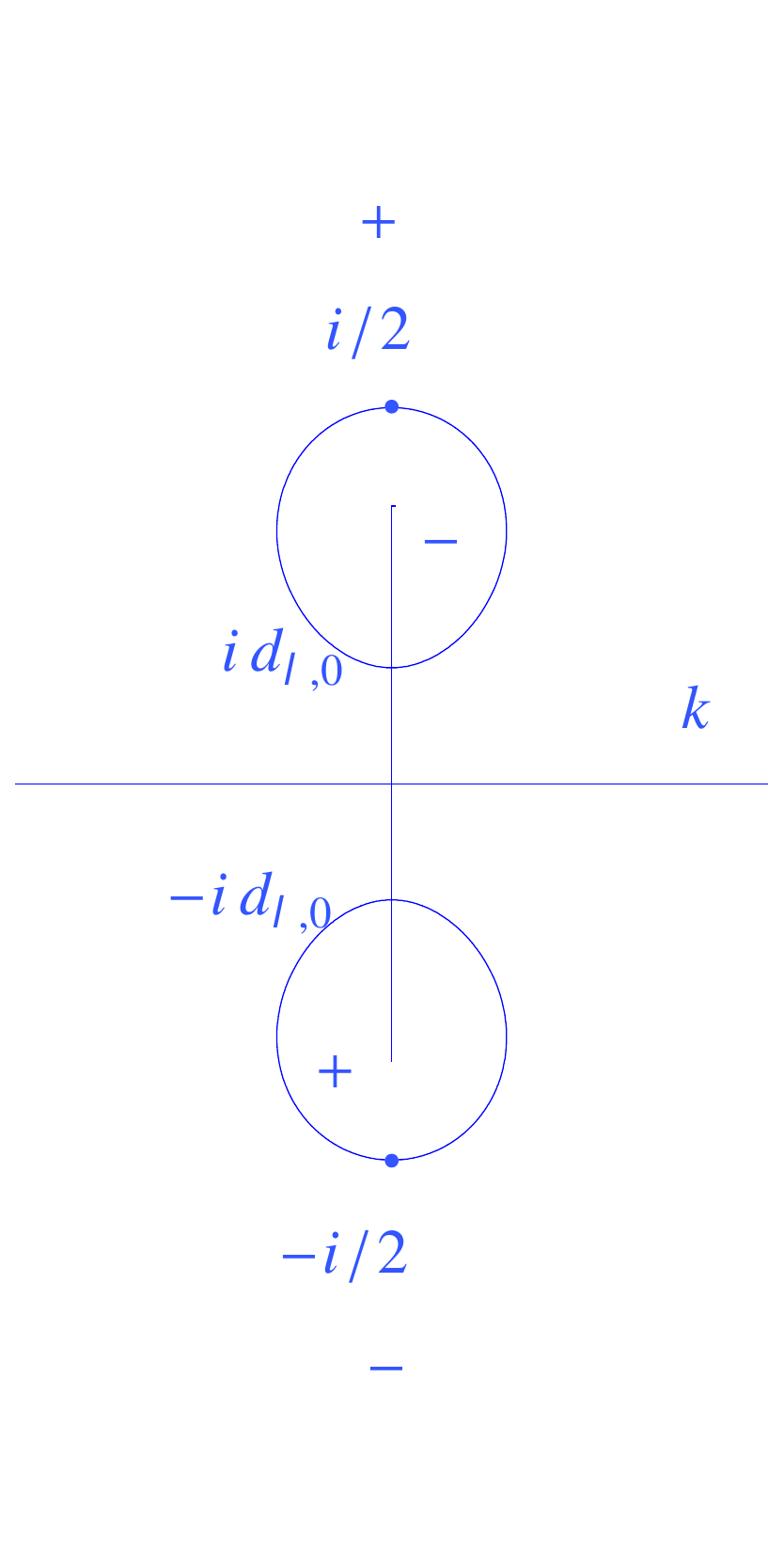}\\
$0\leq\hat\xi<2$
\end{center}
%\caption{Signature table for $\Im\theta(k,\xi)$ at $\xi=0.$}
%\label{Signature table right 4}
\end{minipage}
% ---------- %
\begin{minipage}[ht!]{0.24\linewidth}
\begin{center}
\epsfig{width=30mm,figure=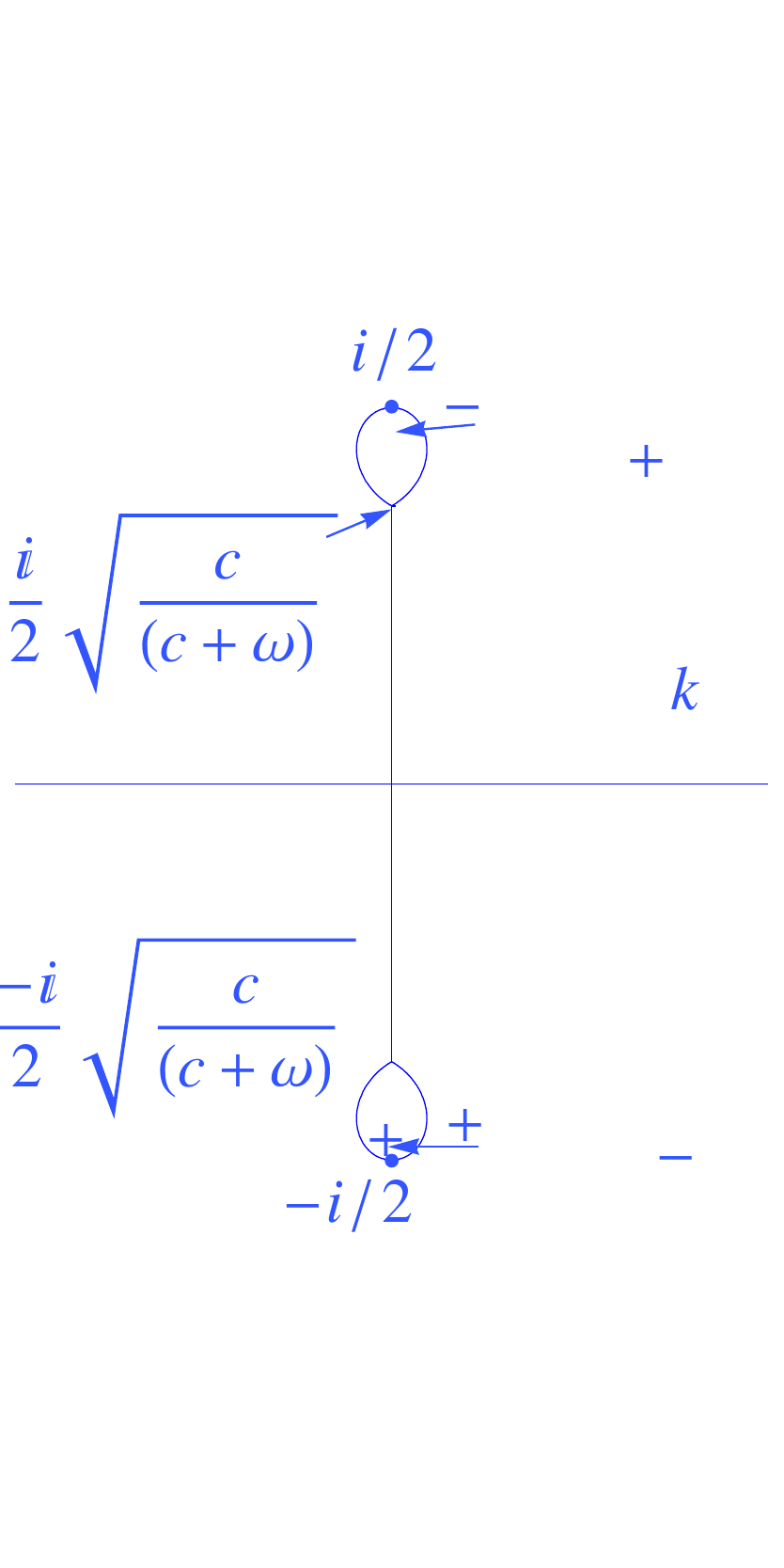}\\
$\hat\xi=2$
\end{center}
%\caption{Signature table for $\Im\theta(k,\xi)$ at $\xi=0.$}
%\label{Signature table right 4}
\end{minipage}
\begin{minipage}[h]{0.24\linewidth}
\begin{center}
\epsfig{width=30mm,figure=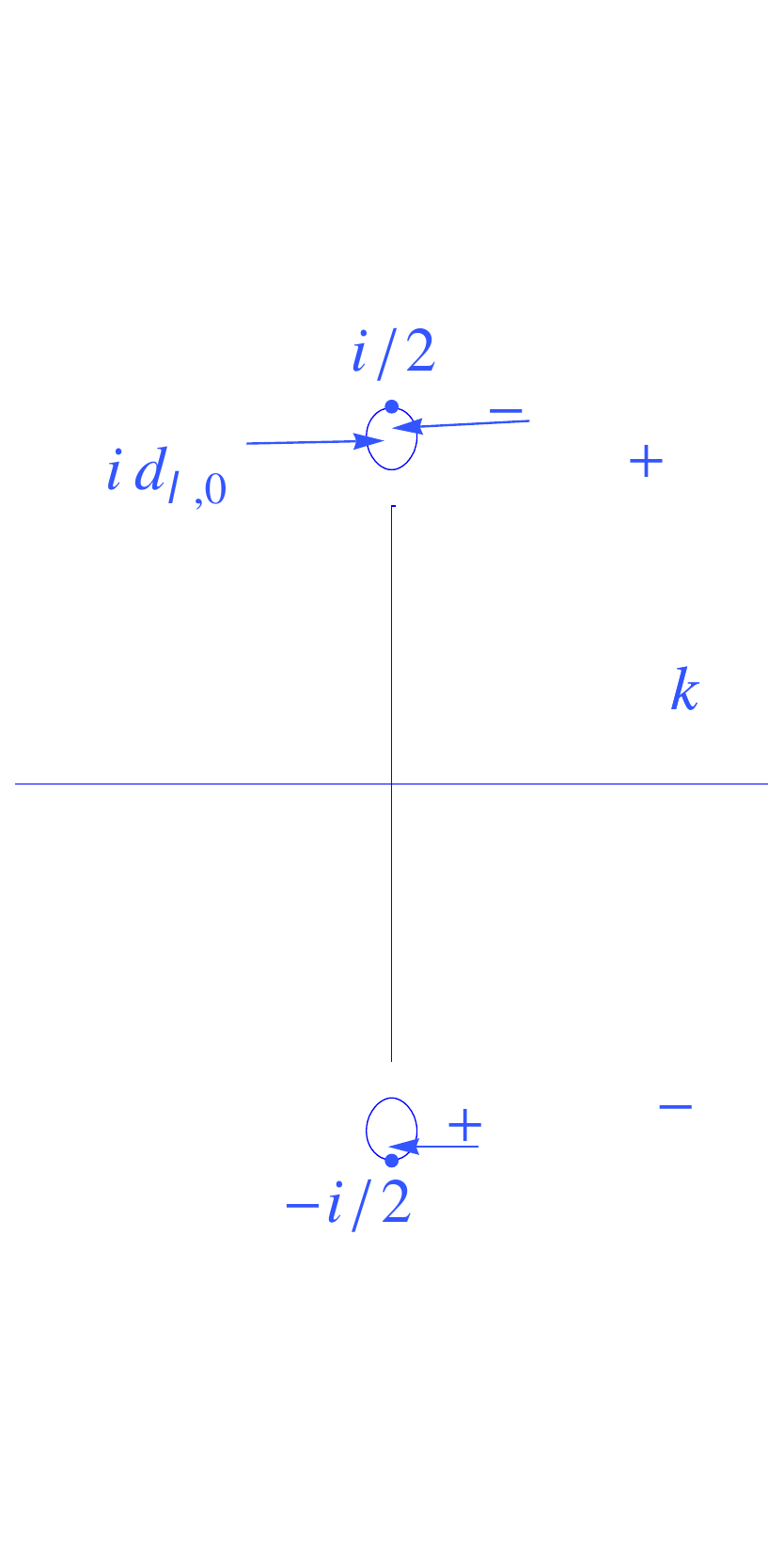}\\
$\hat\xi > 2$
\end{center}
%\caption{Signature table for $\Im\theta(k,\xi)$ at $\xi=0.$}
%\label{Signature table right 4}
\end{minipage}
\caption{Signature table for $\Im g_{\l}(k,\xi)$ when
$1<\frac{c}{\omega}<3$. The plotted contours are the lines $\Im g_{\l}=0.$} \label{Signature table left(2) 5678}
\end{figure}

% -----------------------------------------------------------

\begin{figure}
\begin{minipage}[ht!]{0.24\linewidth}
\begin{center}
\epsfig{width=30mm,figure=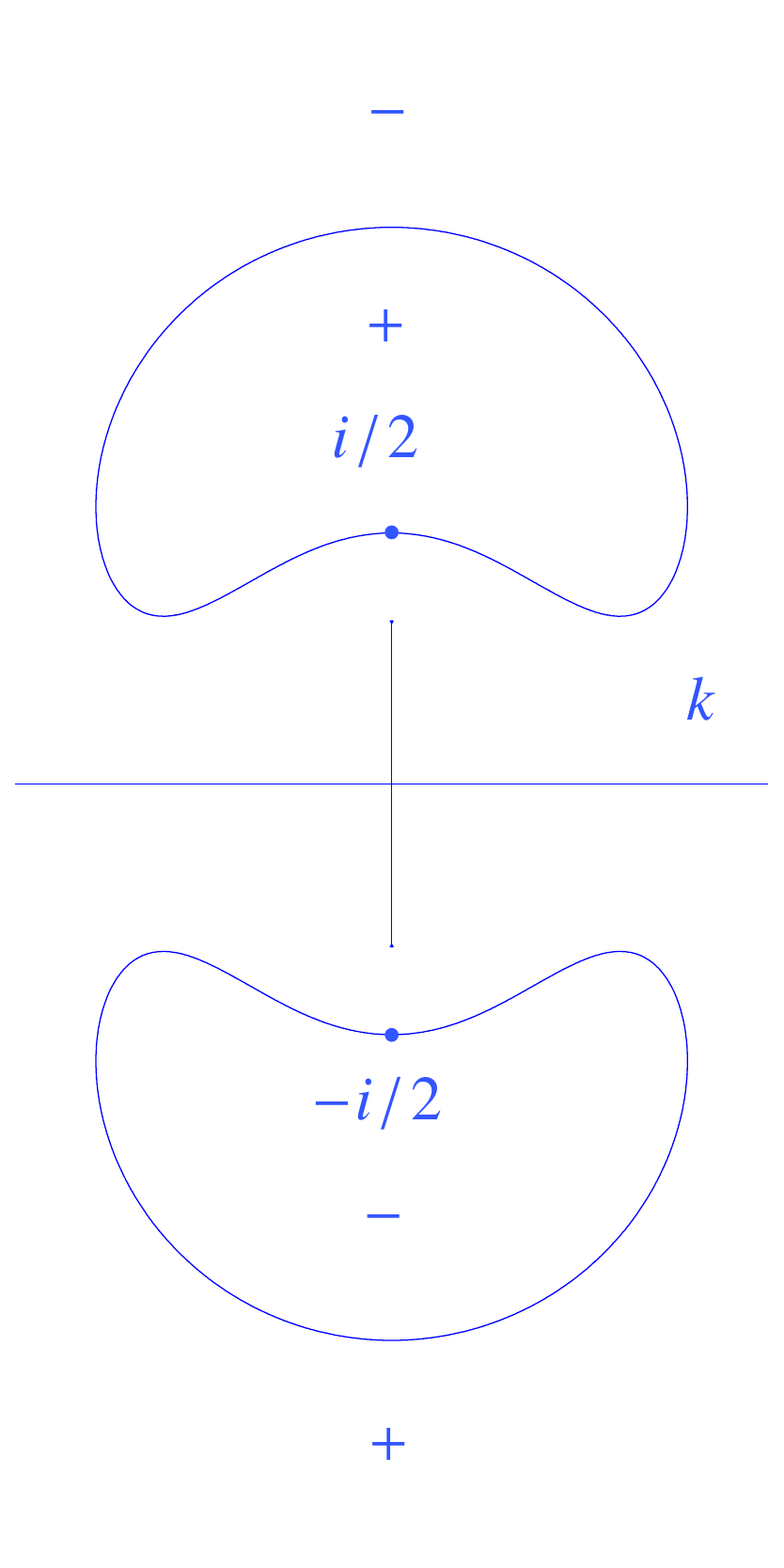}\\
$\hat\xi<\dsfrac{-1}{4}$
\end{center}
%\caption{Signature table for $\Im\theta(k,\xi)$ at $\xi=0.$}
%\label{Signature table right 4}
\end{minipage}
% --------- %
\begin{minipage}[h!]{0.24\linewidth}
\begin{center}
\epsfig{width=30mm,figure=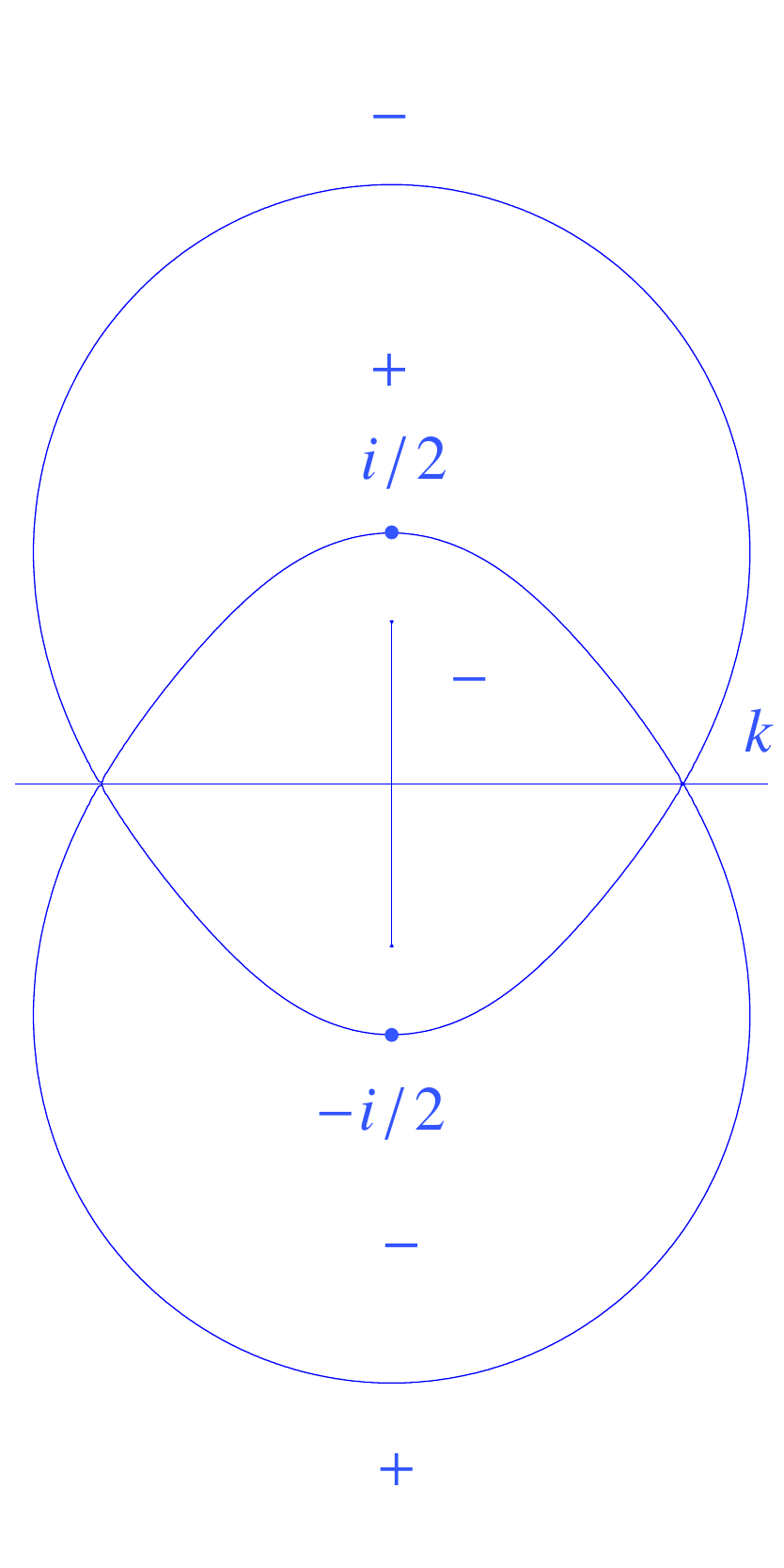}\\
$\hat\xi=\dsfrac{-1}{4}$
\end{center}
%\caption{Signature table for $\Im\theta(k,\xi)$ at $\xi=0.$}
%\label{Signature table right 4}
\end{minipage}
% ---------- %
\begin{minipage}[ht!]{0.24\linewidth}
\begin{center}
\epsfig{width=30mm,figure=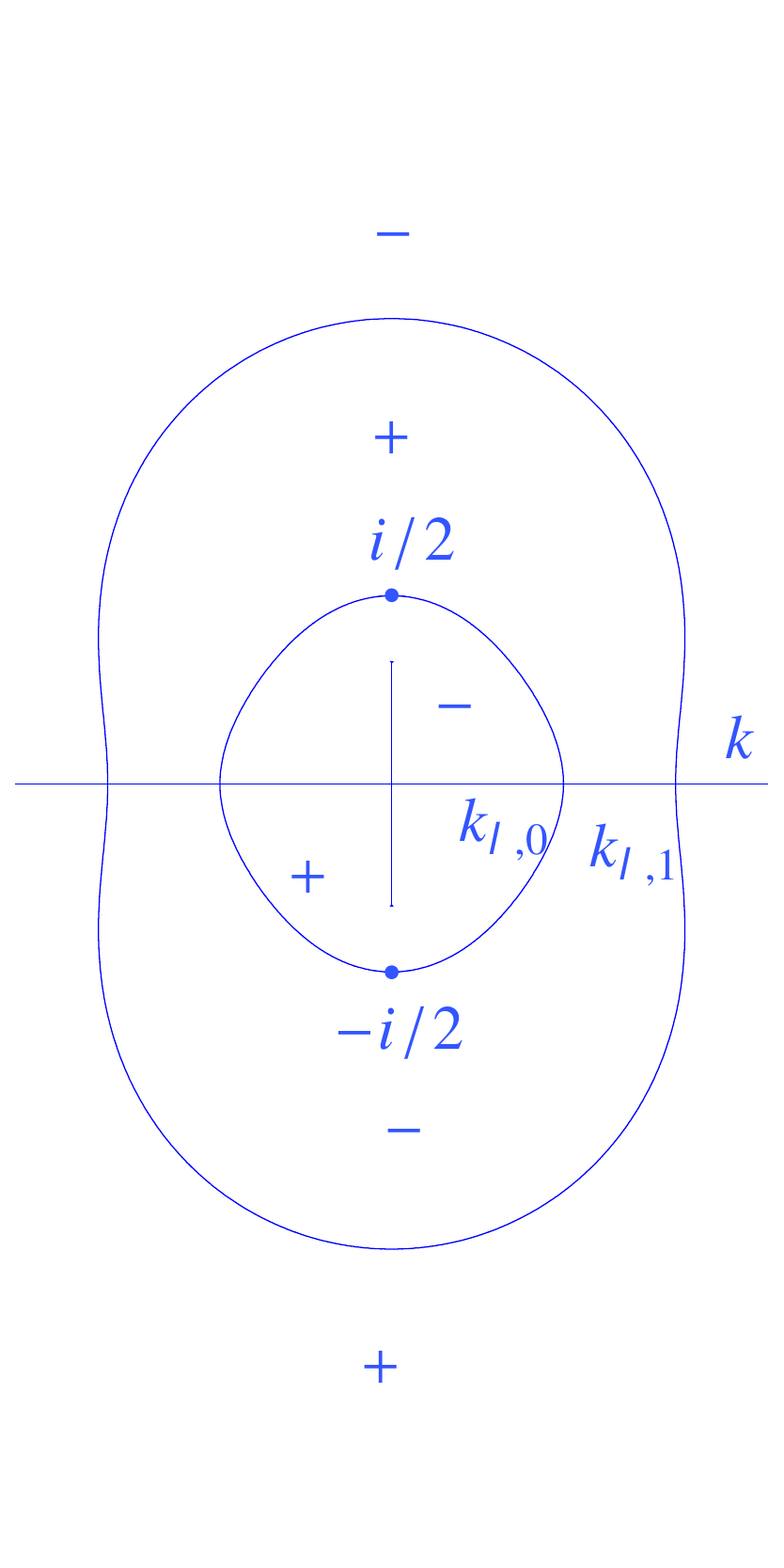}\\
$\frac{-1}{4}<\hat\xi<0$
\end{center}
%\caption{Signature table for $\Im\theta(k,\xi)$ at $\xi=0.$}
%\label{Signature table right 4}
\end{minipage}
\begin{minipage}[ht!]{0.24\linewidth}
\begin{center}
\epsfig{width=30mm,figure=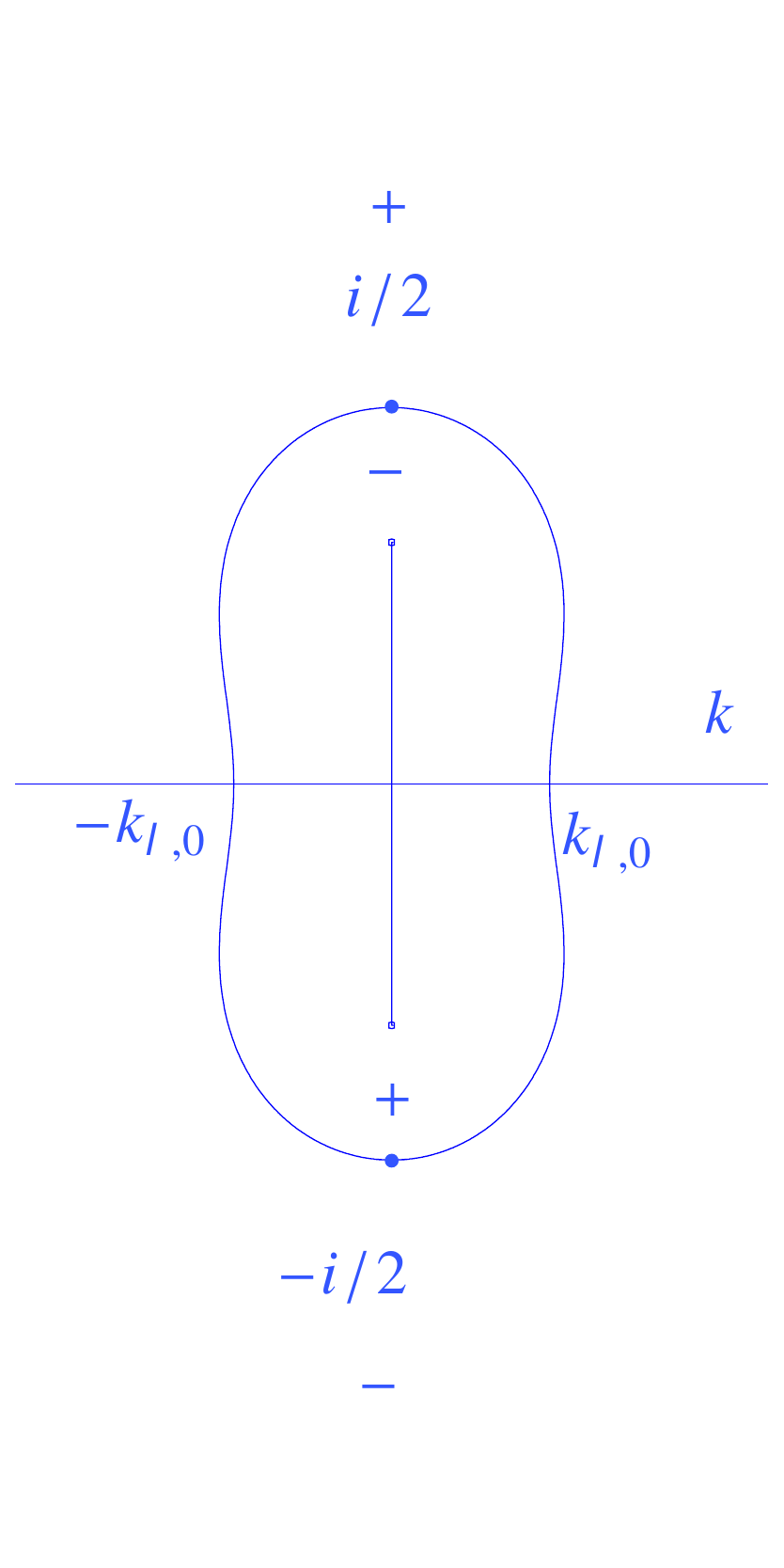}\\
$0\leq\hat\xi<\frac{-2(c-\omega)\omega}{(c+\omega)^2}$
\end{center}
%\caption{Signature table for $\Im\theta(k,\xi)$ at $\xi=0.$}
%\label{Signature table right 4}
\end{minipage}
\caption{Signature table for $\Im g_{\l}(k,\xi)$ when
$0<\frac{c}{\omega}<1$. The plotted contours are the lines $\Im g_{\l}=0.$} \label{Signature table left(3) 1234}
\end{figure}
% -------------- %
\begin{figure}
\begin{minipage}[ht!]{0.24\linewidth}
\begin{center}
\epsfig{width=30mm,figure=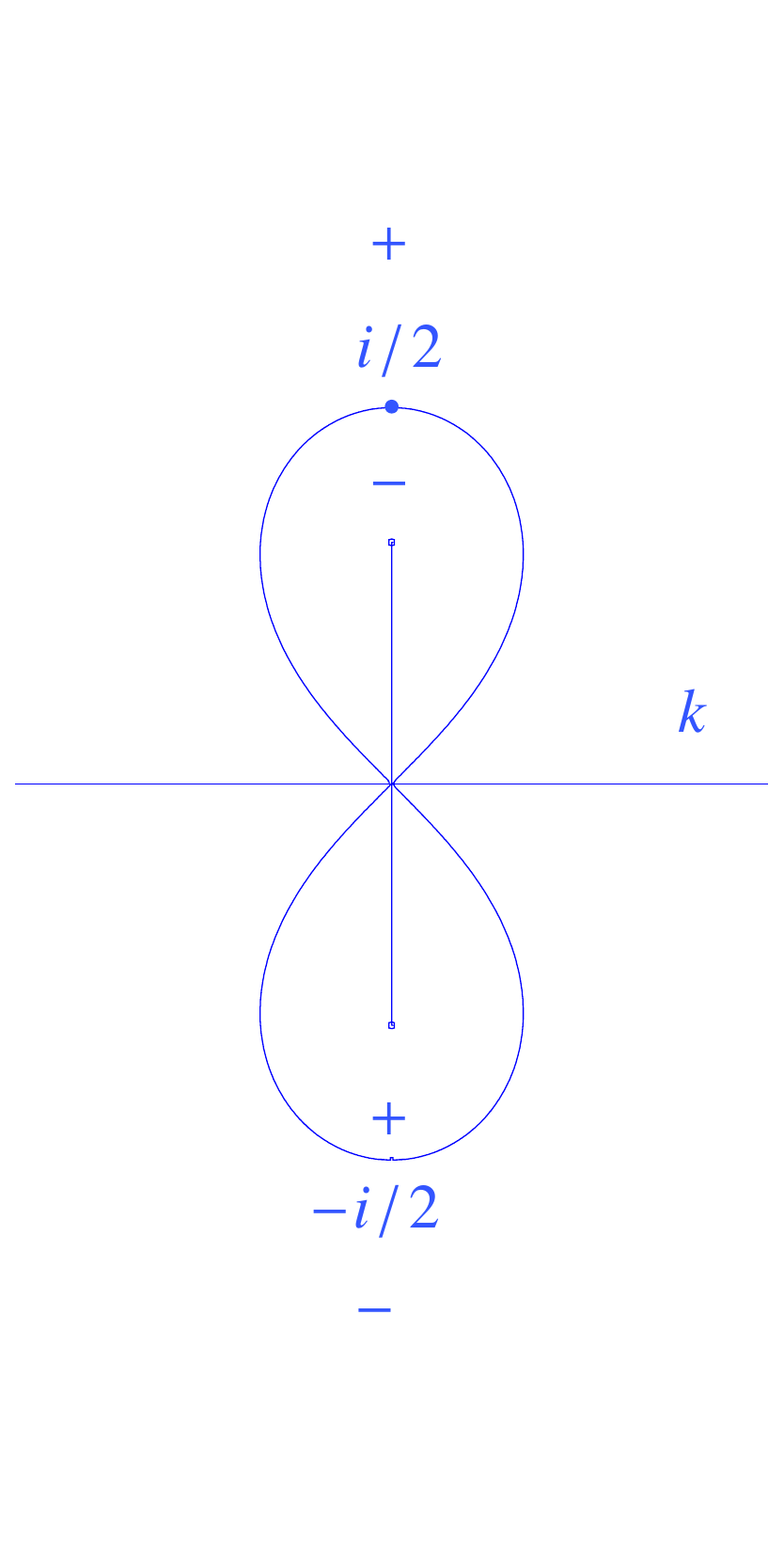}\\
$\hat\xi=\frac{-2(c-\omega)\omega}{(c+\omega)^2}$
\end{center}
%\caption{Signature table for $\Im\theta(k,\xi)$ at $\xi=0.$}
%\label{Signature table right 4}
\end{minipage}
% --------- %
\begin{minipage}[ht!]{0.24\linewidth}
\begin{center}
\epsfig{width=30mm,figure=Signature_table_left_2__6-eps-converted-to}\\
$\frac{-2(c-\omega)\omega}{(c+\omega)^2}<\hat\xi<2$
\end{center}
%\caption{Signature table for $\Im\theta(k,\xi)$ at $\xi=0.$}
%\label{Signature table right 4}
\end{minipage}
% ---------- %
\begin{minipage}[h!]{0.24\linewidth}
\begin{center}
\epsfig{width=30mm,figure=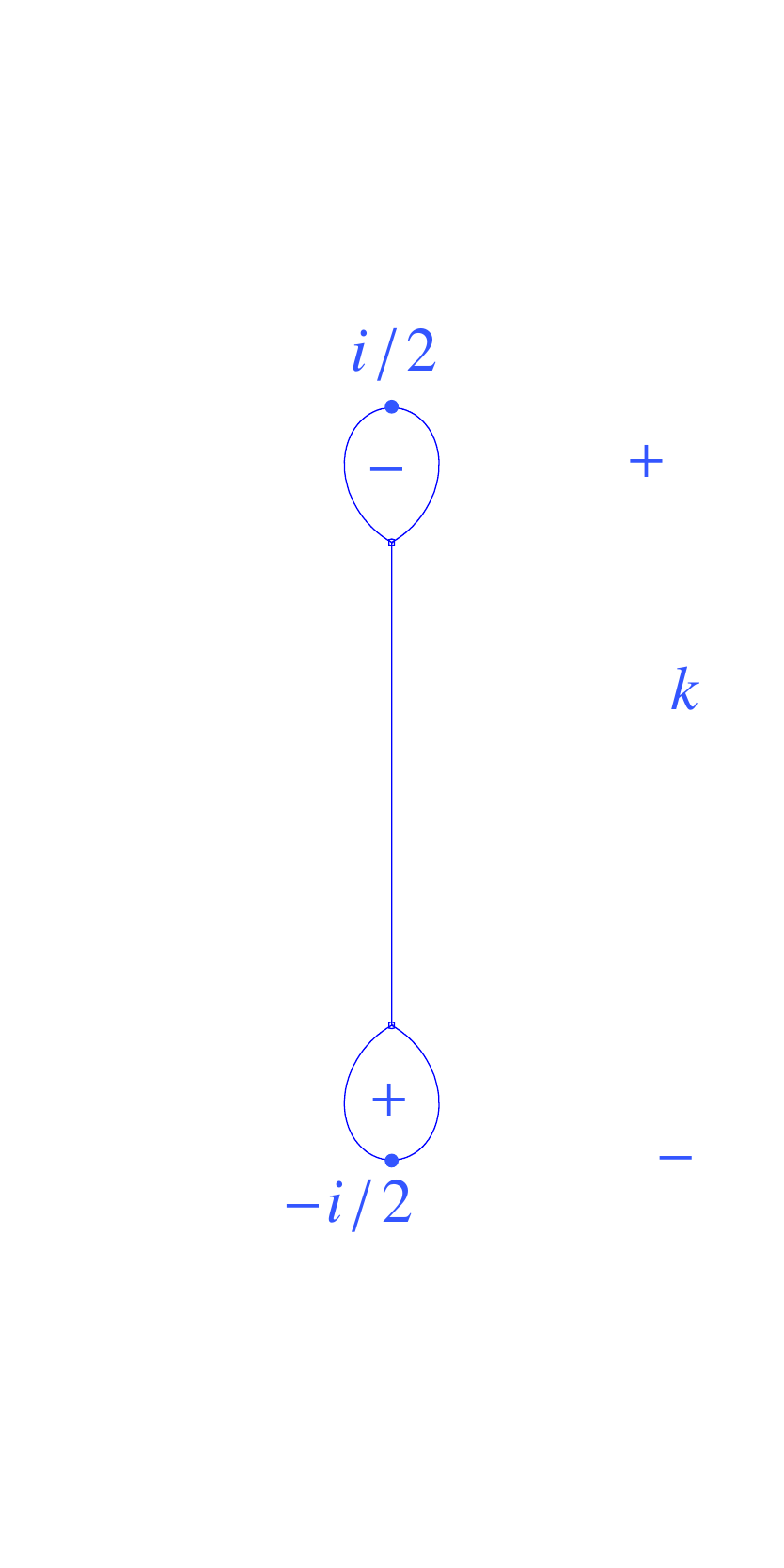}\\
$\hat\xi=2$
\end{center}
%\caption{Signature table for $\Im\theta(k,\xi)$ at $\xi=0.$}
%\label{Signature table right 4}
\end{minipage}
\begin{minipage}[h!]{0.24\linewidth}
\begin{center}
\epsfig{width=30mm,figure=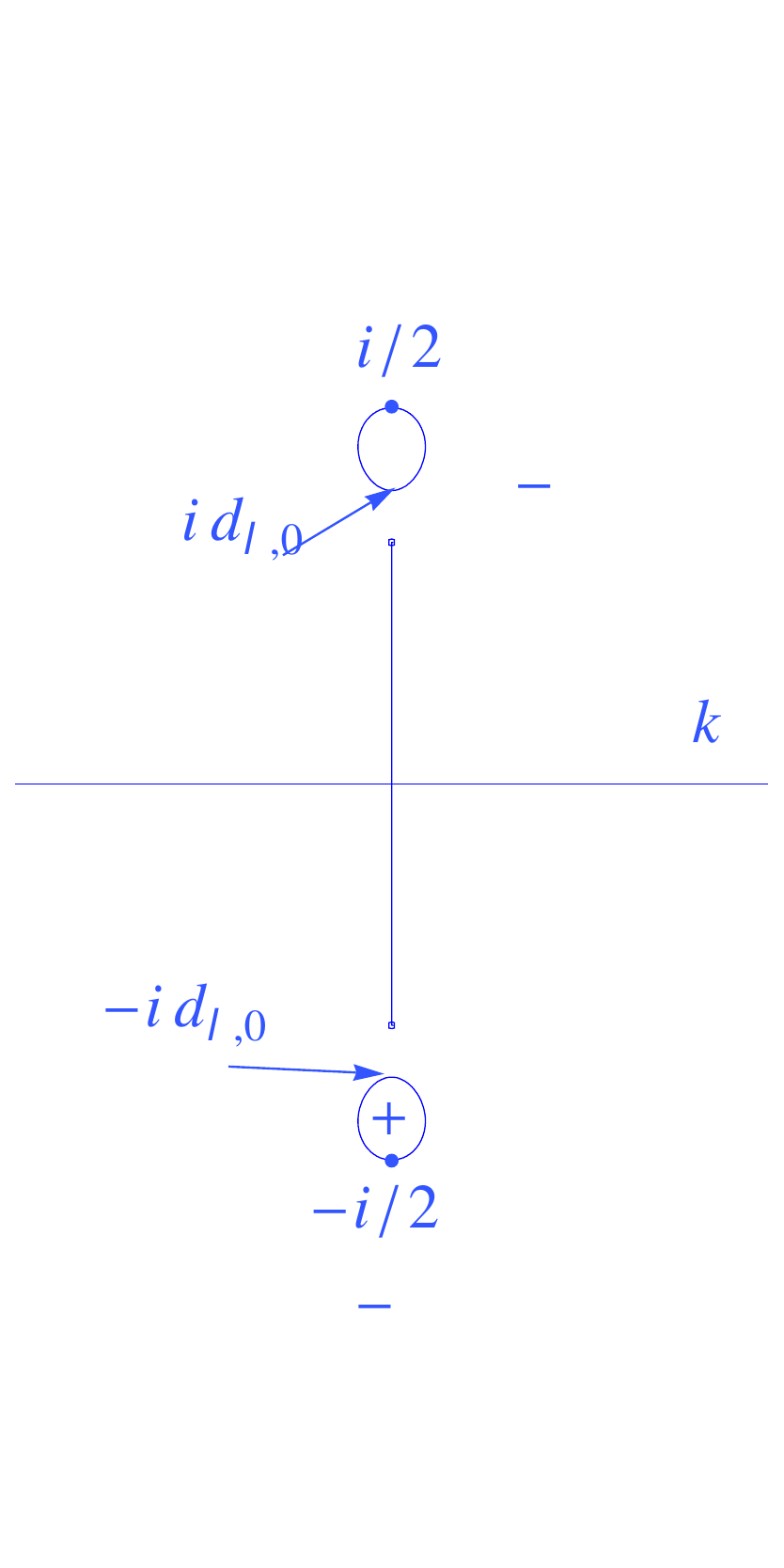}\\
$\hat\xi > 2$
\end{center}
%\caption{Signature table for $\Im\theta(k,\xi)$ at $\xi=0.$}
%\label{Signature table right 4}
\end{minipage}
\caption{Signature table for $\Im g_{\l}(k,\xi)$ when
$0<\frac{c}{\omega}<1$. The plotted contours are the lines $\Im g_{\l}=0.$} \label{Signature table left(3) 5678}
\end{figure}

\subsection{Middle phase functions.}

\begin{figure}
\begin{minipage}[ht!]{1.\linewidth}
%\begin{center}
\vskip-3.1cm\hskip-10mm
\includegraphics[scale=0.8]{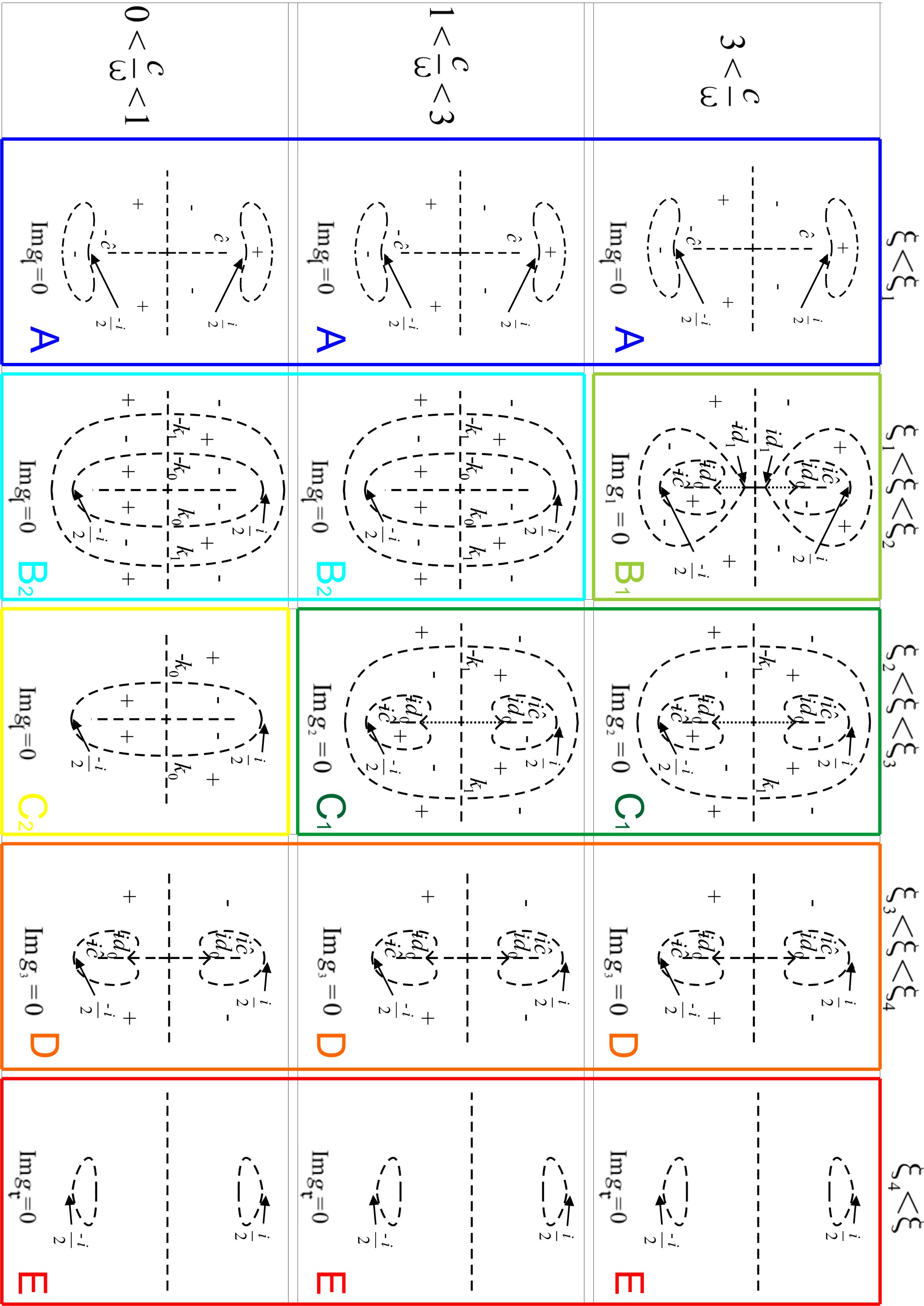}%-eps-converted-to}
%\end{center}
\end{minipage}
\caption{Suitable signature table for $\Im g(k,\xi)$ in different regions of the parameter $\xi.$}
\label{Figure_phase_middle_function_plot}
\end{figure}

We use $g_{\r}$ and $g_{\l}$ in the asymptotic analysis for large positive and large negative $x$, 
when the lines $\Im g_{\r}=0$, $\Im g_{\l}=0$ do not intersect the interval $[\i\c,-\i\c]$. 
When the above lines start to intersect the above segment, we need to construct intermediate phase functions.
The observation
\begin{itemize}
\item $g_{\r,-}-g_{\r,+}=0,$ and
 $\Im g_{\r}>0$ in a vicinity of the segment $[\i\c,0],$
\item $g_{\l,-}+g_{\l,+}=0,$ and
$\Im g_{\l}<0$ in a vicinity of the segment $[\i\c,0]$
\end{itemize}
suggests looking for an (odd) phase function with the following property:
% -----------------------
\begin{enumerate}[1]
\item \hskip-2mm (a). on those subsegments of
$(\i\c,0)$ that lie in the domain $\Im
g(k,\xi)\leq0$, we have
\begin{equation}
\label{prop g 1}\textstyle
g_-(k,\xi)+g_+(k,\xi)=0,
\end{equation}
\setcounter{enumi}{0}\item \hskip-2mm (b). on those subsegments of
$(\i\c,0)$ that lie in the domain $\Im
g(k,\xi)\geq0$, we have
\begin{equation}\label{prop g 2}
\textstyle
g_-(k,\xi)-g_+(k,\xi)=B(\xi),
\end{equation} 
\noindent 
where $B(\xi)$ is a real-valued quantity, which does not depend on 
$k$.
\end{enumerate}
\noindent
Also, a suitable phase function should satisfy the following
properties:
\begin{enumerate}
\setcounter{enumi}{1} \item asymptotics at infinity:
$g(k,\xi)-g_{\r}(k,\xi) =
\mathrm{O}\(k^{-1}\)\ \textrm{ as }
k\rightarrow\infty;$
\vskip-14mm\begin{equation}
\label{prop g infinity} 
\end{equation} 
\item asymptotics at $k=\frac\i 2:$
$\exists\lim\limits_{k\to\i/2} \(g(k,\xi)-g_\r(k,\xi)\)\in\mathbb{C}.$
\vskip-14mm\begin{equation}\label{prop g i/2} \end{equation}
\end{enumerate}
% ----------------
Let us notice that $g_{\r}(k,\xi)$ and
$g_{\l}(k,\xi)$ satisfy these conditions.
Further, 

$\d g_{\r}(k,\xi)=\d g_{\l}(k,\xi)=
\(\frac{\omega}{4\(k\mp\frac{\i}{2}\)^2}+\mathrm{O}(1)\)\d k\quad
\mathrm{as}\quad k\rightarrow\frac{\pm\i}{2}.$
%\vskip-14mm
%\begin{equation}\textstyle\label{prop theta i/2}\end{equation}
% -------------

\noindent 
The suitable $g$-functions has qualitatively the same signature table as $g_{\l}$ 
(see Figures \ref{Signature table left 1234}--\ref{Signature table left(3) 5678}), 
but have different properties on the segment $[\i\c,-\i\c].$

The suitable signature table for $g$-functions in different regions of the parameter $\frac{c}{\omega}$ 
and variable $\xi$ is plotted in Figure \ref{Figure_phase_middle_function_plot}. Here $\xi_j, j=1,..,5$ are borders between 
regions with different phase function $g.$ Further, the points $\pm\i d_0\in(0,\i\c)$, 
$\pm\i d_1\in\pm(0,\i d_0)$ are the points of intersection of the lines $\Im g =0$ with the imaginary axis. (The points $\pm\i d_1$ 
appear only in the situation plotted in subgraphic $B_1$, i.e. $c>3\omega$, $\xi_1<\xi<\xi_2$. 
The points $\pm\i d_0$ 
appear in the situation plotted in subgraphics $B_1$, $C_1$, $D$.)
Points $\pm k_0$, $\pm k_1$
are points of intersection of the curve $\Im g = 0$ with the real line ($k_0$ appears in situation of subgraphics $B_2$, 
$C_2$, and $k_1$ appears in situation of subgraphics $B_2$, $C_1$. If both of them are presented in graphic, then $k_0<k_1$.
We can say that the points $\pm\i d_j, j=1,2$ are transformed into the point $\pm k_j, j=1,2$ when they 
move from the imaginary axis to the real axis.)

 Let us notice 
that the corresponding quantities $\xi_{j}, j=1,4,$ are different for different values of the parameter 
$\frac{c}{\omega}.$ Denote by $\zeta_j$ the corresponding critical values in the $x,t\geq0$ half-plane, 
$\zeta=\frac{x}{\omega t}.$ We calculate them in Lemma \ref{lemma: u_u^mod}, but for the sake of brevity, 
we write them down here. We list in ascending order of the parameter $\xi$ the qualitative description 
of asymptotics in each zone (we use functions $g_i, i=1,2,3,$ defined by (\ref{g_dg})--
(\ref{phase g function 3}),(\ref{dg1_param_eq_sys})--(\ref{dg3_param_eq_sys})):
\begin{equation}\label{xi_j_def_beg}\end{equation}
\vskip-10mm
\noindent $\bullet$ \textbf{Case $\dsfrac{c}{\omega}>3.$} Transition points are:
\\\textbf{1. } $\xi_1=\frac{-1}{4}\(\frac{c+\omega}{\omega}\)^{3/2},$ $\quad\zeta_1=\frac{3c-\omega}{4\omega}.$
\\
% ------------------------ %
\textbf{2. } $\xi_2$ is determined by the transcedental system (\ref{dg1_crit_param_eq_sys}). 
\\
$\zeta_2 = \xi_2 - 
\frac{2\i}{\omega}\lim\limits_{k\to\i/2}(g_2-g_{\r})(k,\xi_2).$ 
% ------------------------ %
\\
\textbf{3. } \textbf{$\xi_3=0,$} $\quad\zeta_3 = - 
\frac{2\i}{\omega}\lim\limits_{k\to\i/2}(g_3-g_{\r})(k,\xi_3=0).$
% ------------------------ %
\\
\textbf{4. } \textbf{$\xi_4=2\(\frac{c+\omega}{\omega}\)^{3/2},$} $\quad\zeta_4=\xi_4.$
% ----
\\
Therefore, we have a zone of fast decaying to $c$ asymptotics $\xi<\xi_1$, 
hyperelliptic (genus 2) zone $\xi_1<\xi<\xi_2$, 
two zones of elliptic asymptotics $\xi_2<\xi<\xi_3$ and $\xi_3<\xi<\xi_4$, and soliton zone $\xi_4<\xi$.  
% ------------
% ------------------------ %
% ------------------------ %
% ------------------------ %
\\
$\bullet$ \textbf{Case $1<\dsfrac{c}{\omega}<3.$}
\\
\textbf{1. } $\xi_1=\frac{-1}{4}\(\frac{c+\omega}{\omega}\)^{3/2},$ $\quad\zeta_1=\frac{3c-\omega}{4\omega}.$
% ------------------------ %
\\
\textbf{2. } $\xi_2=\frac{-2(c-\omega)}{\sqrt{\omega(c+\omega)}},$ $\quad \zeta_2=\frac{2\omega^2-c\omega+c^2}
{\omega(c+\omega)}.$
% ----------
% ------------------------ %
\\
\textbf{3. } $\xi_3=0,$ $\quad\zeta_3 = - 
\frac{2\i}{\omega}\lim\limits_{k\to\i/2}(g_3-g_{\r})(k,\xi_3=0).$
% -------------
% ------------------------ %
\\
\textbf{4. } $\xi_4=2\(\frac{c+\omega}{\omega}\)^{3/2},\quad \zeta_4=\xi_4.$ 
% -------------
% ------------------------ %
\\
Hence, we have a zone of fast decaying to $c$ asymptotics $\xi<\xi_1$, 
zone of slowly decaying to $c$ dispersive wave $\xi_1<\xi<\xi_2$, 
two zones of elliptic asymptotics $\xi_2<\xi<\xi_3$ and $\xi_3<\xi<\xi_4$, and soliton zone $\xi_4<\xi$.
\noindent
% ------------------------ %
\\
$\bullet$ \textbf{Case $0<\dsfrac{c}{\omega}<1.$}
\\
\textbf{1. } $\xi_1=\frac{-1}{4}\(\frac{c+\omega}{\omega}\)^{3/2},$ $\quad\zeta_1=\frac{3c-\omega}{4\omega}.$
\\
\textbf{2. } $\xi_2=0,$ $\zeta_2 = \frac{c}{\omega}.$
% ------------------------ %
\\
\textbf{3. } $\xi_3=\frac{-2(c-\omega)}{\sqrt{\omega(c+\omega)}},$ $\quad \zeta_3 = \frac{2\omega^2-c\omega+c^2}
{\omega(c+\omega)}.$
% ------------------------ %
\begin{equation}\label{xi_j_def_end}\end{equation}
\vskip-10mm
\noindent
\textbf{4. } $\xi_4=2\(\frac{c+\omega}{\omega}\)^{3/2},$  $\quad \zeta_4 = \xi_4.$
% ------------------------ %
\\
Similarly, we have a zone of fast decaying to $c$ asymptotics $\xi<\xi_1$, 
two zones of slowly decaying to $c$ dispersive wave $\xi_1<\xi<\xi_2$ and $\xi_2<\xi<\xi_3$, 
zone of elliptic asymptotics $\xi_3<\xi<\xi_4$, and soliton zone $\xi_4<\xi$.
\vskip5mm
\noindent In the critical cases $\frac{c}{\omega}=3,$ $\frac{c}{\omega}=1,$ $\frac{c}{\omega}=0,$ some of the zones vanish. Namely,
\\$\bullet$\textbf{Case ${\frac{c}{\omega}=3.}$} Then $\xi_1=\xi_2=\frac{-1}{4}\(\frac{c+\omega}{\omega}\)^{3/2},$ and we do not have the hyperelliptic zone of genus 2.
% ------------ %
\\$\bullet$ \textbf{Case ${\frac{c}{\omega}=1.}$} Then $\xi_2=\xi_3=0,$ and we do not have the first elliptic zone.
% ------------ %
\\$\bullet$ \textbf{Case ${\frac{c}{\omega}=0.}$} Then $\xi_3=\xi_4=0,$ and we do not have the elliptic zone.
% ------------------------ %
% ------------------------ %
% ------------------------ %
\vskip2mm 
\noindent We will try to find a phase function $g$ so, that it would exactly describe the situation shown in Figure 
\ref{Figure_phase_middle_function_plot}. For each $i =1,2,3,$ define $g=g_{i}$ as 
%\\\textbf{\textcolor{green}{Maybe change ref to picture on p.~19?}}
\begin{equation}\label{g_dg}\textstyle g_i(k,\xi)=\int\limits_{\i\c}^k\d g_i(k,\xi), \qquad 
\end{equation}
\noindent\textrm{where  } $\d g_{i}$ \textrm { are chosen as follows for the specified regions (see Figure 
\ref{Figure_phase_middle_function_plot}): }
% ---------------------------- %
% ---------------------------- %
% ---------------------------- %
\begin{enumerate}
\item case $B$ (compare with the 3$^{rd}$ graphic in Figure \ref{Signature table left 1234}):
\\
$\d g_1(k,\xi) =
\frac{\omega\xi(k^2+\mu^2)k\sqrt{(k^2+d_0^2)(k^2+d_1^2)}}{\(k^2+\frac{1}{4}\)^{2}\sqrt{k^2+\c^2}}\d
k,\quad
0<d_1<\mu<d_0<\c.$
\vskip-12mm
\begin{equation}\label{phase g function 1}\textstyle\end{equation}
% ---------------------------- %
\item case $C_1$ (compare with the 1$^{st}$ graphic in Figures \ref{Signature table
left 5678}, \ref{Signature table left(2) 5678}): 
For $\xi\neq0:$ 
\\
$\d g_2(k,\xi) \hskip-1mm=
\frac{\omega\xi(k^2+\mu_{0}^2)(k^2-k_1^2)\sqrt{(k^2+d_0^2)}}{\(k^2+\frac{1}{4}\)^{2}\sqrt{k^2+\c^2}}\d
k, \quad
0\hskip-1mm<\hskip-1mm\mu_0\hskip-1mm<\hskip-1mmd_0<\hskip-1mm\c\hskip-1mm<\hskip-1mm\frac{1}{2},\ k_1\hskip-1mm>\hskip-1mm0\hskip0mm.$
\vskip-12mm
\begin{equation}\label{phase g function 2}\textstyle\end{equation}
% ---------------------------- %
\item case $D$ (compare with the 2$^{nd}$ graphic in Figures \ref{Signature table left 5678}, \ref{Signature table left(2) 5678}, \ref{Signature
table left(3) 5678}):
For $\xi\neq0:$ 
\\
$\d
g_3(k,\xi) =
\frac{\omega\xi(k^2+\mu_{0}^2)(k^2+\mu_{1}^2)\sqrt{(k^2+d_0^2)}}{\(k^2+\frac{1}{4}\)^{2}\sqrt{k^2+\c^2}}\d
k, \ 
0\hskip-1mm<\hskip-1mm\mu_0\hskip-1mm<\hskip-1mm d_0 \hskip-1mm< \hskip-1mm\c\hskip-1mm<\hskip-1mm\frac{1}{2}\hskip-1mm<\hskip-1mm\mu_1\hskip0mm.
$
\vskip-12mm\begin{equation}\label{phase g function 3}\textstyle\end{equation}
% --------------
\end{enumerate}
\noindent Note that if $\xi=0$, then in (\ref{phase g function 2})-(\ref{phase g function 3})
$k_1=\infty$, $\mu_1=\infty$, so we will take care of this case separately.  For $\xi=0$ we set:
\\
$\d
g_{2,3}(k,\xi\hskip-1mm=\hskip-1mm0) \hskip-1mm=\hskip-1mm
\frac{\sqrt{1-4\c^2}}{(1-4\mu_0^2)\sqrt{1-4d_0^2}}\frac{\omega(k^2+\mu_{0}^2)\sqrt{k^2+d_0^2}}{\(k^2+\frac{1}{4}\)^{2}\sqrt{k^2+\c^2}}\d
k, \
0\hskip-1mm<\hskip-1mm\mu_0\hskip-1mm<\hskip-1mmd_0\hskip-1mm<\hskip-1mm\c\hskip-1mm<\hskip-1mm\frac{1}{2}.$
\vskip-13mm
\begin{equation}\label{phase g xi=0 23}\textstyle\end{equation}
%\end{enumerate}
% ----------- 
\noindent Finally, the original right and left phase functions can also be written as
\begin{enumerate}
\setcounter{enumi}{3}
\item $\d g_{\r}(k,\xi)=\frac{\omega\xi(k^2+\mu_{0,\r}^2)(k^2+\mu_{1,\r}^2)}{\(k^2+\frac{1}{4}\)^2}\d k$
with  $\mu_{0,\r}, \mu_{1,\r}$ defined in (\ref{mur01}). 

%\begin{equation}\label{mur01}\mu_{0,\r}=\frac{\i}{2}\sqrt{\frac{\xi+1-\sqrt{1+4\xi}}{\xi}}$, $\quad \mu_{1,\r}=\frac{\i}{2}\sqrt{\frac{\xi+1+\sqrt{1+4\xi}}{\xi}}.
%\end{equation}
% ------------- %
\item $\d g_{\l}(k,\xi)=\frac{\omega\xi k(k^2+\mu_{0,\l}^2)(k^2+\mu_{1,\l}^2)}{\(k^2+\frac{1}{4}\)^2\sqrt{k^2+\c^2}}\d k$
with $\mu_{0,\l}, \mu_{1,\l}$ defined in (\ref{mul01}).

%\begin{eqnarray}\nonumber\mu_{0,\l}=\frac{\i}{2}\sqrt{\frac{\xi\sqrt{1-4\c^2}+1-\sqrt{1+4\(1-4\c^2\)^{3/2}\xi}}{\xi\sqrt{1-4\c^2}}},\\\label{mul01}\mu_{1,\l}=\frac{\i}{2}\sqrt{\frac{\xi\sqrt{1-4\c^2}+1+\sqrt{1+4\(1-4\c^2\)^{3/2}\xi}}{\xi\sqrt{1-4\c^2}}}.\end{eqnarray}
\end{enumerate}
% --------- 
\subsection{Equations for the parameters of the middle phase function}
All of the functions (\ref{phase g function 1})-(\ref{phase g function 3}) must satisfy properties (\ref{prop g 1})-(\ref{prop g i/2}), and this leads to systems of equations that determine the parameters $\mu_0,$ $\mu_1,$ $d_0,$ $d_1,$ $k_1$ of the $g$-functions. Property (\ref{prop g infinity}) is satisfied automatically. To satisfy (\ref{prop g 1})-(\ref{prop g 2}) it is enough to satisfy 

$\qquad\qquad\int\limits_{\i d_0}^{\i d_1}d g_1=\int\limits_0^{\i d_0}d g_2=\int\limits_0^{\i d_0}d g_3=0.
$
\vskip-12mm
\begin{equation}\label{int_o^id}
\end{equation}
Finally, to satisfy (\ref{prop g i/2}), let us expand $\d g_i$ in a neighborhood of $k=\i/2:$ 
\\
\begin{itemize}
 \item [$i=1:$]
%\textbf{\textit{i=1:}}
$\dsfrac{\d g_1}{\d k} = \dsfrac{-\xi\omega\sqrt{1-4d_0^2}\sqrt{1-4d_1^2\ }\,(1-4\mu_0^2)}{16\sqrt{1-4\c^2}\
\(k\mp\frac{\i}{2}\)^2}\mp$
\\
$\hskip-2cm\dsfrac{\i\xi\omega
\hskip-.5mm\left[\hskip-.5mm
(\hskip-.5mm1\hskip-.5mm-\hskip-.5mm4d_0^2\hskip-.mm)\hskip-.5mm(\hskip-.5mm1\hskip-.5mm-4d_1^2\hskip-.mm)\hskip-.5mm\textcolor{blue}{(}\hskip-.5mm1\hskip-.5mm-\hskip-.5mm4\mu_0^2\hskip-.5mm-\hskip-.5mm2(\hskip-.5mm1\hskip-.5mm-\hskip-.mm4\c^2\hskip-.5mm)\hskip-.5mm\textcolor{blue}{)}\hskip-.5mm-\hskip-.5mm(\hskip-.5mm1\hskip-.5mm-\hskip-.5mm4\c^2\hskip-.mm)\hskip-.5mm(\hskip-.5mm1\hskip-.5mm-\hskip-.5mm4\mu_0^2)\hskip-.5mm(\hskip-.5mm1\hskip-.5mm-\hskip-.5mm4d_0^2\hskip-.5mm+\hskip-.5mm1\hskip-.5mm-\hskip-.5mm4d_1^2)\hskip-.5mm
\right]
}
{8\sqrt{1-4d_0^2}\sqrt{1-4d_1^2}\(1-4\c^2\)^{3/2}\
\(k\mp\frac{\i}{2}\)}$
\\
$+\mathrm{O}(1),
$
\\
% -------------
Correspondingly, the equations for the parameters are:
\begin{subequations}
\label{dg1_param_eq_sys}
\begin{eqnarray}\textstyle
&&\frac{-\xi\sqrt{1-4d_0^2}\sqrt{1-4d_1^2\ }\,(1-4\mu_0^2)}{16\sqrt{1-4\c^2}}=\frac{1}{4},
\\
&&1-4\mu_0^2=\frac{2(1-4\c^2)(1-4d_0^2)(1-4d_1^2)}{4(\c^2-d_1^2)(1-4d_0^2)-(1-4\c^2)(1-4d_1^2)},
\\
&&\int\limits_{\i d_0}^{\i d_1}\frac{k(k^2+\mu_0^2)\sqrt{k^2+d_0^2}\sqrt{k^2+d_1^2}}{\(k^2+\frac14\)^2\sqrt{k^2+\c^2}}=0.
\end{eqnarray}
\end{subequations}
% --------------
% ------------------------ %
\item[$i=2:$]
%\textbf{\textit{i}=2:} 
$\dsfrac{\d g_2}{\d k} = \dsfrac{-\xi\omega\sqrt{1-4d_0^2}\(4 k_1^2+1\)(1-4\mu_0^2)}{16\sqrt{1-4\c^2}\
\(k\mp\frac{\i}{2}\)^2}\mp$
\\
$
\hskip-5cm\frac{\i\xi\omega
\left[
\color{black}\left[\left\{4\c^2+4k_1^2+(1-4\c^2)4k_1^2\right\}(1-4d_0^2)-(1-4\c^2)(1+4k_1^2)\right]\color{black}(1-4\mu_0^2)-2(1-4\c^2)(1-4d_0^2)(1+4k_1^2)
\right]
}
{8\sqrt{1-4d_0^2}\(1-4\c^2\)^{3/2}\
\(k\mp\frac{\i}{2}\)}
$\\
$+\mathrm{O}(1).
$
\\
Correspondingly, the equations for the parameters are:
\begin{subequations}\label{dg2_param_eq_sys}
\begin{eqnarray}\textstyle
&&\frac{-\xi\sqrt{1-4d_0^2}\(4k_1^2+1\)(1-4\mu_0^2)}{16\sqrt{1-4\c^2}\
}=\frac{1}{4},
\\\nonumber
&&1-4\mu_0^2=\frac{2(1-4\c^2)(1-4d_0^2)(1+4k_1^2)}{\color{black}\left\{4\c^2+4k_1^2+(1-4\c^2)4k_1^2\right\}(1-4d_0^2)-(1-4\c^2)(1+4k_1^2)\color{black}},
\\\\
&&\int\limits_{0}^{\i d_0}\dsfrac{(k^2-k_1^2)(k^2+\mu_0^2)\sqrt{k^2+d_0^2}}{\(k^2+\frac14\)^2\sqrt{k^2+\c^2}}=0.
\end{eqnarray}
\end{subequations}
% -------------------
% ------------------------ %

\item[$i=3:$]
%\textbf{\textit{i=3:}} 
$\dsfrac{\d g_3}{\d k} = \dsfrac{\xi\omega\sqrt{1-4d_0^2}\(4\mu_1^2-1\)(1-4\mu_0^2)}{16\sqrt{1-4\c^2}\
\(k\mp\frac{\i}{2}\)^2}\pm$
\\
$\hskip-2cm\Huge{
\frac{\i\xi\omega\sqrt{1-4\c^2}
\left[
(1-4\mu_0^2)(1-4\mu_1^2)(2(1-4d_0^2)-1)-2(1-4\mu_0^2+1-4\mu_1^2)(1-4d_0^2)
\right]
}
{8\sqrt{1-4d_0^2}\(1-4\c^2\)\
\(k\mp\frac{\i}{2}\)}+\mathrm{O}(1).}
$

\noindent
Correspondingly, the equations for the parameters are:
\begin{subequations}\label{dg3_param_eq_sys}
\begin{eqnarray}\textstyle
&&\frac{\xi\sqrt{1-4d_0^2}\(4\mu_1^2-1\)(1-4\mu_0^2)}{16\sqrt{1-4\c^2}\
}=\dsfrac{1}{4},
\\
&&1-4\mu_0^2=\frac{2(1-4\c^2)(1-4d_0^2)(4\mu_1^2-1)}{\left[4\mu_1^2-1+(1-4\c^2)(4\mu_1^2+1)\right](1-4d_0^2)-(1-4\c^2)(4\mu_1^2-1)}
,\nonumber\\\\
&&\int\limits_{0}^{\i d_0}\frac{(k^2+\mu_0^2)(k^2+\mu_1^2)\sqrt{k^2+d_0^2}\d k}{\(k^2+\frac14\)^2\sqrt{k^2+\c^2}}=0.
\end{eqnarray}
\end{subequations}
% --------
\noindent For $\xi=0$ the equations for the parameters of the function\\ $g_2(k,\xi=0)\equiv g_3(k,\xi=0)$ defined by (\ref{phase g xi=0 23}) are
% ---------
\begin{subequations}\label{dg23_xi=0_param_eq_sys}\begin{eqnarray}
%&&\dsfrac{\xi\sqrt{1-4d_0^2}\(4\mu_1^2-1\)(1-4\mu_0^2)}{16\sqrt{1-4\c^2}\
%}=\dsfrac{1}{4},
%\\\nonumber\\\nonumber
&&1-4\mu_0^2=\dsfrac{2(1-4\c^2)(1-4d_0^2)}{(2-4\c^2)(1-4d_0^2)-(1-4\c^2)}
,\\
&&\int\limits_{0}^{\i d_0}\dsfrac{(k^2+\mu_0^2)\sqrt{k^2+d_0^2}\d k}{\(k^2+\frac14\)^2\sqrt{k^2+\c^2}}=0.
\end{eqnarray}
\end{subequations}

\noindent 
For the initial phase functions we have
%\begin{enumerate}
%\setcounter{enumi}{3}
\item[Case of $g_{\r}:$]
%\\
%\textbf{4.}
$\dsfrac{\d g_{\r}(k,\xi)}{\d k} = \dsfrac{\xi\omega\(1-4\mu_0^2\)\(4\mu_1^2-1\)}{16
\(k\mp\frac{\i}{2}\)^2}\mp$
\\
$
\dsfrac{\i\xi\omega
\left[
(1-4\mu_0^2)(1-4\mu_1^2)-2(1-4\mu_0^2+1-4\mu_1^2)
\right]
}
{8\(k\mp\frac{\i}{2}\)}+\mathrm{O}(1).
$
\\
Here $\mu_0=\mu_{0,\r},\quad \mu_1=\mu_{1,\r}$ are defined in (\ref{mur01}).
\\
Correspondingly, the equations for the parameters are 
\begin{equation*}\textstyle
\left\{
\begin{array}{l}\dsfrac{\xi\(1-4\mu_0^2\)\(4\mu_1^2-1\)}{16
}=\frac{1}{4},
\\
1-4\mu_0^2=\dsfrac{2(4\mu_1^2-1)}{4\mu_1^2+1}.
\end{array}
\right.
\end{equation*}
% ------------------------ %
\item[Case of $g_{\l}:$]
%\textbf{5.} 
$\dsfrac{\d g_{\l}(k,\xi)}{\d k} = \frac{\xi\omega(1-4\mu_0^2)\(4\mu_1^2-1\)}{16\sqrt{1-4\c^2}\
\(k\mp\frac{\i}{2}\)^2}
\mp$
\\
$
\mp\dsfrac{\i\xi\omega
\left[
(1-4\mu_0^2)(1-4\mu_1^2)-2(1-4\c^2)(1-4\mu_0^2+1-4\mu_1^2)
\right]
}
{8\(1-4\c^2\)^{3/2}\
\(k\mp\frac{\i}{2}\)}+\mathrm{O}(1).
$
\\
Here $\mu_0=\mu_{0,\l},\quad \mu_1=\mu_{1,\l}$ are defined in (\ref{mul01}).
\\
Correspondingly, equations for the parameters are
\begin{equation*}\textstyle
\left\{
\begin{array}{l}\dsfrac{\xi(1-4\mu_0^2)\(4\mu_1^2-1\)}{16\sqrt{1-4\c^2}\
}=\frac{1}{4},
\\
1-4\mu_0^2=\dsfrac{2(1-4\c^2)(4\mu_1^2-1)}{2(1-4\c^2)+4\mu_1^2-1}.
\end{array}
\right.
\end{equation*}
% --------------- 
\end{itemize}
\noindent 
%\textit{Because of difficulty of the above equations for the parameters of phase functions, maybe it is worth to look what we can get from assumption that $\d g=\Omega_1(k,\xi)\xi+\Omega_2(k,\xi),$ where $\Omega_1$ is regular at the points $k=\pm\i/2.$}

\subsection{Transitions between regions with different phase functions}
Now let us take a look at the values $\xi_j, j=1,2,3,4,$ which separate regions of different phase functions. 
These values depend on whether $\frac{c}{\omega}$ lies in the interval $(3,\infty),$ $(1,3)$, or $(0,1).$ 
The analysis of systems (\ref{dg1_crit_param_eq_sys}), (\ref{dg2_crit_param_eq_sys}) is presented in 
subsection \ref{SubSect_Exist_param}.
%Since the diversity of different phase functions in the case $c/\omega>3$ covers all the other cases, it is enough to write down the equations for the parameters of the phase functions in this case.
% ------------
\\
\textbf{Case 1.}
 Transition from $\left\{\xi<\xi_1\right\}$ to $\left\{\xi_1<\xi<\xi_2\right\}$ in the case $c/\omega>3$. 
% $g_{\l}(k,\xi)$ to $g_1(k,\xi).$
\\
This case is characterized by changing of the phase function $g_{\l}(k,\xi)$ with $g_1(k,\xi)$; it occurs when
$d_0=d_1=\mu_0.$
Equations (\ref{dg1_param_eq_sys}) give us $1-4d_0^2=4(1-4\c^2)$ and $\xi_1=\frac{-1}{4}\(\frac{c}{\omega}+1\)^{3/2}.$
\\
\textbf{Case 2.} Transition from $\left\{\xi_1<\xi<\xi_2\right\}$ to $\left\{\xi_2<\xi<\xi_3\right\}$
in the case $c/\omega>3$. %$g_1(k,\xi)$ to $g_2(k,\xi).$
\\
This case is characterized by changing of the phase function $g_{1}(k,\xi)$ with $g_2(k,\xi)$; it occurs when
 $d_1=0.$
Equations (\ref{dg2_param_eq_sys}) give us 
% ----------- 
\begin{subequations}
\label{dg1_crit_param_eq_sys}
\begin{eqnarray}\textstyle
&&\dsfrac{-\xi_2\sqrt{1-4d_0^2}\,(1-4\mu_0^2)}{4\sqrt{1-4\c^2}}=1,
\\
&&1-4\mu_0^2=\dsfrac{2(1-4\c^2)(1-4d_0^2)}{4\c^2(1-4d_0^2)-(1-4\c^2)},
\\
&&\int\limits_{\i d_0}^{0}\dsfrac{k^2(k^2+\mu_0^2)\sqrt{k^2+d_0^2}}{\(k^2+\frac14\)^2\sqrt{k^2+\c^2}}=0.
\end{eqnarray}
\end{subequations}
% ----------
%The same analysis as in subsection \ref{SubSect_Exist_param} shows that this system is uniquely solvable, and therefore determine the value $\xi_2.$

%\color{red} Remark. We see that the denominator in (\ref{dg1_crit_param_eq_sys}b) should be positive. So, $d_0<\frac12\sqrt{2-\frac{1}{4\c^2}}.$ Why???
%\color{black}
% ----------
\noindent
\textbf{Case 3.} Transition from $\left\{\xi_2<\xi<\xi_3\right\}$ to $\left\{\xi_3<\xi<\xi_4\right\}$ in the case $c/\omega>1.$
%$g_{2}(k,\xi)$ to $g_3(k,\xi).$
\\
This case is characterized by changing of the phase function $g_{2}(k,\xi)$ with $g_3(k,\xi)$; it occurs when
$k_1=\infty.$
Equations (\ref{dg2_param_eq_sys}) give us \begin{subequations}\label{dg2_crit_param_eq_sys}
\begin{eqnarray}\textstyle
&&\xi_3=0,
\\
&&1-4\mu_0^2=\dsfrac{2(1-4\c^2)(1-4d_0^2)}{\color{black}\(1+1-4\c^2\)(1-4d_0^2)-(1-4\c^2)\color{black}},\\
&&\int\limits_{0}^{\i d_0}\dsfrac{(k^2+\mu_0^2)\sqrt{k^2+d_0^2}}{\(k^2+\frac14\)^2\sqrt{k^2+\c^2}}=0.
\end{eqnarray}
\end{subequations}
%The above system (\ref{dg2_crit_param_eq_sys}) determines uniquely the quantities $\mu_0$, $d_0.$

%\color{red} 
%Remark. Denominators in equations (\ref{dg1_param_eq_sys}b), (\ref{dg2_param_eq_sys}b), (\ref{dg3_param_eq_sys}b), (\ref{dg1_param_eq_sys}b), (\ref{dg2_crit_param_eq_sys}b)
%should be poitive. Why???
%\color{black}

%\textcolor{red}{STOP. I finished here.}

%We distinguish the following sectors where we use different phase functions:
%\\\textbf{$\frac{c}{\omega}>3$. $\hat\xi<\frac{-1}{4}$}
%$g(k,\xi):=\hat\theta$; (Graphics \ref{Signature table left 1234})

\subsection{Unique solvability of the systems for the parameters of $g$-functions.}\label{SubSect_Exist_param}
In this subsection we discuss the solvability of the above systems (\ref{dg1_param_eq_sys}), (\ref{dg2_param_eq_sys}), (\ref{dg3_param_eq_sys}), (\ref{dg1_crit_param_eq_sys}), (\ref{dg2_crit_param_eq_sys}),  for the parameters of the phase functions.
% ---------------- 
%For the sake of brevity we denote
%
%$
%b_0:=1-4d_0^2;\ 
%b_1:=1-4d_1^2; \
%m_0:=1-4\mu_0^2;$
%
%$ 
%m_1:=4\mu_1^2-1;\ 
%n_1:=1+4 k_1^2;\
%a:=1-4\c^2.
%$
\begin{lem}
Each of the systems (\ref{dg1_param_eq_sys}), (\ref{dg2_param_eq_sys}), (\ref{dg3_param_eq_sys}), (\ref{dg1_crit_param_eq_sys}), (\ref{dg2_crit_param_eq_sys}) has a unique solution that satisfies $0\leq d_1\leq\mu_0\leq d_0<\c<\frac12.$
\end{lem}

\begin{proof}
Let us treat, for instance, systems (\ref{dg1_crit_param_eq_sys}) and (\ref{dg1_param_eq_sys}), 
the other ones can be treated in a similar way.
% ----------- 
\\
$\bullet$ Consider first system (\ref{dg1_crit_param_eq_sys}). Equation (\ref{dg1_crit_param_eq_sys}b) 
determines $1-4\mu_0^2$ as a decreasing function of $1-4d_0^2$, which has a vertical asymptote at 

$1-4d_0^2=\dsfrac{1-4\c^2}{4\c^2}=\dsfrac{\omega}{c},\ \ 
\textrm{ (see Figure \ref{Graphics_existence_of_c_param}, left) and we have that}\\$ 
$1-4\mu_0^2=\dsfrac{2}{\frac{c}{\omega}-1}$ 
as $1-4d_0^2=1.$
% ------- 

On the other hand, in the same manner as in \cite{Kotlyarov_Minakov_2010}, we can show that equation 
(\ref{dg1_crit_param_eq_sys}c) determines $\mu_0$ as an increasing function of $d_0$, or $1-4\mu_0^2$ as 
an increasing function of $1-4 d_0^2$, with $$1-4\mu_0^2=1 \textrm{ as } 1-4d_0^2=1.$$
% --------- 
We have: $1-4\mu_0^2=\frac{2}{\frac{c}{\omega}-1}<1$ provided $\frac{c}{\omega}>3,$ hence, 
(\ref{dg1_crit_param_eq_sys}b)--(\ref{dg1_crit_param_eq_sys}c) determine a unique pair $\mu_0\leq d_0$, 
where the equality is attained only when $\mu_0=d_0=0.$
Further, substituting those $\mu_0$, $d_0$ in (\ref{dg1_crit_param_eq_sys}a), we determine $\xi_2<0.$
% ---------- 

%\begin{figure}[ht!]
%\begin{minipage}[0.48\linewidth]
%\begin{center}
%\includegraphics[scale=0.8]{Graphic_existence_of_param-eps-converted-to}
%\end{center}
%\caption{Signature table for $\Im\theta(k,\xi)$ at
%$\dsfrac{-1}{4}<\xi<0.$} \label{Signature table right 4}
% ------------ %
%\end{minipage}
%\caption{Chart of system (\ref{dg1_crit_param_eq_sys}b)-(\ref{dg1_crit_param_eq_sys}c).}
%\label{Graphic_existence_of_param}
%\end{figure}

\begin{figure}
\begin{minipage}[ht!]{.4\linewidth}
%\begin{center}
\vskip-3.1cm\hskip-.0010mm
\includegraphics[scale=0.5]{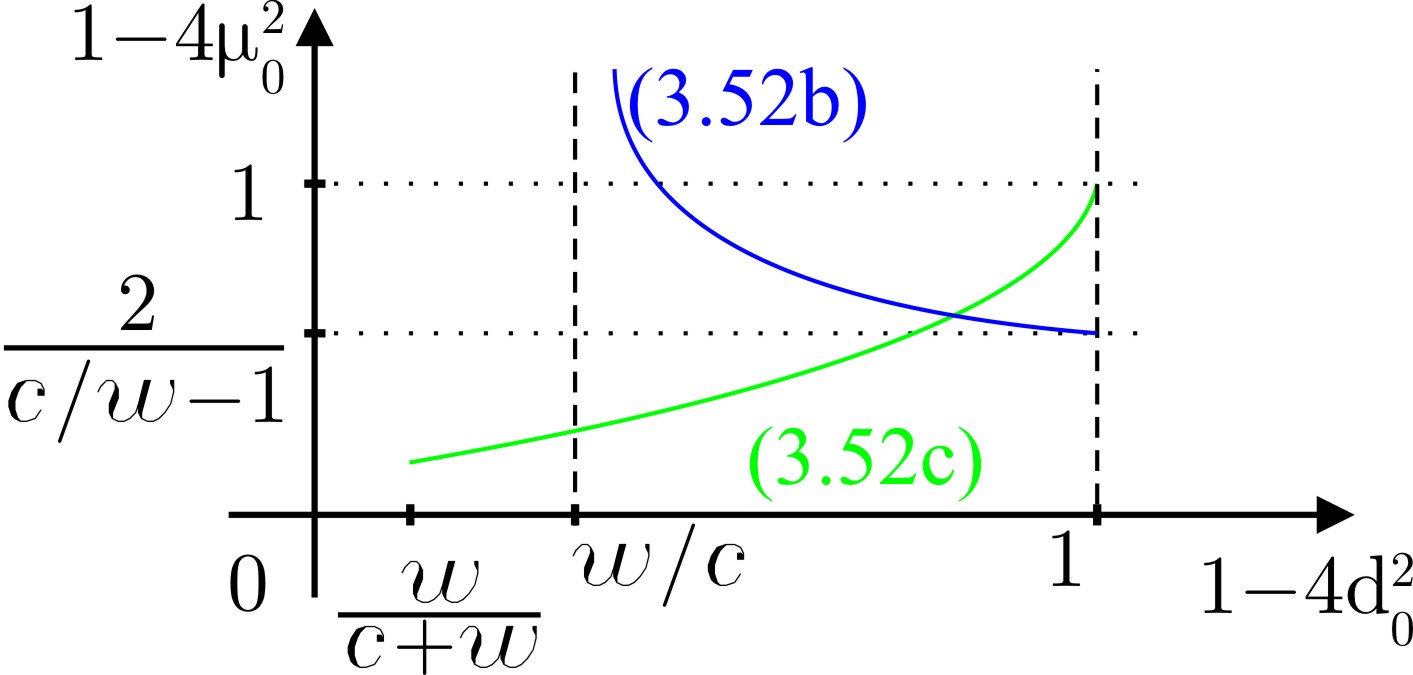}
%\end{center}
\end{minipage}
% -----------------------
\begin{minipage}[ht!]{.4\linewidth}
%\begin{center}
\vskip-3.1cm\hskip10mm
\includegraphics[scale=0.5]{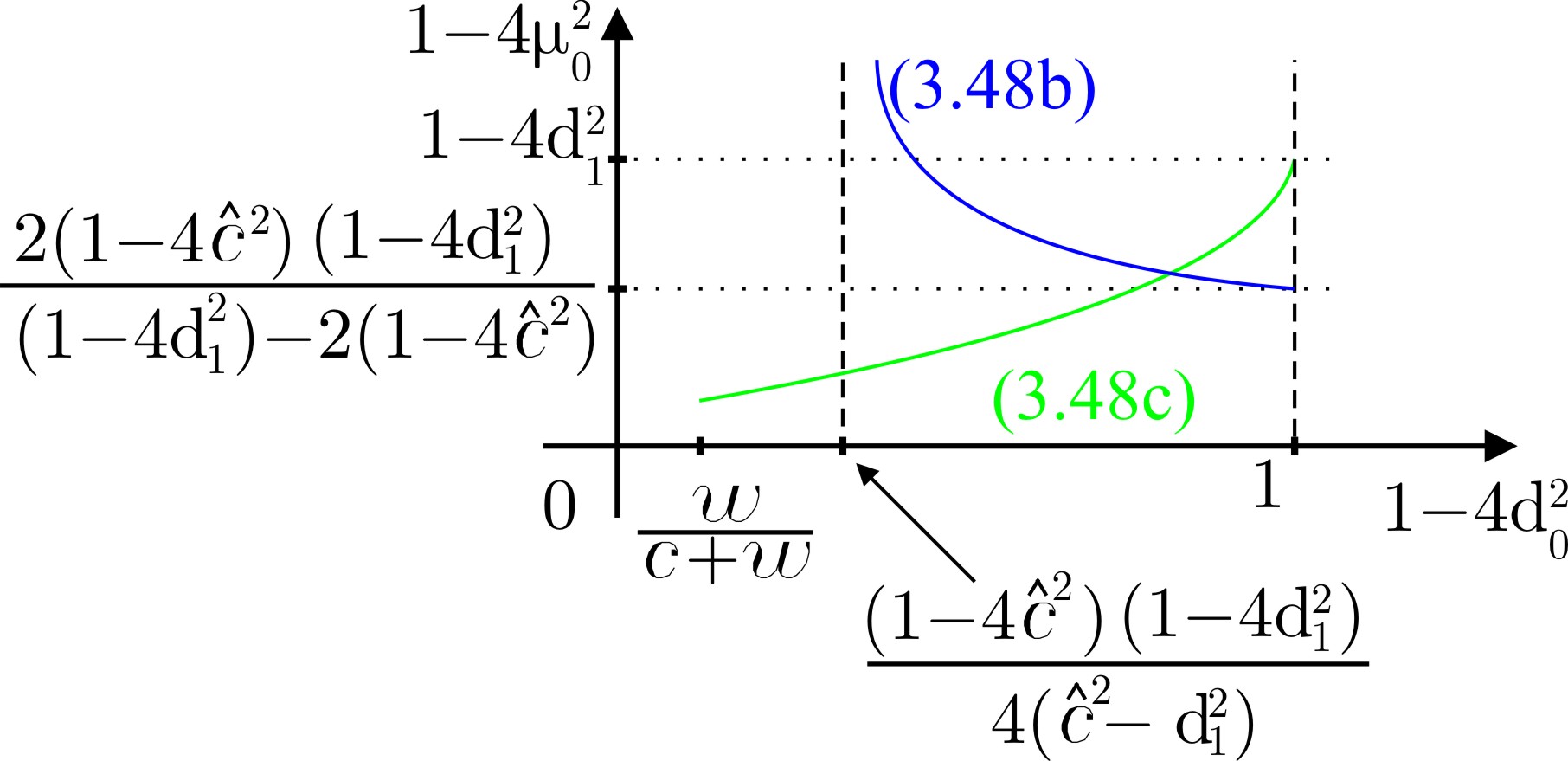}
%\end{center}
\end{minipage}

\caption{Chart of system (\ref{dg1_crit_param_eq_sys}b)-(\ref{dg1_crit_param_eq_sys}c) (left) and 
(\ref{dg1_param_eq_sys}b)-(\ref{dg1_param_eq_sys}c) (right).}
\label{Graphics_existence_of_c_param}
\end{figure}

$\bullet$ Now consider system (\ref{dg1_param_eq_sys}). First we consider system 
(\ref{dg1_param_eq_sys}b)-(\ref{dg1_param_eq_sys}c) as a system of variables $d_0, \mu_0$ with a parameter $d_1.$ 
% ----------- 
Equation (\ref{dg1_param_eq_sys}b) determines $1-4\mu_0^2$ as a decreasing function of $1-4d_0^2$, which 
has a vertical asymptote at 

$1-4d_0^2=\dsfrac{(1-4\c^2)(1-4 d_1^2)}{4(\c^2-d_1^2)},\ \ $ (see Figure \ref{Graphics_existence_of_c_param}, 
right) and we have that

$1-4\mu_0^2=\dsfrac{2(1-4\c^2)(1-4d_1^2)}{(1-4d_1^2)-2(1-4\c^2)}$ as $1-4d_0^2=1-4 d_1^2.$
% ----------- 

\noindent
On the other hand, in the same manner as in \cite{Kotlyarov_Minakov_2010}, we can show that equation 
(\ref{dg1_param_eq_sys}c) determines $\mu_0$ as an increasing function of $d_0$, or $1-4\mu_0^2$ as an 
increasing function of $1-4 d_0^2$, with 

$1-4\mu_0^2=1-4d_1^2$ as $1-4d_0^2=1-4d_1^2.$ 
% ----------

\noindent We have

\noindent  $$\dsfrac{(1-4\c^2)(1-4 d_1^2)}{4(\c^2-d_1^2)}<1-4d_1^2 \textrm{  and }
\dsfrac{2(1-4\c^2)(1-4d_1^2)}{(1-4d_1^2)-2(1-4\c^2)}<1-4d_1^2$$

\noindent provided 
$4(1-4\c^2)<1-4d_1^2.$

\noindent
If the latter inequality is satisfied, then (\ref{dg1_param_eq_sys}b)--(\ref{dg1_param_eq_sys}c) 
determine a unique pair $\mu_0\leq d_0$, where the equality is attained only when 
% -------------
$$d_1=\mu_0=d_0 \ \ \ \ \ \ \textrm{     and     }\ \ \ \ \ \
4(1-4\c^2)=1-4d_1^2.$$
% ------------
Further, substituting those $\mu_0$, $d_0$ in (\ref{dg1_param_eq_sys}a), we determine $d_1$, which 
satisfies $0\leq d_1<\mu_0$ by virtue of (\ref{dg1_param_eq_sys}c).

Hence, the only thing which remains to do is to establish the inequality $$4(1-4\c^2)<1-4d_1^2$$ inside of the zone. 
It is satisfied in the 
critical case when $d_1=0.$ Suppose the contrary, $4(1-4\c^2)=1-4d_1^2.$ In this case (\ref{dg1_param_eq_sys}b)
can be rewritten as 
$$
(1-4\mu_0^2)\left[(1-4d_0^2)-(1-4d_1^2)\right]+2(1-4d_0^2)\left[(1-4\mu_0^2)-(1-4d_1^2)\right]=0,
$$ 
and since $d_1\leq\mu_0\leq d_0,$ we obtain $d_1=\mu_0=d_0,$ and hence we have the opening of the hyperelliptic zone. This 
contradiction shows that inside of the zone we have $4(1-4\c^2)<1-4d_1^2$.

%Let us notice that each of equations (\ref{dg1_param_eq_sys}b), (\ref{dg2_param_eq_sys}b), (\ref{dg3_param_eq_sys}b) implies that $\mu_0$ a a function of $d_0$ is a decreasing function. On the other hand, in the same manner as in \cite{Kotlyarov_Minakov_2010}, we can show that each of equations (\ref{dg1_param_eq_sys}c), (\ref{dg2_param_eq_sys}c), (\ref{dg3_param_eq_sys}c) implies that $\mu_0$ is an increasing function as a function of variable $d_0.$ That is why, there is just one pair $(\mu_0,d_0)$ depending on $d_1$ that satisfies (b), (c) equations. Finally, from equation (a) we find $d_1.$
\end{proof}

%\begin{figure}\label{Figure_phase_middle_function}
%\vskip-4cm\hskip-4.cm\includegraphics
%[scale=0.95]
%[width=1.6\textwidth, natwidth=1210, natheight=1642]
%{IMG_0006v2.jpg}
%\end{figure}

\begin{figure}
\begin{minipage}[ht!]{1.\linewidth}
\vskip-2.7cm\hskip-10mm\includegraphics
[scale=0.8]
%[width=1.6\textwidth, natwidth=1210, natheight=1642]
{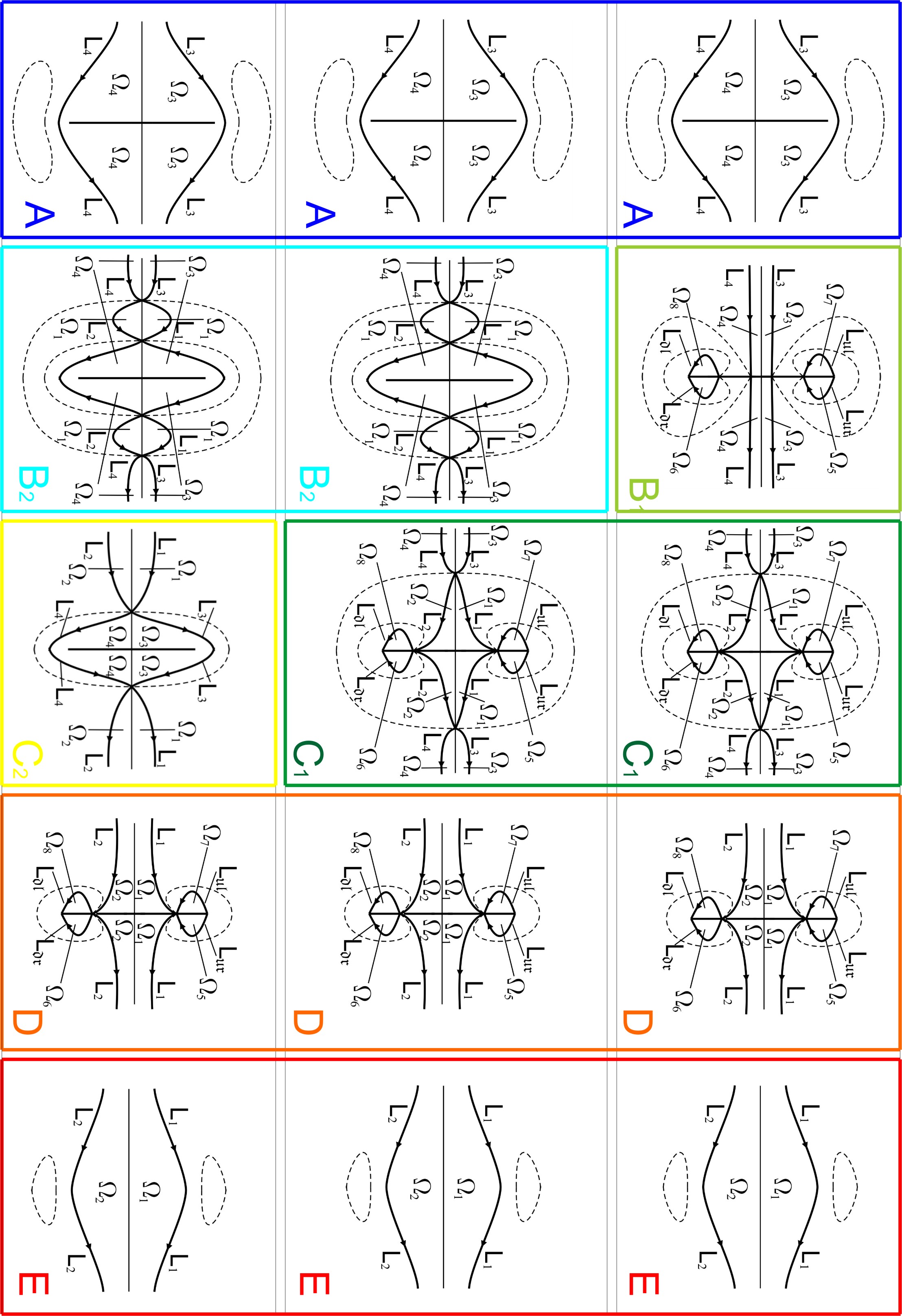}%-eps-converted-to}
\end{minipage}
\caption{Contour transformations for the RH problem in different regions of the parameter $\xi.$}
\label{Figure_cont_middle_function}
\end{figure}

\section{Riemann-Hilbert problem transformations}\label{sect_RH problem transformations}
%\subsection{Elliptic-wave region \color{red} $-6c^2t<x<4c^2t$}

%\subsubsection{Transition from initial R-H problem to the model one.}

In this section we present a series of transformations of the original RH problem from Lemma \ref{lem_RH}
in the domains $\xi_1<\xi<\xi_4$ in the spirit of the nonlinear steepest descent method \cite{DZ93}, which leads to either 
trivial 
model problems or model problems explicitly solvable in terms of elliptic (genus 1) or hyperelliptic (genus 2) functions. 

\noindent
\textbf{Step 1.}\\
First, we exchange the phase function $g_{\r}(y,t;k)$ with $g$-function $g(k,\xi)$ constructed in Section \ref{sect_phase function}:
\[V_{\r}^{(1)}(\xi,t,k)=V_{\r}(\xi,t,k)e^{i(t g(k,\xi)-
g_{\r}(y,t;k))\sigma_3}.\] 

For the brevity of exposition we keep the same letter $g$ for all the cases of variable $\frac{c}{\omega}$, $\xi$ shown in 
Figure \ref{Figure_phase_middle_function_plot}. The actual meaning of $g(k,\xi)$ is as follows:
\begin{itemize}
 \item in cases $A, B_2, C_2$ we take $g(k,\xi)=g_{\l}(k,\xi)$;
 \item in case $B_1$ we take $g(k,\xi)=g_{1}(k,\xi)$;
 \item in case $C_1$ we take $g(k,\xi)=g_{2}(k,\xi)$;
 \item in case $D$ we take $g(k,\xi)=g_{3}(k,\xi)$;
 \item in case $E$ we take $g(k,\xi)=g_{\r}(k,\xi)$.
\end{itemize}

Then the jump matrix and pole conditions for $V_{\r}^{(1)}$ have the same form as for $V_{\r}$ with the phase function 
being changed from $g_{\r}(y,t;k)$ to $tg(k,\xi)$ everywhere in formulas (\ref{J_R})-(\ref{M infinity}); the jump matrices on the intervals 
$\(\i\c,-\i\c\)$ are read now as follows:
\begin{equation}\nonumber
J_{\mathfrak{r}}(x,t;k)=\begin{pmatrix}\e^{\i t(g_-(k,\xi)-g_+(k,\xi))}&0\\ f(k)\,\e^{\i t(g_-(k,\xi)+g_+(k,\xi))}&\e^{-\i t(g_-(k,\xi)-g_+(k,\xi))}\end{pmatrix},\quad
k\in\(\i\c,0\);
\end{equation}
% --------------- %
\begin{equation}\label{J1_0_-ic}
J_{\mathfrak{r}}(x,t;k)=\begin{pmatrix}\e^{\i t(g_-(k,\xi)-g_+(k,\xi))}&\ol{ f(\ol k)}\,\e^{-\i t(g_-(k,\xi)+g_+(k,\xi))}\\0&\e^{-\i t(g_-(k,\xi)-g_+(k,\xi))}\end{pmatrix},\quad
k\in\(0,-\i\c\);
\end{equation}
\noindent
\textbf{Step 2.}
\textbf{Regular Riemann--Hilbert problem.}\\ Following ideas from \cite{Kostenko_Shepelsky_Teschl_2009}, we
transform our meromorphic RH problem into a holomorphic one. In order to
achieve this, we define the function $V^{(2)}_{\mathfrak{r}}(x,t;k)$
as follows:
\begin{equation}\label{hat M}
V^{(2)}_{\mathfrak{r}}(x,t;k)=\left\{
\begin{array}{l}
V^{(1)}_{\mathfrak{r}}(x,t;k)\begin{pmatrix}1&0\\\dsfrac{-\i\gamma_{+,j}^2\e^{2\i
t g(\xi,\i\kappa_j)}}{k-\i\kappa_j}&1\end{pmatrix},\qquad
|k-\i\kappa_j|<\varepsilon,
\\
V^{(1)}_{\mathfrak{r}}(x,t;k)\begin{pmatrix}1&\dsfrac{\i\gamma_{+,j}^2\e^{2\i
t g(\xi,\i\kappa_j)}}{k+\i\kappa_j}\\0&1\end{pmatrix}, \qquad
|k+\i\kappa_j|<\varepsilon,\\
V^{(1)}_{\mathfrak{r}}(x,t;k),\qquad \textrm{ elsewhere},
\end{array}
\right.
\end{equation}
where $\varepsilon>0$ is a sufficiently small number such that
circles $|k-\i\kappa_j|=\varepsilon$ do not intersect and lie in the domain
$\Im z(k)>0.$

The vector-valued function $ V^{(2)}_{\mathfrak{r}}(x,t;k)$ 
solves the
following RH problem: find a sectionally holomorphic function
$V^{(2)}_{\mathfrak{r}}(x,t;k)$ which satisfies the following conditions: \begin{enumerate}
\item $V^{(2)}_{\mathfrak{r}}(x,t;.)$ is holomorphic away from the
contour $\Sigma_{\mathfrak{r}}$ and circles $C_j:=\left\{k:\
|k-\i\kappa_j|=\varepsilon\right\},$ $\overline C_j:=\left\{k:\
|k+\i\kappa_j|=\varepsilon\right\}.$ The orientation on $C_j$,
$\overline C_j$ is counterclockwise, so the positive side is
inside the circles.
It satisfies the same symmetry conditions and has the same asymptotics for large $k$ as $V^{(1)}(\xi,t;k);$ 
\item The jump condition $
V^{(2)}_{\mathfrak{r},-}(\xi,t;k)= V^{(2)}_{\mathfrak{r},+}(\xi,t;k)
J^{(2)}_{\mathfrak{r}}(\xi,t;k)$ is satisfied, where
\[J^{(2)}_{\mathfrak{r}}(\xi,t;k)\equiv J^{(1)}_{\mathfrak{r}}(\xi,t;k),\quad k\in\Sigma_{\mathfrak{r}},\]
\begin{equation}\label{J2 Cj}
J^{(2)}_{\mathfrak{r}}(\xi,t;k)=
\begin{pmatrix}1&0\\\dsfrac{\i\gamma_{+,j}^2\e^{2\i
tg(\i\kappa_j,\xi)}}{k-\i\kappa_j}&1\end{pmatrix}, \quad k\in C_j;
\end{equation}
\begin{equation}\label{J2 Cj overline}
J^{(2)}_{\mathfrak{r}}(\xi,t;k)=
\begin{pmatrix}1&\dsfrac{-\i\gamma_{+,j}^2\e^{2\i
tg(\i\kappa_j,\xi)}}{k+\i\kappa_j}\\0&1\end{pmatrix}, \quad k\in
\overline C_j.
\end{equation}
\end{enumerate}

% -------------------------------------------------------

\noindent \textbf{Step 3.} Following \cite{Kostenko_Shepelsky_Teschl_2009}, to deal with poles we
introduce the function 
$$\Lambda (k) = \prod\limits_{j=1}^N\frac{k+\i\kappa_j}{k-\i\kappa_j} $$ and 
the following transformation of the RH problem, which
acts inside the disks surrounding the poles:
$V^{(3)}_{\mathfrak{r}}(x,t;k) = V^{(2)}_{\mathfrak{r}}(x,t;k)
D(k),$ where
\[D(k) = \left\{\begin{array}{l}
\begin{pmatrix}
1 & \frac{k-\i\kappa_j}{-\i\gamma_{+,j}^2\ \e^{2\i tg(\i\kappa_j,\xi)}}\\
\frac{\i\gamma_{+,j}^2\ \e^{2\i tg(\i\kappa_j,\xi)}}{k-\i\kappa_j} & 0
\end{pmatrix}
\begin{pmatrix}\dsfrac{1}{\Lambda (k)}&0\\0&\Lambda (k)\end{pmatrix},\quad
|k-\i\kappa_j|<\varepsilon,
\\
\begin{pmatrix}
0 & \frac{-\i\gamma_{+,j}^2\ \e^{2\i tg(\i\kappa_j,\xi)}}
{k+\i\kappa_j}
\\
\frac{k+\i\kappa_j}{\i\gamma_{+,j}^2\ \e^{2\i tg(\i\kappa_j,\xi)}} & 0
\end{pmatrix}
\begin{pmatrix}\dsfrac{1}{\Lambda (k,\xi)}&0\\0&\Lambda (k,\xi)\end{pmatrix},\quad
|k+\i\kappa_j|<\varepsilon,
\\
\Lambda^{-\sigma_3}(k,\xi)
%\begin{pmatrix}\dsfrac{1}{\Lambda (k,\xi)}&0\\0&\Lambda (k)\end{pmatrix}
,\
\rm{ elsewhere}.
\end{array}\right.\]
It is easy to notice that $D(k)$ is a piecewise holomorphic
function and does not have poles at the points $\pm\i\kappa_j, \
j=1,...,N$.
\\
\noindent
Now the jump matrix
$J^{(3)}_{\mathfrak{r}}(x,t;k)=(D_+(k,\xi))^{-1}J^{(2)}_{\mathfrak{r}}(x,t;k)
D_-(k,\xi)$ 
\begin{eqnarray*}
J^{(3)}_{\mathfrak{r}}(x,t;k)&=& \begin{pmatrix}1 &
\frac{(k-\i\kappa_j)\Lambda ^2(k)}{\i\gamma_{+,\j}^2\e^{2\i
tg(\i\kappa_j,\xi)}}\\0&1\end{pmatrix}, |k-\i\kappa_j|=\varepsilon,
\\
& = & \begin{pmatrix}1 & 0\\
\frac{(k+\i\kappa_j)\Lambda ^{-2}(k)}{-\i\gamma_{+,\j}^2\e^{2\i
tg(\i\kappa_j,\xi)}}&1\end{pmatrix}, |k+\i\kappa_j|=\varepsilon,
\end{eqnarray*}
% ---------- %
%\begin{eqnarray*}
%&=& \begin{pmatrix}1 & \ol{r(k)}\Lambda ^{2}(k)\e^{-2\i tg} \\
%-r(k)\Lambda ^{-2}(k)\e^{2\i tg} &
%1-|r(k)|^2\end{pmatrix},
%k\in\mathbb{R}\setminus\left\{0\right\},
%\\
%&=& \begin{pmatrix}\e^{\i t(g_--g_+)} & 0 \\
% f(k)\Lambda ^{-2}(k)\e^{\i t(g_-+g_+)} & \e^{-\i t(g_--g_+)}
%\end{pmatrix},
%k\in\(0,\i\c \),
%\\
%&=& \begin{pmatrix}\e^{\i t(g_--g_+)} & \ol{ f(\ol
%k)}\Lambda ^{2}(k)\e^{-\i t(g_-+g_+)} \\ 0 & \e^{-\i t(g_--g_+)}
%\end{pmatrix},
%k\in\(0,-\i\c\).
%\end{eqnarray*}

\noindent on the circles around the points $\i\kappa_j$ is exponentially 
small for large $t$, and therefore, they do not contribute in the leading term of the asymptotics. 

We recall (\ref{prop g 1})-(\ref{prop g 2}) that $$g_--g_+ = B \textrm{ for } k\in(\i d_0,\i d_1)\cup(-\i d_1,-\i d_0),\ \ \ \ \textrm{ and }$$
$$g_-+g_+=0 \textrm{ for } k\in(\i\c,\i d_0)\cup(-\i d_0,-\i\c)\cup(\i d_1,-\i d_1).$$ 
This holds for all the regions where the above 
quantities have sense. In the case when $d_1$ is not defined, we substitute it with $0$ in the above relations 
(see also Figure \ref{Figure_phase_middle_function_plot}).
%                 
%               
% -----------------------------------------------------
% 
\\\\        
\noindent \textbf{Step 4.}

\noindent On different parts of the real line we need a lower/upper or an upper/lower triangular factorization of 
the jump matrix.
Thee lower/ upper triangular factorization is 
$$
J^{(3)} = \begin{pmatrix}1&0\\-r(k)\Lambda^{-2}(k)\e^{2\i t g}\end{pmatrix}\begin{pmatrix}1&\ol{r(k)}\Lambda^2(k)\e^{-2\i t g}\\0&1\end{pmatrix}.$$
To provide an upper/lower triangular factorization of the jump matrix supported on the real line, 
we introduce the $\delta$- transformation:
a new function $$V^{(4)}(\xi,t;k)=V^{(3)}(\xi,t;k)\delta^{-\sigma_3}
$$
has a jump matrix
$$ \quad J^{(4)}=\delta_+^{\sigma_3}J^{(3)}\delta_-^{-\sigma_3}.$$
% ----------- 
Then on the real axis 
\[J^{(4)}\hskip-.5mm(\hskip-.5mm\xi\hskip-.5mm,\hskip-.5mmt\hskip-.5mm;\hskip-.5mmk\hskip-.5mm)\hskip-1.mm=\hskip-1.5mm\begin{pmatrix}\hskip-1.5mm1& \frac{k}{z(k)}a\hskip-.5mm(\hskip-.5mmk\hskip-.5mm)\hskip-.5mm\overline{b(\overline{k})}\delta_+^2\hskip-.5mm(\hskip-.5mmk\hskip-.5mm)\hskip-.5mm\Lambda\hskip-.5mm^2\hskip-.5mm(\hskip-.5mmk\hskip-.5mm)\hskip-.5mm
\e^{-2\i tg}\\0&1\hskip-1.5mm\end{pmatrix}\hskip-2.5mm
\begin{pmatrix}1&0\\\frac{-k}{z(k)}b(k)\overline{a(\overline{k})}\Lambda\hskip-.5mm^{-2}\hskip-.5mm(\hskip-.5mmk\hskip-.5mm)\hskip-.mm
\delta_-^{-2}\hskip-.5mm(\hskip-.5mmk\hskip-.5mm)\e^{2\i tg}&1
\hskip-.5mm\end{pmatrix},
\]
% -------- 
providing \[\frac{\delta_+}{\delta_-}=1-|r(k)|^2, \quad k\in(-\infty,-k_1)\cup(-k_0,k_0)\cup(k_1,\infty),
% --------- 
\]
% ----------
where (we refer to $A, B_j, C_j, D, E, j=1,2,$ as to subgraphics of Figure \ref{Figure_phase_middle_function_plot}) 
$$
k_1=\begin{cases}
0, \textrm{  for } A, B_1, 
\\
k_1(\xi), \textrm{  for } B_2, C_1, 
\\
+\infty, \textrm{  for } C_2, D, E,
\end{cases}
\textrm{ and }
k_0=\begin{cases}
0, \textrm{  for } A, B_1, C_1, D, E,
\\
k_0(\xi), \textrm{  for } B_2, C_2.
\end{cases}
$$

% ------- 
In the cases $D, E$, this $\delta$-transformation is trivial, $\delta(k)\equiv 1.$
% ----- 
\\
For $k_1\in (0,+\infty)$ we have 
% -----
\begin{equation}\label{delta g_2,3,r}\delta(k,\xi)=\exp\left(\frac{1}{2\pi\i}\left\{
\int\limits_{-\infty}^{-k_1(\xi)}+\int\limits_{-k_0(\xi)}^{+k_0(\xi)}+\int\limits_{k_1(\xi)}^{+\infty}\right\}
\dsfrac{\log\(1-|r(s)|^2\)\d s}{s-k}\right).\end{equation}
% ------- 

%In the case $k_1=0, $ we define $\delta$ to have an additional discontinuity on the interval $(\i\c,-\i\c).$ Due to this we skip further the $F-$ transformation.

%\textbf{\textcolor{green}{What does that mean? C B}}

%% -----------

%Let us notice that for $k\in\(\i\c,-\i\c\):$
%\begin{equation}\label{delta Delta}\frac{\delta_+}{\delta_-}=\e^{\i\Delta},\quad \Delta:=\begin{cases}-2\arg a_+(k)-\frac{\pi}{2},\quad k\in(\i\c,0),\\-2\arg a_+(\ol{k})-\frac{\pi}{2},\quad k\in(0,-\i\c).\end{cases}\end{equation}
%% ----------- 
The jump matrix on the interval $(\i \c,-\i \c)$ is as follows:
% ----------- 
\begin{eqnarray}\label{Jump3.5}
J^{(4)}(\xi,t;k)&=& \begin{pmatrix}\dsfrac{\delta_+}{\delta_-}\e^{\i t(g_--g_+)} & 0 \\
\dsfrac{ f(k)\e^{\i t(g_-+g_+)}}{\Lambda^2(k)\delta_+\delta_-} & \dsfrac{\delta_-}{\delta_+}\e^{-\i t(g_--g_+)}
\end{pmatrix},
k\in\(\i\c, 0\),
\\\nonumber
=&& \begin{pmatrix}\dsfrac{\delta_+}{\delta_-}\e^{\i t(g_--g_+)} & \ol{ f(\ol
k)}\Lambda^{2}(k)\delta_+\delta_-\e^{-\i t(g_-+g_+)} \\ 0 & \e^{-\i t(g_--g_+)}\dsfrac{\delta_-}{\delta_+}
\end{pmatrix},
k\in\(0,-\i\c\).
\end{eqnarray}
% -----------------------------------------------------

\noindent \textbf{Step 5.}\\
Here we deal with oscillation behavior of jump matrix on the real line. Using triangular factorizations provided in Step 4, 
we
remove oscillations from $\R$, moving exponentials from jump matrices to those domains, where they are small.
The function $g(k,\xi)$ and its imaginary part (see the signature
table of $\Im g(k,\xi)$ at Figure \ref{Figure_phase_middle_function_plot} on 
p.~\pageref{Figure_phase_middle_function_plot})
%(\ref{Im g=0})) 
suggest
a choice of a new contour $\Sigma_2$ for RH problem (see Figure \ref{Figure_cont_middle_function} on p.~\pageref{Figure_cont_middle_function}).
% Let $L_{\pm}$ be some curves that lie in the domains $\Omega_1$, $\Omega_2$, respectively, have asymptotes 
%$\{k: k=s\pm\varepsilon, -\infty<s<\infty\}$ with a positive $\varepsilon,$ and contain the points $\pm\i d(\xi),$ respectively. Then $\Sigma_2=\R\cup[-ic,ic]\cup L_+\cup L_-$.
We use the standard lower-upper and upper-lower factorization 
%(\ref{flu1}) 
of the jump
matrix on the real axis (see for example \cite{DZ93}, formulas (0.11) and (0.23)) and apply the transformation

$V^{(5)}(\xi,t,k)=V^{(4)}(\xi,t,k)G^{(5)}(\xi,t,k),$ where
\begin{equation}
G^{(5)}(\xi,t,k)=\begin{cases}\left(\begin{array}{ccc}1&0\\-r(k)\Lambda^{-2}(k)\delta^{-2}(k)
e^{2itg(k,\xi)}&1
\end{array}\right),\quad k\in\Omega_1,\\
\left(\begin{array}{ccc}1&-\overline{r(\overline{k})}\Lambda^2(k)\delta^2(k)e^{-2itg(k)}\\0&1
\end{array}
\right),\quad k\in\Omega_2,\\
\begin{pmatrix}1& \dsfrac{k}{z(k)}a(k)\overline{b(\overline{k})}\delta^2(k)\Lambda^2(k)
\e^{-2\i tg}\\0&1\end{pmatrix},\quad k\in\Omega_3,\\
\begin{pmatrix}1&0\\\dsfrac{k}{z(k)}b(k)\overline{a(\overline{k})}\Lambda^{-2}(k)
\delta^{-2}(k)\e^{2\i tg}&1
\end{pmatrix},\quad k\in\Omega_4,\\
I,\quad \rm{elsewhere},
\end{cases}
\end{equation}
where the domains $\Omega_j, j=1,2,3,4$ are indicated at Figure \ref{Figure_cont_middle_function}. (In certain 
cases 
$A, B_1, B_2, C_1, C_2, D, E$ illustrated in Figure \ref{Figure_cont_middle_function}, some of the domains $\Omega_j, j=1,4$ 
are absent and the value of $G^{(5)}$ is not defined there, in other words, we skip the definition of $G^{(5)}$ 
for those domains).
 Thus, 
we get the new RH problem:
% ---------------
$$V^{(5)}_-(\xi,t,k)=V^{(5)}_+(\xi,t,k)J^{(5)}(\xi,t,k),
\qquad
V^{(5)}(\xi,t,k)\rightarrow (1,1),\quad k\rightarrow\infty
$$
% ---------------- 
with the following jump matrices:
% ------------------
$$
J^{(5)}=G^{(5)}|_{\Omega_1},  k\in L_1\cup L_3;\ \ \ \ 
J^{(5)} =\(G^{(5)}|_{\Omega_2}\)^{-1},   k\in L_2\cup L_4;\ \ \ \
J^{(5)}
=I, k\in\R.
$$
% ------------------ 
% ---------- 
On the interval $(\i\c,-\i\c)$ the jump matrix is different for different situations:
\\
--for $C_1, D, E$ %$g_2,g_3,g_{\r}$ 
we have

$
J^{(5)}=
\e^{ \i t (g_--g_+)\sigma_3},
 k\in (\i d_0,-\i d_0);\
 J^{(5)}=J^{(4)}, k\in(\i\c,\i d_0)\cup(-\i d_0,-\i\c);
$
\\
% ----------- 
--for the case $A$, $B_1$, $B_2$, $C_2$ 
%$g_{\l}, g_1$ 
we have
\begin{eqnarray}
\nonumber J^{(5)}(\xi,t,k)&=&
\begin{pmatrix}0&-\dsfrac{\delta_+\delta_-\Lambda^2\e^{-\i t(g_-+g_+)}}{ f(k)}\\\dsfrac{ f(k) \e^{\i t(g_-+g_+)}}{\delta_-\delta_+\Lambda^2}&0\end{pmatrix},
 k\in (\i\ *,0),\\\nonumber
&=&
\begin{pmatrix}0&\ol{ f(\ol{k})}\Lambda^2\delta_-\delta_+\e^{-\i t(g_-+g_+)}\\\dsfrac{-\e^{\i t(g_-+g_+)}}{\ol{ f(\ol{k})}\Lambda^2(k)\delta_-\delta_+}&0\end{pmatrix},
 k\in (0, -\i\ *),
\\\label{Jump4}  &=&J^{(4)}, k\in(\i\c,\i d_1)\cup(-\i d_1,-\i\c),\nonumber\end{eqnarray}
where instead of $*$ we substitute $d_1$ in case $B_1$, and substitute $\c$ in cases $A, B_2, C_2.$

%\begin{figure}[ht]
%\begin{center}
%\epsfig{width=100mm,figure=Contourelliptic1-eps-converted-to}
%\end{center}
%\caption{The contour $\Sigma_2$ for the R-H problem $M^{(2)}(\xi,t,k)$}
%\label{Contourelliptic1}
%\end{figure}

\noindent These expressions for the jump matrix on the interval $(\i \c,-\i \c)$ are calculated using 
(\ref{prop g 1})-(\ref{prop g i/2}),
% properties 1-5 of the function $g(k,\xi)$
the jump relation $$f(k)=r_-(k)-r_+(k),\quad k\in(0,\i\c)$$ and
the symmetry relation $\quad r(k)=\overline{r(\overline{-k})}.$

\vskip0.5cm

\noindent\textbf{Step 6.} Opening lenses (concerns only the cases 
$B_1,$ $C_1$, $D$ (see Figure \ref{Figure_cont_middle_function})).
%see the lines $L_{\mathfrak{ur}}$, $L_{\mathfrak{ul}}$, 
%$L_{\mathfrak{dr}}$, $L_{\mathfrak{dl}}$
\\
Here we deal with oscillation jump matrices on the interval $\i\c,-\i\c,$ providing triangular factorization as we did for 
the real line in Step 5.
First we note that the function $f(k)=\frac{z(k+0)}{ka(k-0)a(k+0)}$ has the following analytic
continuation from the  interval $(-\i\c,\i\c)$:
% ------------

$f(k)=\widehat{f}_+(k),\quad k\in(-\i\c,\i\c),\qquad\mathrm{where}\qquad\widehat{f}(k)=
\dsfrac{z(k)}{ka(k)\overline{b\(\overline{k}\)}}.
$
\vskip-14mm\begin{equation}
\end{equation}
% --------------- 
Following ideas from \cite{Kotlyarov_Minakov_2010},  \cite{KM12}, we factorize the jump matrix $J^{(5)}$ on
$(i\c,\i d_0)\cup(-\i\c,-\i d_0)$ as follows:
% ------------- 
\\
$J^{(5)}\hskip-1mm=\hskip-1mm
F_+^{-\sigma_3}\hskip-1mm\begin{pmatrix}1&\displaystyle\frac
{F_+^2\Lambda^2\delta_+^2e^{-2itg_+}}{\widehat{f}_+(k)}\\0&1
\end{pmatrix}
\hskip-1mm
\begin{pmatrix}0&-i\\-i&0 \end{pmatrix}
\hskip-1mm
\begin{pmatrix}1&\displaystyle\frac
{-F_-^2\Lambda^2\delta_-^2\e^{-2itg_-}}{\widehat{f}_-(k)}\\0&1
\end{pmatrix}\hskip-2mmF_-^{\sigma_3}$
\vskip-15mm\begin{equation}\label{factorization_openinig_lenses_up}
\end{equation}
% ---------
\vskip3mm
\noindent
for $k\in (\i\c,\i d_0)$ and
\\
$J^{(5)}=
F_+^{-\sigma_3}\begin{pmatrix}1&0\\\displaystyle\frac
{e^{2itg_+(k,\xi)}}{
\ol{\widehat{f}_+(\overline{k})}\Lambda^2F_+^2\delta_+^2}&1
\end{pmatrix}\begin{pmatrix}0&i\\i&0\end{pmatrix}
\begin{pmatrix}1&0\\\displaystyle\frac
{-e^{2itg_-(k,\xi)}}{\ol{\widehat{f}_-(\overline{k})}\Lambda^2F_-^2\delta_-^2}&1
\end{pmatrix}F_-^{\sigma_3},
$
\vskip-15mm
\begin{equation}\label{factorization_openinig_lenses_down}\end{equation}
\vskip3mm\noindent
for $k\in (-\i\c,-\i d_0)$.
\\
Direct calculations show that it is
possible if $F(k,\xi)$ satisfies the relation
$  F_-(k,\xi)  F_+(k,\xi)=\begin{cases}\dsfrac{\i  f(k)}{\Lambda^{2}\delta_+\delta_-},\quad
k\in(\i\c,\i d_0),\\\displaystyle\frac{i}{\ol{ f(\ol{k})}\Lambda^2\delta_+\delta_-},\quad
k\in(-\i\c,-\i d_0).\end{cases}
$

In case of $B_1$ we will also require these relations to be hold in the intervals $(\i d_1,0)$ and $(0,-\i d_1).$
%\vskip-15mm
%\begin{equation}{\label{F+F-}}
%\nonumber\end{equation}
%\vskip3mm

%For the cases when we use the phase functions $g_{\l}$ (as shown in Figure on p.~ \pageref{Figure_phase_middle_function_plot}), the function $\delta$ is defined by formula (\ref{delta g_l1}), and the rhs of the above formulas are equal to 1. Hence $F\equiv 1.$

%In the cases of the phase functions $g_1,g_2,g_3,$ 
We define 
$F(k,\xi)=$
%\\
%$
\begin{equation}\label{F}\hskip-1mm\exp\hskip-1.5mm\left\{\hskip-2mm\dsfrac{\mathrm{w}\hskip-.5mm(\hskip-.5mmk\hskip-.5mm)\hskip-.5mm}{2\hskip-.5mm\pi\hskip-.5mm\i}\hskip-2mm
\left[\hskip-1mm\int\limits_{\i \c}^{\i d_0}\hskip-1.5mm+\hskip-1.5mm\int\limits_{\i d_1}^{0}\right]\hskip-2mm \frac{\log\hskip-.5mm\frac{i f}{\delta^2\Lambda^{2}}}
{s-k}\frac{\d s}{\mathrm{w\hskip-.5mm_+}(\hskip-.5mms\hskip-.5mm)} 
\hskip-1mm+\hskip-1mm
\dsfrac{\mathrm{w}\hskip-.5mm(\hskip-.5mmk\hskip-.5mm)\hskip-.5mm}{2\hskip-.5mm\pi\hskip-.5mm\i\hskip-.5mm}\hskip-1.5mm
\left[\hskip-1.mm\int\limits_{0}^{-\i d_1}\hskip-1.5mm+\hskip-1.5mm\int\limits_{-\i d_0}^{-\i \c}\hskip-.mm\right]\hskip-2.mm \frac{\log\frac{\i}{\overline{f}\delta^2\Lambda^{2}}}{s-k}\frac{\d s}{\mathrm{w\hskip-.5mm_+}\hskip-.5mm(\hskip-.5mms\hskip-.5mm)\hskip-.5mm} 
\hskip-1mm+\hskip-1mm
\dsfrac{\mathrm{w}\hskip-.5mm(\hskip-.5mmk\hskip-.5mm)\hskip-.5mm}{2\hskip-.5mm\pi\hskip-.5mm\i\hskip-.5mm}
\hskip-2.mm\left[\int\limits_{\i d_0}^{\i d_1}\hskip-1.mm+\hskip-2.mm\int\limits_{-\i d_1}^{-\i d_0}\right] \hskip-1.5mm\frac{\i\Delta}{\hskip-.5mms\hskip-.5mm-\hskip-.5mmk\hskip-.5mm}\frac{\d s}{\mathrm{w}\hskip-.5mm(\hskip-.5mms\hskip-.5mm)\hskip-.5mm}\right\}\hskip-1mm,
\end{equation}
%$
%\begin{equation}\label{F}
%\end{equation}
where we substitute $d_1$ with $0$ for the cases $C_1$ and $D$.
Here $$\w(k)=\sqrt{(k^2+\c^2)(k^2+d_0^2)(k^2+d_1^2)}\quad \mbox{ for }\quad B_1, \quad \mbox{and}$$
$$\w(k)=\sqrt{(k^2+\c^2)(k^2+d_0^2)}\quad \mbox{ for } \quad C_1, D.$$
To remove the essential singularity of $F(k,\xi)$ at infinity, we set 

$\Delta:=\begin{cases}\i\(\left[\int\limits_{\i \c}^{\i d_0}+\int\limits_{\i d_1}^{0}\right]\log\frac{i f}{\delta^2\Lambda^2}\frac{s\d s }{\w_+(s)}\) 
\(\int\limits_{\i d_0}^{\i d_1}\frac{s\d s }{\w(s)}\)^{-1} \mbox{ for } B_1,
\\
\i\(\int\limits_{\i \c}^{\i d_0}\log\frac{i f}{\delta^2\Lambda^2}\frac{\d s }{\w_+(s)}\) 
\(\int\limits_{\i d_0}^{0}\frac{\d s }{\w(s)}\)^{-1}  \mbox{ for } C_1, D.
\end{cases}
$
\vskip-20mm
\begin{equation}\label{Delta}
\end{equation}
\vskip8mm

\begin{com}\label{deltaF=1}
Studying conjugation properties of $\delta(k,\xi)F(k,\xi)$ on $\R\cup(\i\c,-\i\c)$, it is easy to see that in the case $A$ (see Figure \ref{Figure_phase_middle_function_plot})
we have 
\begin{equation}\label{delta g_l1}\delta(k)F(k)=\begin{cases}\frac{1}{a(k)\Lambda(k)}\sqrt{\frac{z(k)}{k}},\quad \Im\ z(k)>0,
\\\\
\frac{\ol{a(\ol{k})}}{\Lambda(k)}\sqrt{\frac{k}{z(k)}},\quad \Im \ z(k)<0.
\end{cases}
\end{equation}
\end{com}

\noindent
Function $F(k,\xi)$ has the following properties as a function of variable $k$:
\\
$\bullet$ $ F(k,\xi)$ is analytic outside the contour $[\i\c,-\i\c]$ (since by definition (\ref{F}) $F$ is a 
Cauchy type integral
over contour $[\i\c,-\i\c]$);
\\
$\bullet$ $ F(k,\xi)$ does not
vanish (since $F$ is an exponential function);
\\
$\bullet$
$ F(k,\xi)\rightarrow 1$ as $k\rightarrow\infty$ (this is straightforward by expanding $F$ as $k\to\infty$ and 
taking into account (\ref{Delta}));
\\
$\bullet$
$ \frac{F_+(k,\xi)}{F_-(k,\xi)} =  \e^{\i \Delta},\quad k\in(\i d_0,\i d_1)\cup(-\i d_1,-\i d_0),$
with $\Delta=\Delta(\xi)$ defined in (\ref{Delta}) (this is due to Sokhotski-Plemelj formula applied to (\ref{F})).

%\begin{figure}[ht]
%\begin{center}
%\epsfig{width=100mm,figure=Contourelliptic2-eps-converted-to}
%\end{center}
%\caption{The contour $\Sigma_3$ for the R-H problem $M^{(3)}(\xi,t,k)$}
%\label{Contourelliptic2}
%\end{figure}

In the case $B_1,$
%of the phase function $g_1$, 
we have also to factorize the jump matrix on $(\i d_0,\i d_1)\cup(-\i d_1,-\i d_0):$
% -------------
\begin{subequations}\label{fact_lens_g1}
\begin{eqnarray}
J^{(5)}\hskip-1mm=\hskip-1mm
\begin{pmatrix}
1&0\\\dsfrac{-r_+\e^{2\i tg_+}}{\Lambda^2\delta_+^2}&1
\end{pmatrix}
\hskip-1mm
\e^{(\i t B+\i\Delta)\sigma_3}
\hskip-1mm
\begin{pmatrix}
1&0\\\dsfrac{r_-\e^{2\i tg_-}}{\Lambda^2\delta_-^2}&1
\end{pmatrix}, k\in(\i d_0,\i d_1),
\end{eqnarray}
% --------------------
\begin{eqnarray}
J^{(5)}\hskip-1.5mm=\hskip-1.5mm
\begin{pmatrix}
\hskip-1mm
1 \hskip-1mm &\hskip-1mm -\ol{r_+}\Lambda^2\delta_+^2\e^{\hskip-.2mm-\hskip-.2mm2\hskip-.2mm\i\hskip-.2mm t\hskip-.2mmg_+}\hskip-3mm
\\
\hskip-1mm 0\hskip-1mm &\hskip-1mm 1
\hskip-3mm
\end{pmatrix}
\hskip-1.5mm\e^{\hskip-.5mm(\hskip-.5mm\i t B+\i\Delta \hskip-.5mm)\sigma_3}
\hskip-2mm
\begin{pmatrix}\hskip-1mm
1&\ol{r_-}\Lambda\hskip-.5mm^{2}\hskip-.5mm\delta_-^2\e^{2\i tg_-}\\\hskip-1mm 0&1
\hskip-2mm\end{pmatrix}\hskip-1mm, k\hskip-1mm\in\hskip-1mm(\hskip-1mm-\i d_1,-\i d_0).
\nonumber\\
\end{eqnarray}
\end{subequations}
%Here we used the identity $\delta_+\delta_-^{-1}=\e^{\i\Delta}$ (\ref{delta Delta}).
% -------- 
\\
By using the factorizations (\ref{factorization_openinig_lenses_up}), (\ref{factorization_openinig_lenses_down}), we get
the following RH problem:
% ------------------\\
$$V^{(6)}=V^{(5)}G^{(6)}(\xi,t,k),\quad
V^{(6)}_-=V^{(6)}_+J^{(6)}(\xi,t,k),
\quad
V^{(6)}\rightarrow (1,1),\quad k\rightarrow\infty,
$$
% ----------------- \\
where
\begin{eqnarray*}
G^{(6)}(\xi,t,k)&=&F^{-\sigma_3}(k,\xi)\left(\begin{array}{ccc}1&
\displaystyle
\frac{F^2(k,\xi)\Lambda^2\delta^2e^{-2itg(k,\xi)}}{\widehat{f}(k)}\\0&1
\end{array}
\right), 
\qquad k\in\Omega_5\bigcup\Omega_7,\\
&=&F^{-\sigma_3}(k,\xi)\left(\begin{array}{ccc}1&0\\\displaystyle\frac{e^{2itg(k,\xi)}}
{\overline{\widehat{f}(\overline{k})}\Lambda^2\delta^2F^2(k,\xi)}&1\end{array}
\right),
\qquad k\in\Omega_6\bigcup\Omega_8,\\
&=&F^{-\sigma_3}(k,\xi),
\qquad k\notin\Omega_5\bigcup \Omega_6\bigcup\Omega_7\bigcup\Omega_8,
\end{eqnarray*}
and
\begin{eqnarray*}
J^{(6)}(\xi,t,k)=&F^{\sigma_3}G^{(6)}|_{\Omega_7},\ k\in L_{\mathfrak{ur}}, \quad
=\(G^{(6)}|_{\Omega_5}\)^{-1}F^{-\sigma_3},\  k\in L_{\mathfrak{ul}},\\
=&F^{\sigma_3}G^{(6)}|_{\Omega_8},\ k\in L_{\mathfrak{dr}}, \quad
=\(G^{(6)}|_{\Omega_5}\)^{-1}F^{-\sigma_3},\  k\in L_{\mathfrak{dl}},\\
\end{eqnarray*}
% --------- 
$J^{(6)}\hskip-1mm=\hskip-1mm\left(\hskip-1mm\begin{array}{ccc}1&0\\\displaystyle\frac{-r(k)e^{2itg(k,\xi)}}{F^2(k,
\xi)\Lambda^2\delta^2}&1
\end{array}\hskip-1mm\right), k\in L_1;\ \ 
=\hskip-1mm\left(
\hskip-1mm
\begin{array}{ccc}1&\displaystyle\frac{\overline{r(\overline{k})}F^2(k,
\xi)\Lambda^2\delta^2}
{e^{2itg(k,\xi)}}\\0&1
\end{array}\right), k\in L_2,
$
\\
$J^{(6)}\hskip-1.mm=\hskip-1.5mm\begin{pmatrix}\hskip-1.5mm1& \dsfrac{k}{z(k)}\frac{a\hskip-.5mm(\hskip-.5mmk\hskip-.5mm)\hskip-.5mm\overline{b(\overline{k})}\delta_+^2\hskip-.5mm(\hskip-.5mmk\hskip-.5mm)\hskip-.5mm\Lambda\hskip-.5mm^2\hskip-.5mm(\hskip-.5mmk\hskip-.5mm)\hskip-.5mm}
{\e^{2\i tg}}\\0&1\hskip-1.5mm\end{pmatrix}\hskip-.5mm,\hskip-.5mm k\hskip-1mm\in \hskip-1mm L_3;
\hskip1.5mm=\hskip-1.5mm
\begin{pmatrix}1&0\\\dsfrac{k}{z(k)}\dsfrac{b(k)\overline{a(\overline{k})}\e^{2\i tg}}{\Lambda\hskip-.5mm^{2}\hskip-.5mm(\hskip-.5mmk\hskip-.5mm)\hskip-.mm
\delta_-^{2}\hskip-.5mm(\hskip-.5mmk\hskip-.5mm)}&1
\hskip-1.mm\end{pmatrix}\hskip-1.mm, k \hskip-1mm\in \hskip-1mm  L_4,
$
\\
$J^{(6)}=\e^{(itB(\xi)+i\Delta(\xi))\sigma_3},
k\hskip-1mm\in \hskip-1mm\pm(\i d_0,\i d_1),\
=\hskip-2mm\left(\hskip-2mm\begin{array}{ccc}0&\mp i\\\mp i&0\end{array}\hskip-1mm\right)\hskip-1mm,
k\in\(\pm \i\c,\pm \i d_0\)\cup(\pm\i d_1,0).
$

\noindent 
Exceptionally in the case $B_1$ we have also, in view of (\ref{fact_lens_g1}), jumps on the lenses $L_{2\mathfrak{ul}}$,
$L_{2\mathfrak{ur}}$, $L_{2\mathfrak{dl}}$, $L_{2\mathfrak{dr}}$ surrounding the intervals
$(\i d_0,\i d_1)$ and $(-\i d_1,-\i d_0)$ (we denote the corresponding domains as $\Omega_{2\mathfrak{ur}}$, 
$\Omega_{2\mathfrak{ul}}$, $\Omega_{2\mathfrak{dr}}$, $\Omega_{2\mathfrak{dl}}$): 
% -------- 
\[G^{(6)}=\begin{pmatrix}1&0\\-\dsfrac{r\e^{2\i tg}}{\Lambda^2\delta^2}&1\end{pmatrix},\ k\in\Omega_{2\mathfrak{ur}}
\cup\Omega_{2\mathfrak{ul}},
\quad
=\begin{pmatrix}1&\ol{r}\Lambda^2\delta^2\e^{-2\i tg}\\0&1\end{pmatrix},\ k\in\Omega_{2\mathfrak{dr}}\cup
\Omega_{2\mathfrak{dl}}.
\]
% -----------
Then $J^{(6)}=\begin{cases}G^{(6)}|_{L_{2\mathfrak{ur}}},& k\in L_{2\mathfrak{ur}},
\\\(G^{(6)}\)^{-1}|_{L_{2\mathfrak{ul}}},& k\in L_{2\mathfrak{ul}},
\\G^{(6)}|_{L_{2\mathfrak{dr}}},& k\in L_{2\mathfrak{dr}},
\\\(G^{(6)}\)^{-1}|_{L_{2\mathfrak{dl}}},& k\in L_{2\mathfrak{dl}}.\end{cases}$

\noindent\textbf{Step 7.} (Hyper-) elliptic model problem. \\
The jump matrix $J^{(6)}$ is close to the identity matrix everywhere except for the interval $(\i\c,-\i\c).$

We preceded the estimation of the contribution of the contours $L_j, j=1,2,5,6,7,8,$ which will be done in the next step, 
 by solving of the 
model problem with the jump matrix supported only on the interval $(\i\c,-\i\c).$

\begin{com}
The contribution to the asymptotics of the points $\pm\i d_j$, $\pm k_j$, $j=0,1,$ can be estimated by constructing 
the appropriate parametrices in the vicinities of these points. Parametrices in the vicinities of the points $\pm k_j$, 
$j=0,1,$
can be constructed in terms of the parabolic cylinder functions (cf. \cite{DZ93}), while the parametrices in the 
vicinities of the points $\pm\i d_j$, $j=0,1,$ can be constructed in terms of the Airy functions (cf. \cite{Bleher11}, 
section 4.2.4.2). However, in the case of step-like initial data, the input of the parametrices is not of the leading 
order,
 and the leading order asymptotics is already quite cumbersome. Hence, we confine ourselves to a rough estimation 
of the input o
f the 
 points $\pm\i d_j$, $\pm k_j$, $j=0,1,$ without calculation of the explicit formulas.
\end{com}

%\textbf{\textcolor{green}{This is rather brief!}}

%{aaa!!! references to parametrices!!!}
Hence, we introduce a model problem 
$$V^{(mod)}_-=V^{(mod)}_+J^{(mod)},\quad V^{(mod)}\rightarrow I\quad\textrm{as}\quad k\rightarrow\infty, \ \ \ \textrm{ where}$$ 
\begin{equation}\label{model problem}
J^{(mod)}(\xi,t,k)=\begin{cases}
e^{(itB(\xi)+i\Delta(\xi))\sigma_3},\quad k\in(\i d_0,\i d_1)\cup(-\i d_1,-\i d_0)\\
\left(\begin{array}{ccc}0&\mp i\\\mp i&0\end{array}\right),\quad
k\in(\pm \i\c,\pm \i d_0)\cup(\pm\i d_1,0).
\end{cases}
\end{equation}
%\color{blue}
\color{black}\textbf{Remark.} In the cases of the phase functions $g_{\l}, g_1,$ after the $\delta$-transformation we get a 
singularity at the origin. This singularity is preserved up to the model problem. Nevertheless, going back in the 
train of transformations of the RH problem, which we did in Steps 1-6 of the present 
Section \ref{sect_RH problem transformations}, and noticing that the singularity at 0 of the model problem and the one of 
$V^{(6)}$ are 
of the same order, since there are any lenses opening at the origin, 
we conclude that we are allowed to deal with this singular RH problem because coming back to $V$ we get an ordinary solution.
\color{black}

\textbf{a. Elliptic situation.} \color{black} To solve the model problem we distinguish the cases $d_1=0, d_0>0$ 
(elliptic case), 
$d_0>d_1>0$ (hyperelliptic case of genus 2), and the trivial model 
problem which leads to the constant asymptotics $(d_1=d_0=0)$.

The solution of the vector elliptic model problem is very similar to one introduced in 
\cite{{Egorova_Gladka_Kotlyarov_Teschl}}. Consider the Riemann surface of the function 

$\hskip2cm\mathrm{w}^2(k)=(k^2+\c^2)(k^2+d_0^2),$
\\
 with cuts along
$(i\c,\i d_0)$ and $(-\i d_0,-i\c)$, which is fixed by the condition $\mathrm{w}(0)>0$ on the first sheet. The $a$-cycle and
 $b$-cycle are introduced as in Figure \ref{Graphics_Riemann_surfaces}a.

%\textbf{\textcolor{green}{What picture? Missing!}}

\begin{figure}
\begin{minipage}[ht!]{.4\linewidth}
%\begin{center}
\vskip-3.1cm\hskip-.0010mm
\includegraphics[scale=0.7]{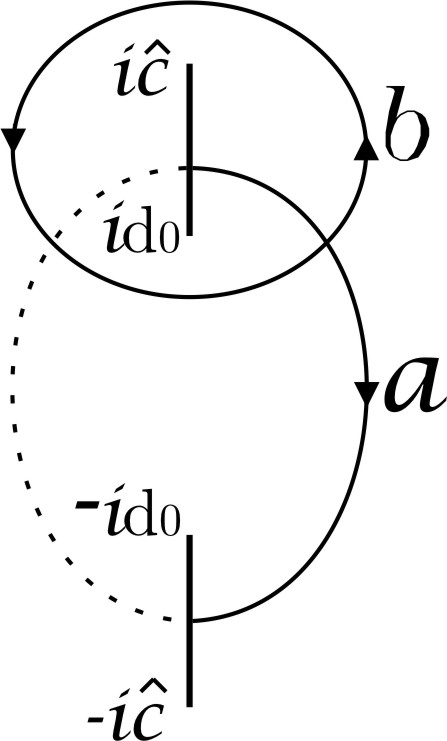}
\centerline{\hskip-3cm(a)}%\end{center}
\end{minipage}
% -----------------------
\begin{minipage}[ht!]{.4\linewidth}
%\begin{center}
\vskip-3.1cm\hskip10mm
\includegraphics[scale=0.7]{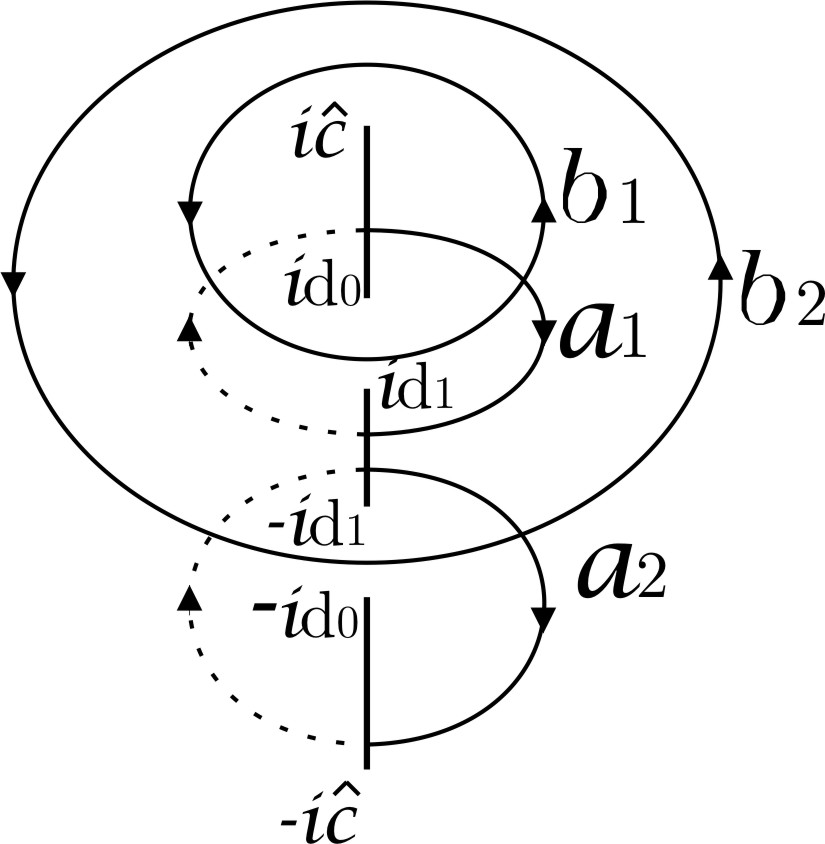}
\centerline{(b)}
%\end{center}
\end{minipage}
\caption{Riemann surfaces of genus 1 (left) and genus 2 (right).}
\label{Graphics_Riemann_surfaces}
\end{figure}

 The basis of the holomorphic
differential forms, the Abel mapping and the B-period $\tau$ are given as follows:

$\omega=2\pi
i\displaystyle\(\displaystyle\frac{dk}{\mathrm{w}(k)}\)\(\displaystyle
\int\limits_{a}\displaystyle\frac{dk}{\mathrm{w}(k)}\)^{-1},
\quad A(P)=\int\limits_{i\c}^{P}\omega, \quad \tau=\tau(\xi):=\displaystyle
\int\limits_{b}\omega<0.$
\\
In terms of the theta function

$\Theta(z)=\Theta(z,\tau(\xi))=\sum \limits_{m=-\infty}^{\infty}
\exp\left\{\displaystyle\frac{1}{2}\tau(\xi)m^2+zm\right\},
$
\\
the solution $V^{(mod)}(\xi,t;k)=\(V^{(mod)}_{[1]}(\xi,t;k),V^{(mod)}_{[2]}(\xi,t;k)\)$ 
of the model problem is given as follows:
\begin{eqnarray}
\nonumber
V^{(mod)}_{[1]}=\sqrt[4]{\frac{k^2+d_0^2}{k^2+\c^2}}\frac{\Theta\(A(k)-\pi\i-\frac{\i tB+\i \Delta}{2}\)\Theta\(A(k)-\frac{\i tB+\i \Delta}{2}\)\Theta^2\(\frac{\pi\i}{2}\)}{\Theta\(A(k)-\pi\i\)\Theta\(A(k)\)\Theta\(\frac{\pi\i}{2}-\frac{\i tB+\i\Delta}{2}\)\Theta\(\frac{\pi\i}{2}+\frac{\i tB+\i\Delta}{2}\)},
\\
\label{Vmod_ell_1}
\\
\nonumber
V^{(mod)}_{[2]}=\sqrt[4]{\frac{k^2+d_0^2}{k^2+\c^2}}\frac{\Theta\(-A(k)-\pi\i-\frac{\i tB+\i \Delta}{2}\)\Theta\(-A(k)-\frac{\i tB+\i \Delta}{2}\)\Theta^2\(\frac{\pi\i}{2}\)}{\Theta\(A(k)-\pi\i\)\Theta\(A(k)\)\Theta\(\frac{\pi\i}{2}-\frac{\i tB+\i\Delta}{2}\)\Theta\(\frac{\pi\i}{2}+\frac{\i tB+\i\Delta}{2}\)}.
\\\label{Vmod_ell_2}
\end{eqnarray}

\begin{com}
 Further we will need also the matrix solution $M^{(mod)}(\xi,t;k)$ of the model problem, which is fixed by asymptotics
$M^{(mod)}(\xi,t;k)\to I$ as $k\to\infty$, and hence posseses the property $\det M^{(mod)}\equiv 1.$ It can be constructed 
in the same way as in \cite{Bleher11}, section 4.2.4.2.
Let us notice that the vector and matrix solutions are connected by the formula $V^{(mod)}=(1,1)M^{(mod)}.$ 
The model solution
posseses the following behavior in the vicinities of the points $\pm\i d_0:$ 
$$M_{ij}^{(mod)}=O\((k-\i d_0)^{-1/4}\),\quad i,j,=1,2,$$
and this estimate is uniform with respect to $\xi,t$. The exact form 
of the matrix solution is not important for us, hence we don't write it here.
\end{com}

\vskip.5cm
\textbf{b. Hyperelliptic model problem.}
To solve the model problem (\ref{model problem}) with $d_1>0$ we introduce the
Riemann surface $X$ of the function
\vskip2mm

$\hskip2cm\mathrm{w}^2(k)=(k^2+\c^2)(k^2+d_0^2)(k^2+d_1^2)$
\\ as in \cite{KM12}. The
upper and the lower  sheets of the surface are two complex planes which
are merged along the cuts $[\ii \c, \ii d_0]$, $[\i d_1,-\i d_1]$
and $[-\ii d_0, -\ii \c] $. The square root is fixed by the condition that on the upper sheet of this surface
$\mathrm{w}(1)>0$. The basis $\{\BS a_1, \BS b_1,\BS a_2, \BS
b_2\}$ of cycles of this Riemann surface is as follows (see Figure \ref{Graphics_Riemann_surfaces}b).  The $\BS
a_1$ -cycle starts on the upper sheet from the right side of the
cut $[\ii \c, \ii d_0]$, goes to the right side of the cut $[\ii
d_1, -\ii d_1]$, proceeds to the lower sheet, and then returns to
the starting point. The $\BS b_1$-cycle is a closed counterclockwise oriented simple loop around the  cut $[\ii \c, \ii
d_0]$. The $\BS a_2$ -cycle starts on the upper sheet from the right
side of the cut $[\ii d_1, -\ii d_1]$, goes to
the right side of the cut $[-\ii d_0, -\ii \c]$, proceeds to the
lower sheet, and then returns to the starting point. The $\BS
b_2$-cycle is a closed counterclockwise oriented simple loop
around the  segment $[\ii \c, -\ii d_1]$. 
%\\\\
%\textbf{\textcolor{green}{Picture?}}\\

\noindent
Introduce the basis
$
\d\omega=\begin{pmatrix}\d\omega_1\\
\d\omega_2
\end{pmatrix}
$
of the normalized holomorphic differentials:
\vskip2mm
\noindent
$d\omega_1=\displaystyle\frac{\(c_1 k+c_2\)\d k}{\mathrm{w}(k)},\ d\omega_2=\displaystyle\frac{\(c_1 k-c_2\)\d k}{\mathrm{w}(k)}, \textrm{ where }
$
\\
\vskip2mm\noindent  \vskip-1mm 
$c_1=\dsfrac{\pi\i}{\int\limits_{a_1}\dsfrac{k\d k}{\w(k,\xi)}}\in(-\infty,0),
\qquad c_2=\dsfrac{\pi\i}{\int\limits_{a_1}\dsfrac{\d k}{\w(k,\xi)}}\in(0,-\i\infty).$
\\
In terms of the $B$-period $B=B(\xi)$ 
%\\
$$\(\(B_{jl}=\ds\int_{b_j}\d\omega_l\)\)=\begin{pmatrix}B_1&B_2\\B_2&B_1\end{pmatrix},
\ B_j:= \int\limits_{\BS b_1} d\omega_j\ , j=1,2,\ ,\textrm{ }B_1<B_2<0\ ,
$$
%\\
the corresponding theta function is given as follows
\vskip2mm
$\Theta(z)=\Theta(z\vert B(\xi))=\sum \limits_{m\in\mathbb{Z}^2}
\exp\left\{\displaystyle\frac{1}{2}\(B(\xi)m,m\)+(z,m)\right\}\
,\quad z\in\mathbb{C}^2.
$
\\
\noindent Since $B$ is a symmetric matrix, we have 
$\Theta\begin{pmatrix}z_1\\z_2\end{pmatrix}=\Theta\begin{pmatrix}z_2\\z_1\end{pmatrix}.$
\vskip-14mm
\begin{equation}\label{Theta: z1_z2 symmetry}
\hskip7cm
\end{equation} 
Introduce the Abel map 
$\quad
A:X\rightarrow\mathbb{C}^2/\left(2\pi
i\mathbb{Z}^2+B(\xi)\mathbb{Z}^2 \right)\ ,\
A(P)=\int\limits_{\ii \c}^{P}\d\omega 
$
\\
% -----------
%\begin{equation}\label{A}
%\end{equation} 
% ---------
and the
functions $\widehat\varphi(k,\xi),\widehat\psi(k,\xi):\{\textrm{the first sheet of
$X$}\}\rightarrow \mathbb{C}$
\\
$
%\begin{equation}\label{phipsi}
\begin{array}{l}\widehat\varphi_j(k,\xi)=
\displaystyle \frac{\Theta(A(k)-A(D_j)-K-
\frac12\(itB_g(\xi)+\ii\Delta(\xi)\)(1,1)^{T})}
{\Theta(A(k)-A(D_j)-K)},\\
\widehat\psi_j(k,\xi)=\displaystyle\frac{\Theta(-A(k)-A(\widehat D_j)-K-\frac12\(itB_g(\xi)+\ii\Delta(\xi)\)(1,1)^T)}
{\Theta(-A(k)-A(\widehat D_j)-K)},\end{array}\quad j= {1,2},
$
%\end{equation}
with the vector $K=\(K_1, K_2\)^T$ of Riemannian constants and nonspecial divisors 
$$K\equiv \begin{pmatrix}\frac{B_1}{2}+\frac{B_2}{2}\\\pi\i+\frac{B_1}{2}+\frac{B_2}{2}\end{pmatrix},
\quad
\widehat D_1=\left\{\i d_0+0,0^++0\right\}, \ \widehat D_2=\left\{-\i d_0+0,0^++0\right\},$$ consisting 
of the branch points $\pm\i d_0$ and the point on the upper sheet $0^+.$ Since we consider a cut complex plane, 
it does matter at which side of the bank 
the point is taken, hence we write $\pm0$ to specify this.

%Functions ({\ref{phipsi}) have the following
%properties:
%\[\begin{array}{l}\widehat\varphi_{j+}(k,\xi)=\widehat\psi_{j-}(k,\xi)\\\\
%\widehat\psi_{j+}(k,\xi)=\widehat\varphi_{j-}(k,\xi)\end{array},\quad k\in (\ii
%\c,\ii d_0)\cup(\i d_1,-\i d_1)\cup(-\ii d_0,-\ii \c),\]
%\[\begin{array}{l}\widehat\varphi_{j-}(k,\xi)=\widehat\varphi_{j+}(k,\xi)e^{\frac12\(itB_g(\xi)+\ii\Delta(\xi)\)}\\\\
%\widehat\psi_{j-}(k,\xi)=\widehat\psi_{j+}(k,\xi)e^{-\frac12\(itB_g(\xi)-\ii\Delta(\xi)\)}\end{array},
%\quad k\in (\ii d_0,\i d_1)\cup(-\i d_1,-\ii d_0).\]

\noindent Introduce also the function 

$\hskip2cm\gamma(k)=\sqrt[4]{\dsfrac{k^2+d_0^2}{k^2+\c^2}\cdot\dsfrac{k^2}{k^2+\d_1^2}}\ ,$
\\
 which is
analytic outside of the union of segments $[\ii \c,\ii d_0]\cup[\ii
d_1,-\ii d_1]\cup[-\i d_0,-\i \c].$ 
%and satisfies the jump
%conditions
%\[
%\gamma_-(k,\xi)=-\ii\gamma_+(k,\xi), \qquad k\in[\ii \c,\ii
%d_0]\cup[\i d_1,0],
%\]
%\[
%\gamma_-(k,\xi)=\ii\gamma_+(k,\xi), \qquad k\in
%[0,-\i d_1]\cup [-\ii d_0,-\ii
%\c].
%\]

\noindent It is easy to see that 
\begin{subequations}\label{divisors_particular}\begin{eqnarray}
&A(\i d_0+0)=\begin{pmatrix}-B_1/2\\-B_2/2\end{pmatrix},& A(-\i d_0+0)=\begin{pmatrix}-B_2/2+\pi\i,\\-B_1/2+\pi\i\end{pmatrix}\\\nonumber\\
&A(\infty)=\begin{pmatrix}\i\varphi\\i(\pi-\varphi)\end{pmatrix}&\textrm{ 
with }\varphi=\frac{\pi}{2}+\displaystyle\int\limits_{\i\c}^{+\i\infty}\frac{c_2\d k}{\w(k)}>\frac{\pi}{2} \textrm{, }\\ 
&A(0^+\hskip-1mm+\hskip-1mm0)\hskip-1mm=\hskip-1mm\begin{pmatrix}x+\pi\i\\x\end{pmatrix}
&\textrm{ with }x\hskip-1mm=\hskip-1mm\frac{-B_1-B_2}{4}\hskip-1mm+\hskip-4mm\displaystyle\int\limits_{\i d_1+0}^{0^++0}\hskip-1mm\dsfrac{c_1k\d k}{\w(k)}\hskip-1mm\in\hskip-1mm\mathbb{R}.
\end{eqnarray}
\end{subequations}

\noindent Concerning asymptotics as $k\to\infty$, we notice that properties (\ref{Theta: z1_z2 symmetry}),  (\ref{divisors_particular}) imply 
$\widehat\varphi_1(\infty)\widehat\varphi_2(\infty)=\widehat\psi_1(\infty)\widehat\psi_2(\infty),$
\\
and hence,
\begin{equation}\widehat V(\xi,t;k):=\(\gamma(k)\dsfrac{\widehat\varphi_1(k)\widehat\varphi_2(k)}
{\widehat\varphi_1(\infty)\widehat\varphi_2(\infty)},\ \ 
\gamma(k)\dsfrac{\widehat\psi_1(k)\widehat\psi_2(k)}{\widehat\psi_1(\infty)\widehat\psi_2(\infty)}\)
 \label{Vmod_hyper}
\end{equation}
% -----------
solves the hyperelliptic model problem.

%\begin{com} The hypereliptic model problem solution can be obtained in an equivalent but different form by solving 
%the correponding 
%matrix RH problem and multiplying it from the left 
%by the vector $\vec{ e }=(1,1)$. It gives another expression for the solution of the initial value problem (\ref{CH})-(\ref{init_cond}), which is more complicated however than the one we get.
%\end{com}

\begin{com}
 Further we will need also the matrix solution $M^{(mod)}(\xi,t;k)$ of the model problem, which is fixed by asymptotics
$M^{(mod)}(\xi,t;k)\to I$ as $k\to\infty$, and hence posseses the property $\det M^{(mod)}\equiv 1.$
It can be constructed in the same way as in \cite{Bleher11}, section 4.2.4.2.
Let us notice that the vector and matrix solutions are connected by the formula $V^{(mod)}=(1,1)M^{(mod)}.$ The model solution
posseses the following behavior in the vicinities of the points $\pm\i d_l$, $l=1,2:$ 
$$M_{ij}^{(mod)}=O\((k-\i d_l)^{-1/4}\),\quad l,i,j,=1,2,$$
and this estimate is uniform with respect to $\xi,t$. The exact form 
of the matrix solution is not important for us, hence we don't write it here.
\end{com}

\textbf{c. Constant model problem}. Finally, in the case when $d_0=d_1=0$ (see cases $A$, $B_2$, $C_2$ of 
Figure \ref{Figure_cont_middle_function}), which is a trivial case,
the solution of the model problem 
(\ref{model problem}) is given by 

$\hskip2cm V^{(mod)}(\xi,t;k)=\(\sqrt[4]{\frac{k^2}{k^2+\c^2}}, \sqrt[4]{\frac{k^2}{k^2+\c^2}}\).$

%The corresponding matrix solution $M^{(mod)}$ is
%$$M^{(mod)}(\xi,t;k)=\begin{pmatrix}\frac{1}{2}\(\gamma(k)+\)\end{pmatrix}$$

\textbf{Step 8.} Estimate of the contribution of the jumps on contours $L_j,$ $j=1,...,8$ and final asymptotics.

We start with the following lemma:
\begin{lem}\label{Prop_apr_estimates} Let $V^{(6)}(\xi,t;k)$ be the (vector-valued) solution of the 
RH problem from Step 6 of Section \ref{sect_RH problem transformations} (with jump contour as in Figure 
\ref{Figure_cont_middle_function}); $V^{(mod)}(\xi,t;k)$ and $M^{(mod)}(\xi,t;k)$ be the vector and 
matrix-valued solutions of the corresponding model problem from Step 7 
(with jump matrix $[\i\c,-\i\c]$). 
Define the correction function $V_{err}=V^{(6)}\(M^{(mod)}\)^{-1}.$ Denote 
$$V_{err}(\xi,t;k) =: V_{err}^{(0)}(\xi,t)+V_{err}^{(1)}(\xi,t)(k-\i/2)+V_{err}^{(2)}(\xi,t;k)(k-\i/2)^2.$$
Then $$|V_{err,j}^{(0)}(\xi,t)-1|=\mathrm{O}\(t^{-1/2}\),\ 
|V_{err,j}^{(l)}(\xi,t;k)|=\mathrm{O}\(t^{-1/2}\), \ j=1,2,\ l=1,2,$$ 
uniformly for $k\in\mathbb{C}$, $\xi_m+\varepsilon<\xi<\xi_{m+1}-\varepsilon$ as $t\to\infty,\quad m=1,2,3,$ 
where $\varepsilon>0$ is a sufficiently small positive number.
\end{lem}

\begin{SketchProof} 

For definitness let us consider the case $C_1$, when both types of the points $\pm\i d_0$ and $\pm k_0$ are presented. The 
other cases can be treated similarly.

The jump contour $\Sigma_{err}:=\Sigma^{(6)}\setminus[\i\c,-\i\c]$ for the correction function $V_{err}$ contains
the following points of nonuniformity of the jump matrix (see Figures \ref{Figure_cont_middle_function},
\ref{Figure_phase_middle_function_plot}):\\
$\bullet$ the real points {$\pm k_0\in\R$;}\\
$\bullet$ the imaginary points {$\pm \i\mu_0\in[\i\c,-\i\c]$;}\\
$\bullet$ the end points {$\pm \i\c.$}

The jump matrix $J^{(6)}$ is not close to the identity matrix in $L_{\infty}$ sense. However,
the jump matrices in the vicinities of the points $\pm k_0$, $\pm\i d_0$ are highly oscillatory, but uniformly bounded.
This provides that the solution of the RH problem is also uniformly bounded, 
$$|V_{i}^{(6)}|\le C,\quad |M_{ij}^{(6)}|\le C,\quad i,j=1,2,$$
where $C$ is a constant that does not depend on $\xi,t,k.$

The jump matrix $J_{err}$ for the correction function $V_{err}$ is equal to 
$$J_{err}=M^{(mod)}_{+}J^{(6)}\(J^{(mod)}\)^{-1}\(M^{(mod)}_{+}\)^{-1}=
\begin{cases}
I,\ k\in(\i\c,-\i\c),
\\
M^{(mod)}J^{(6)}\(M^{(mod)}\)^{-1},\\\qquad k\in\Sigma^{(6)}\backslash[\i\c,-\i\c],
\end{cases}$$
where, as usual, by $+/-$ we denote the limit from the positive/ negative side of the contour.
% ---------------- %
By Sokhotski-Plemelj formula, we have the standard representation of $V_{err}:$
% -----------------
% -----------------
%We have the following equation for $V^{err,+}$:
%\[V_{err,+}(k)=(1,1)+\frac{1}{2\pi\i}\int\limits_{\Sigma_{err}}\frac{V_{err,+}(s)\d s}{(s-k)_+},\] 
%and the consequent representation for $V_{err}$:
\[V_{err}(k)=(1,1)+\frac{1}{2\pi\i}\int\limits_{\Sigma_{err}}\frac{V_{err,+}(s)(I-J_{err})\d s}{s-k},\]
and hence (we denote here $E:=(1,1)$),
\[V_{err}(k)=E+\int\limits_{\Sigma_{err}}\frac{F(s)\d s}{2\pi\i\(s-\frac{\i}{2}\)} + 
\int\limits_{\Sigma_{err}}\frac{F(s)\d s \(k-\frac{\i}{2}\)}{2\pi\i(s-\frac{\i}{2})^2}\ +
\int\limits_{\Sigma_{err}}\frac{F(s)\d s \(k-\frac{\i}{2}\)^2}{2\pi\i(s-\frac{\i}{2})^2(s-k)}\ =\]
\[=V_{err}^{(0)}+ V_{err}^{(1)}\,(k-\i/2)+V_{err}^{(2)}\,(k-\i/2).\]
where $F(\xi,t;s):=V^{(6)}_{+}(\xi,t;s)M_{mod}^{-1}(\xi,t;s)(I-J_{err}(\xi,t;s)).$
Let estimate, for example, $V^{(1)}_{err}.$ Quantities $V^{(0)}_{err}-(1,1)$ and $V^{(2)}_{err}$ can 
be estimated similarly.

Let us denote the parts of the contour $\Sigma_{err}$ in the vicinities of the points $\pm\i d_0$, $\pm k_0,$ as follows:
\begin{itemize}
 \item $[\i d_0]:\ $  $L_{1,\pm},$ depending on whether 
$\Re k>0$ or $\Re k<0;$
 \item $[-\i d_0]:\ $  $L_{2,\pm},$ depending on whether 
$\Re k>0$ or $\Re k<0;$
\item $[\i k_0]:\ $ $L_{j,\r}, j=1,2,3,4; $
\item $[-\i k_0]:\ $ $L_{j,\l}, j=1,2,3,4.$
\end{itemize}
The rest of the contour $\Sigma_{err}$ we denote by $\Sigma_{err}^C.$ Let us notice, that the points $\pm\i\c$ can be 
excluded from analysis. Indeed, by moving up slightly the lenses $L_{u\l},$ $ L_{u,\r}$ towards a point $\i(\c+\delta), $ 
$\delta>0,$ we transform 
the oscillatory behavior of jump matrices into exponentially decaying. 
The same is for the point $-\i\c.$

\begin{enumerate}
 \item 

To estimate $F(s)$ on $L_{1,\pm}$, let us expand the 
function $g(k,\xi)$ in a vicinity of the point
$\i d_0$ :
$$g(k)=g_{\pm}+g^{(3/2)}\(\frac{k-\i d_0}{-\i}\)^{3/2}(1+O(k-\i d_0)), k\to\i d_0\pm0, 
$$
 with $$g_{\pm}\in\mathbb{R},\ g^{(3/2)}>0,$$
where the sign $+$ or $-$ is taken for $L_{1,\pm},$ respectively. 
Let us make a change of variable on $L_{1,+}:$
$$g(k)=g_{\pm}+g^{(3/2)} \(\e^{\i\alpha}u\)^{3/2},\quad \frac{k-\i d_0}{-\i}=\e^{\i\alpha} u(1+O(u)),
\quad u\to 0, u>0,$$
where $0<\alpha<\pi/3$ is the angle which $L_{1,+}$ makes with the ray $(\i d_0,0).$
Let us notice, that we can choose $L_{1,\r}$ in such a way that $u>0$ is real.
Hence, 
we can estimate $F(s)$ on $L_{1,\pm}$ as
$$|F_j(s)|\leq\frac{C}{\sqrt[4]{|k-\i d_0|}}\left|\e^{-2t g^{(3/2)} \sin\frac{3\alpha}{2}\, u^{3/2}}\right|,\quad j=1,2,$$
where $C>0$ is a generic constant that does not depend on $\xi, t, k.$ 
Hence, $$\int\limits_{L_{1,+}}\frac{F(s)\d s}{\(s-\frac{\i}{2}\)^2}\le \frac{C}{\mathrm{dist}^2 (\Sigma_{err},\i/2)}
\int\limits_{0}^{+\infty}\frac{1}{\sqrt[4]u}\e^{-2t g^{(3/2)} \sin\frac{3\alpha}{2}\, u^{3/2}}\d u\le Ct^{-1/2}.$$

Integrals over the contours $L_{1,-}$, $L_{2,\pm}$ can be estimated similarly. 

\item To estimate $F(s)$ on $L_{3,\r}, $ let us expand the function $g(k,\xi)$ in a vicinity of the point $k_0:$
$$g(k)=g^{(0)}+g^{(2)}(k-k_0)^2(1+O(k-k_0)),$$
with $$g^{(0)}\in\R,\ g^{(2)}<0.$$

Let us make a change of variable on $L_{1,+}:$
$$g(k)=g^{(0)}+g^{(2)} \(\e^{\i\beta}v\)^{2},\quad k-\i k_0 = \e^{\i\beta} v(1+O(v)),
\quad v\to 0, v>0,$$
where $0<\beta<\pi/2$ is the angle which $L_{3,\r}$ makes with the ray $(\i k_0,+\infty).$
Let us notice, that we can choose $L_{3,\r}$ in such a way that $v>0$ is real.
Hence, 
we can estimate $F(s)$ on $L_{3,\r}$ as
$$|F_j(s)|\leq \left|\e^{2t g^{(2)} \sin{2\beta}\,\, v^{2}}\right|,\quad j=1,2,$$
where $C>0$ is a generic constant that does not depend on $\xi, t, k.$ 
Hence, $$\int\limits_{L_{3,\r}}\frac{F(s)\d s}{\(s-\frac{\i}{2}\)^2}\le \frac{C}{\mathrm{dist}^2 (\Sigma_{err},\i/2)}
\int\limits_{0}^{+\infty}\e^{2t g^{(2)} \sin{2\beta}\,\, v^{2}}\d v\le Ct^{-1/2}.$$

Integrals over the contours $L_{1,-}$, $L_{2,\pm}$ can be estimated similarly. 
\end{enumerate}

Since the integral over $\Sigma_{err}^C$ is exponentially small, we have that $$V_{err,j}^{(1)}=\int\limits_{\Sigma_{err}}
\frac{F_j(s)\d s}{2\pi\i\(s-\frac{\i}{2}\)^2}=\mathrm{O}\(t^{-1/2}\),\quad j=1,2.$$

%\noindent \textit{End of the proof.}%$\hfill\mathbf{\square}$
\end{SketchProof}

\begin{com}The estimate of the integrals in the vicinities of the points $\pm k_0,\pm k_1$ is sharp,
while the estimate of integrals in the vicinities of the points $\pm\i\mu_0$, $\pm\i\mu_1$
 is a rough one: the actual behavior of $V_{err}-(1,1)$ is of the order $t^{-1}.$
\end{com}

\begin{lem}\label{lemma: u_u^mod}
 Let $u^{(mod)}(x,t),$ $m^{(mod)}(x,t)$ be defined as in $(\ref{M1*M2}\mathrm{a,b})$, where we substitute $V$ instead
of $V^{(mod)}.$
Then for sufficiently small $\varepsilon>0$, $$u(x,t) = u^{(mod)}(x-\tilde x(x,t),t)(1+O(t^{-1/2})),\ m(x,t) = m^{(mod)}(x,t)(1+O(t^{-1/2})),$$
uniformly for $(\zeta_l+\varepsilon) \omega t\leq x\leq (\zeta_{l+1}-\varepsilon) \omega t, \quad l=1,2,3, $
where

$\tilde x(x,t) = \log\left[F^2\hskip-1mm\(\hskip-1mm\frac{\i}{2}\)
\delta^{2}\hskip-1mm\(\frac{\i}{2},\xi\)
\Omega^2\hskip-1mm\(\frac{\i}{2}\)\e^{-2\i t\(g-g_{\r}\)(\i/2)}\right]$

\noindent
is a slowly changing function. Here $g,g_r, F,\delta, \Omega$ are defined in Section \ref{sect_RH problem transformations},
and the parameter $\xi \hskip-.5mm = \hskip-.5mm \frac{y}{\omega t}$ is the same as in parametric definition of $u^{(mod)}$, $m^{(mod)}.$ 

\end{lem}

% ------------
\begin{Proof}
 Following the chain of transformations of the RH problem made in Steps 1--6 of the present 
Section \ref{sect_RH problem transformations}, we conclude that in a neighborhood of the point $k=\i/2$
\begin{eqnarray}\nonumber V&=&V^{(1)}\e^{-i t(g-g_{\r})\sigma_3}=V^{(2)}\e^{-i t(g-g_{\r})\sigma_3}=
\mbox{because}=V^{(3)}\Omega^{\sigma_3}(k)\e^{-i t(g-g_{\r})\sigma_3}=\\\nonumber
&=&V^{(4)}\delta^{\sigma_3}(k,\xi)\Omega^{\sigma_3}(k)\e^{-i t(g-g_{\r})\sigma_3}
=V^{(5)}(\xi,t;k)\delta^{\sigma}(k,\xi)\Omega^{\sigma_3}(k)\e^{-i t(g-g_{\r})\sigma_3}=
\\\nonumber&=&V^{(6)}F^{\sigma_3}(k,\xi)\delta^{\sigma_3}(k,\xi)\Omega^{\sigma_3}(k)\e^{-i t(g-g_{\r})\sigma_3}
%=\\
%&=&V^{(mod)}F^{\sigma_3}(k,\xi)\delta^{\sigma_3}(k,\xi)\Omega^{\sigma_3}(k)\e^{-i t(g-g_{\r})\sigma_3}\(1+o(1)\)
,
\quad t\rightarrow\infty.\label{RH problem transformations}\end{eqnarray}

Taking into account Lemma \ref{Prop_apr_estimates} and the property
$V^{(mod)}=(1,1)M^{(mod)}$, we see, that $V^{(6)}=V_{err}\(M^{(mod)}\)^{-1}$, and since all the elements of 
the vector and matrix functions $V^{(mod)}$, $M^{(mod)}$ are separated from $0$, we can write
$$V^{(6)} = V^{(err)}M^{(mod)} = V^{(mod)}
\(I+\mathrm{O}\(t^{-\frac12}\)\).
$$

\noindent
Hence, from (\ref{M1*M2}a), (\ref{M1*M2}b), we obtain 

$\hskip-5mm\e^{x-y}\hskip-1mm=\hskip-1mm\left.\dsfrac{V^{(mod)}_{[1]}}{V^{(mod)}_{[2]}}\right|_{k=\frac{\i}{2}} 
\hskip-3mm F^2\hskip-1mm\(\hskip-1mm\frac{\i}{2}\)
\delta^{2}\hskip-1mm\(\frac{\i}{2},\xi\)
\Omega^2\hskip-1mm\(\frac{\i}{2}\)\e^{-2\i t\(g-g_{\r}\)(\i/2)}\(1+\mathrm{O}\(t^{-\frac12}\)\),
$

$\hskip-5mm\sqrt{\frac{m+\omega}{\omega}}\hskip-0.5mm\(\hskip-0.5mm 1\hskip-1mm+
\hskip-1mm\frac{2\i}{\omega}u(x,t)\hskip-0.5mm\(
\hskip-0.5mm k \hskip-0.mm - \hskip-1.mm \frac{\i}{2}\hskip-0.5mm\)
\hskip-0.5mm+\hskip-0.mm\mathrm{O}\hskip-0.5mm\(\hskip-0.5mmk\hskip-0.mm-\hskip-1.mm\frac{\i}{2}\hskip-.5mm\)^{
\hskip-1.mm 2}\)
\hskip-0.5mm=\hskip-0.5mm
V^{(mod)}_{[1]}\hskip-.5mm(x,\hskip-.5mm t; \hskip-.5mm k)V^{(mod)}_{[2]}\hskip-.5mm(\hskip-.5mmx\hskip-.5mm,
t\hskip-.5mm;\hskip-.5mmk\hskip-.5mm)
\hskip-1.5mm\(\hskip-1.5mm 1 \hskip-1.5mm + \hskip-1.5mm \mathrm{O} \hskip-.5mm \( \hskip-.5mm t^{-\frac12} \hskip-.5mm \) 
\hskip-1.5mm\)\hskip-.5mm.$

\begin{subequations}
\label{xyum_parametrix repre}
\vskip-27mm\begin{eqnarray}
&&
\\
\nonumber
&&
\nonumber\\&&
\nonumber\\&&
\end{eqnarray}
\end{subequations}
The corresponding borders between the different asymptotic sectors in the $x,t$-half-plane are expressed in terms 
of the 
corresponding borders in the $y,t$-half-plane as follows:
% ----------- 

$\zeta\equiv\frac{x}{\omega t}=\xi-\frac{2\i}{\omega}\(g-g_{\r}\)\(\frac{\i}{2},\xi\)+\mathrm{O}\(t^{-1}\),\quad t\to\infty,
$
\vskip-14mm
\begin{equation}{\label{borders_xt_yt_elliptic}}
\end{equation}
% -------

\end{Proof}

Finally, Lemma \ref{lemma: u_u^mod} and formulas (\ref{Vmod_ell_1}), (\ref{Vmod_ell_2}), (\ref{Vmod_hyper})
with straight, but cumbersome calculations lead to 
\begin{teor}\label{Teor_Elliptic} \textbf{Elliptic asymptotics (genus 1).} Under Assumptions 1-4 of Section \ref{sect: Introduction}, the asymptotics of the solution of the Cauchy problem  (\ref{CH})-(\ref{init_cond}) can be described as follows.
% ----------- 
Let $\varepsilon>0$ be a sufficiently small positive number. Let $\zeta_j, j=2,3,4,$ be as in 
(\ref{xi_j_def_beg})--(\ref{xi_j_def_end}).
\\
% ---------- 
Then for $t\to\infty,$ in both of the following two cases (see Figure \ref{Figure_phase_middle_function_plot}):
\begin{enumerate}
\item (case $C_1$) $\frac{c}{\omega}>3$ or $1<\frac{c}{\omega}<3,$ and 
$\(\zeta_2+\varepsilon\)\omega t<x<\(\zeta_3-\varepsilon\)\omega t$,

\item (case $D$) $\frac{c}{\omega}>3,$ $1<\frac{c}{\omega}<3$, or $0<\frac{c}{\omega}<1,$ and 
$\zeta\equiv \frac{x}{\omega t}\in\(\zeta_3+\varepsilon\)\omega t<x<\(\zeta_4-\varepsilon\)\omega t$,

\end{enumerate}

\noindent 
the asymptotics of the solution of the Cauchy problem (\ref{CH})-(\ref{init_cond}) is given parametrically 
by the following formulas:
% ------------ 
\\
$
x\hskip-1mm=\hskip-1mmy\hskip-.5mm-\hskip-.5mm2\i t\hskip-.5mm\(\hskip-.5mmg\hskip-1mm-\hskip-1mmg_{\r}\hskip-.5mm\)
\hskip-1mm(\hskip-.5mm\i/2\hskip-.5mm)\hskip-.5mm+\hskip-.5mm\log \hskip-1mm F^2\hskip-1mm\(\hskip-.5mm\i/2,\xi\hskip-.5mm\)
\hskip-.5mm+\log \delta^{2}
\hskip-.5mm(\hskip-.5mm\i/2,\xi\hskip-.5mm)\hskip-.5mm +\hskip-.5mm\log\Lambda^2\(\i/2\)\hskip-1mm+\hskip-1mmE
\hskip-1mm\(\hskip-1mm\frac{tB(\xi)+\Delta(\xi)}{2}\hskip-1mm\)\hskip-.5mm+
$
\begin{equation}\label{elliptic_asymp_x}
\hfill\hskip-.mm+\textrm{O}(t^{-\frac12}),
\end{equation}%$
%\\
%$
\begin{equation}\label{elliptic_asymp_u}u(x,t)\hskip-.5mm=\hskip-.5mm\dsfrac{c(1-d_0^2(\xi)\hat c^{-2})}{(1-4d_0^2(\xi))}
\hskip-.5mm+\hskip-.5mm\Gamma(\xi)\hskip-.5mm\(E'\(\frac{tB(\xi)+\Delta(\xi)}{2}\)-E'(0)\)+\textrm{O}(t^{-\frac12}),\ \xi=\frac{y}{t},
\end{equation}
%$
%\\
$\sqrt{\dsfrac{m(x,t)+\omega}{\omega}}
=\sqrt{\dsfrac{1-4d_0^2}{1-4\c^2}}
\dsfrac{G\(\pi+\frac{ tB+\Delta}{2}\)G\(\frac{ tB+\Delta}{2}\)}{G\(\pi\)G(0)}\frac{H^2(0)}{H^2\(\frac{ tB+\Delta}{2}\)}
$
\\
\vskip1mm\noindent
with $g\equiv g_2$ in the first case ($C_1$) and $g\equiv g_3$ in the second case ($D$). Here,
\vskip2mm
$
E(U)=\log\dsfrac{\Theta(\i U-A(i/2)+\pi\i)\Theta(\i U - A(i/2))}{\Theta(\i U + A(i/2)+\pi\i)\Theta(\i U + A(i/2))},
$
\\
\vskip1mm\noindent
$G(U)=\Theta\(A\(\frac{\i}{2}\)+\i U\)\Theta\(A\(\frac{\i}{2}\)-\i U\),\quad H(u)=\Theta\(\frac{\pi\i}{2}+\i U\)\Theta\(\frac{\pi\i}{2}-\i U\),$
\[\Gamma(\xi) = \frac{\pi\omega}{4 \w\(\frac\i 2\) \displaystyle\int_0^{d_0}\frac{\d s}{\sqrt{(\c^2-s^2)(d_0^2-s^2)}} },\quad \hat c=\sqrt{\frac{c}{4(c+\omega)}}.\]

\end{teor}

%ccc

%ddd

\begin{teor}\label{Teor_HyperElliptic} \textbf{Hyperelliptic asymptotics (genus 2).}
Under Assumptions 1-4 of Section \ref{sect: Introduction} the asymptotics of the solution of the Cauchy problem  (\ref{CH})-(\ref{init_cond}) can be described as follows.
% ---------- 
\\
Let $\frac{c}{\omega}>3$ and $\varepsilon>0$ be a sufficiently small positive number.
Let $\zeta_2$ be as in (\ref{xi_j_def_beg})--(\ref{xi_j_def_end}).
\noindent
Then for $t\to\infty$ and  
 $\(\zeta_1+\varepsilon\)\omega t<x<\(\zeta_2-\varepsilon\)\omega t$ (case $B_1$ in Figure 
\ref{Figure_phase_middle_function_plot}),
% -------
the asymptotics of the solution of the Cauchy problem (\ref{CH})-(\ref{init_cond}) is given 
parametrically in terms of the modulated hyperelliptic functions:
%the solution of the iv problem
%takes the
% --------------
%\vskip2mm
\noindent %$
\begin{eqnarray}\nonumber x\hskip-1mm&=&\hskip-1mmy\hskip-.5mm-\hskip-.5mm2\i t\hskip-1mm\(g\hskip-1mm-\hskip-1mmg_{\r}\)
\hskip-1mm\(\frac\i2\)+\log F^2\hskip-1mm\(\frac\i2,\xi\)+\log \delta^{2}\(\frac\i2,\xi\) +
\log\Lambda^2\(\frac\i2\)\hskip-1mm
\\
&+&\hskip-1mmE\hskip-1mm\(\hskip-1mm\frac{tB(\xi)+\Delta(\xi)}{2}\hskip-1mm\)+o(1),
\label{hyperelliptic_asymp_x}\end{eqnarray}%$
\begin{eqnarray}\nonumber
\frac{2\i u}{\omega}&=&\frac{8\i(\c^2-d_0^2+d_1^2-8\c^2d_1^2+16\c^2d_0^2d_1^2)}{(1-4\c^2)(1-4d_0^2)(1-4d_1^2)}+
\\\nonumber
&+&H\(A(\i/2),A(\widehat D_1),\frac{tB+\Delta}{2}\)-H\(-A(\i/2),A(\widehat D_1),\frac{tB+\Delta}{2}\)+
\\\nonumber
&+& H\(A(\i/2),A(\widehat D_2),\frac{tB+\Delta}{2}\)-H\(-A(\i/2),A(\widehat D_2),\frac{tB+\Delta}{2}\)
\\\nonumber
&-&H\(A(\i/2),A(\widehat D_1),0\)+H\(-A(\i/2),A(\widehat D_1),0\)\\
&-&H\(A(\i/2),A(\widehat D_2),0\)+H\(-A(\i/2),A(\widehat D_2),0\),\label{hyperelliptic_asymp_u}\end{eqnarray}
% ----------------------

$\sqrt{\dsfrac{m+\omega}{\omega}}=\sqrt{\dsfrac{1-4d_0^2}{1-4\c^2}\cdot\dsfrac{1}{1-4d_1^2}}
\cdot\dsfrac{G\(\frac{tB+\Delta}{2}\)}{G(0)}\dsfrac{H^2(0)}{H^2(\frac{tB+\Delta}{2})}.$

Here, \begin{eqnarray*}\textstyle E(U)=\log\left[\frac{\Theta\(\i U\begin{pmatrix}1\\1\end{pmatrix}-A(\i/2)+A(\widehat D_1)+K\)}{\Theta\(\i U\begin{pmatrix}1\\1\end{pmatrix}+A(\i/2)+A(\widehat D_1)+K\)}\right.\times
\\
\times\frac{\Theta\(\i U\begin{pmatrix}1\\1\end{pmatrix}-A(\i/2)+A(\widehat D_2)+K\)}{\Theta\(\i U\begin{pmatrix}1\\1\end{pmatrix}+A(\i/2)+A(\widehat D_2)+K\)}\times
\end{eqnarray*}
% ----------
$\times\dsfrac{\Theta\(A(\i/2)+A(\widehat D_1)+K\)}
{\Theta\(-A(\i/2)+A(\widehat D_1)+K\)}
\times\left.\dsfrac{\Theta\(A(\i/2)+A(\widehat D_2)+K\)}
{\Theta\(-A(\i/2)+A(\widehat D_2)+K\)}\right]
$
% -------- 
\noindent
\begin{eqnarray}\textstyle \nonumber H(U_1,U_2,U_3)&=&\frac{\partial_1\Theta\(U_1-U_2-K-\i U_3\begin{pmatrix}1\\1\end{pmatrix}\)\dsfrac{c_1\i/2+c_2}{\w(\i/2)}}
{\Theta\(U_1-U_2-K-\i U_3\begin{pmatrix}1\\1\end{pmatrix}\)}+
\\
&+&
\frac{\partial_2\Theta\(U_1-U_2-K-\i U_3\begin{pmatrix}1\\1\end{pmatrix}\)\dsfrac{c_1\i/2-c_2}{\w(\i/2)}}
{\Theta\(U_1-U_2-K-\i U_3\begin{pmatrix}1\\1\end{pmatrix}\)}
\end{eqnarray}

\noindent and 
% ----------
\begin{eqnarray*}\textstyle G(U)=\Theta\(A(\frac\i 2)-A(\widehat D_1)-K-\i U\begin{pmatrix}1\\1\end{pmatrix}\)\Theta\(-A(\frac\i 2)-A(\widehat D_1)-K-\i U\begin{pmatrix}1\\1\end{pmatrix}\)\\\times
\Theta\(A(\frac\i 2)-A(\widehat D_2)-K-\i U\begin{pmatrix}1\\1\end{pmatrix}\)\Theta\(-A(\frac\i 2)-A(\widehat D_2)-K-\i U\begin{pmatrix}1\\1\end{pmatrix}\),
\end{eqnarray*}
% --------- 
\begin{eqnarray*}\textstyle
H(U)=\Theta\(A(\infty)-A(\widehat D_1)-K-\i U\begin{pmatrix}1\\1\end{pmatrix}\)
\Theta\(A(\infty)-A(\widehat D_2)-K-\i U\begin{pmatrix}1\\1\end{pmatrix}\).
\end{eqnarray*}

\end{teor}

%eee

%fff

Using the formulas (\ref{RH problem transformations}), (\ref{xyum_parametrix repre}) and (\ref{borders_xt_yt_elliptic}),
we get 
\begin{teor}\label{Teor_constant}
Under Assumptions 1-4 of Section \ref{sect: Introduction} the asymptotics of the solution of the Cauchy problem  (\ref{CH})-(\ref{init_cond}) can be described as follows.
% ------------------
\\
Let $\varepsilon>0$ be a sufficiently small positive number.

Then for $t\to\infty,$ in all of the following three cases (see Figure \ref{Figure_phase_middle_function_plot}):

\begin{enumerate}
\item (case $A$) $\frac{c}{\omega}\in(0,1)\cup(1,3)\cup(3,\infty)$ and $
-\infty<x<\(\frac{3c-\omega}{4\omega}
-\varepsilon\)\omega t,$

\item (case $B_2$) $\frac{c}{\omega}\in(1,3)$ and $
\(\frac{3c-\omega}{4\omega}+\varepsilon\)\omega t< x<\(
\frac{2\omega^2-c\omega+c^2}{\omega(c+\omega)}
-\varepsilon\)\omega t,$

or

$\frac{c}{\omega}\in(0,1)$ and $
\(\frac{3c-\omega}{4\omega}+\varepsilon\)\omega t <x<\(
\frac{c}{\omega}
-\varepsilon\)\omega t,$

\item (case $C_2$) $\frac{c}{\omega}\in(0,1)$ and $
\(\frac{c}{\omega}+\varepsilon\)\omega t<x<\(
\frac{2\omega^2-c\omega+c^2}{\omega(c+\omega)}
-\varepsilon\)\omega t,$

\end{enumerate}

\noindent 
the main term of asymptotics of the solution of the initial value problem (\ref{CH}), (\ref{init_cond}), (\ref{conditions u(x,t)}), (\ref{cond_m<c})  is
equal to the initial constant:
% ----------- %
\[
u(x,t)=c+\mathrm{O}\(t^{-1/2}\).
\]
\[
m(x,t)=c+\mathrm{O}\(t^{-1/2}\).
\]
In the $1^{st}$ case ($A$), the estimate can be refined:
\[
u(x,t)=c+\mathrm{O}\(\e^{-Ct}\).
\]
\[
m(x,t)=c+\mathrm{O}\(\e^{-Ct}\),
\]
where $C=C(\varepsilon)>0$ is a positive number.

\end{teor}
\begin{proof}
Indeed, the corresponding borders in the $x,t$-half-plane $\zeta=\frac{x}{\omega t}$ in terms of the borders in the $y,t$-half-plane 
$\xi=\frac{y}{\omega t}$ are calculated via (\ref{borders_xt_yt_elliptic}):
\begin{enumerate}
\item Let $\frac c\omega\in(0,1)\cup(1,3)\cup(3,\infty).$ Then $\xi_1=-\frac{1}{4}\(\frac{c+\omega}{\omega}\)^{3/2}$ and 
$\zeta_1=\frac{3c-\omega}{4\omega}.$

\item Let $\frac c\omega\in(1,3).$ Then $\xi_2=\frac{-2(c-\omega)}{\sqrt{\omega(c+\omega)}}
$ and 
$\zeta_2=\frac{2\omega^2-c\omega+c^2}{\omega(c+\omega)}.$

\item Let $\frac c\omega\in(0,1).$ Then $\xi_2=0$, $\xi_3=\frac{-2(c-\omega)}{\sqrt{\omega(c+\omega)}}
$ and 
$\zeta_2=\frac{c}{\omega}$, $\zeta_3=\frac{2\omega^2-c\omega+c^2}{\omega(c+\omega)}.$
\end{enumerate}

\end{proof}

%ggg

\section{Asymptotics in the domains $\frac{x}{\omega t}>2\(\frac{c}{\omega} + 1\)$ and  $\frac{x}{\omega
t}<\frac{3c-\omega}{4\omega}.$}\label{sect: Asymptotics} 
For the full description of asymptotic behavior of solution to the Cauchy problem (\ref{CH}), (\ref{init_cond}), 
we also include the 
following theorem, which describes the asymptotics in the soliton region. 
These asymptotics have been studied in (\cite{M15}). The short sketch for the solitonic region is as 
follows: we follow all Steps 1-4 in RH problem transformations Section \ref{sect: Asymptotics} with initial function 
$g_{\r}$ with the 
following modifications for the partial transmission coefficient $\Lambda(k):$
% ------------ 
$$\Lambda (k,\xi) = \prod\limits_{\kappa_0(\xi)+\sigma<\kappa_j}\frac{k+\i\kappa_j}{k-\i\kappa_j}\ ,
\quad \xi\geq 2\(\frac{c}{\omega} + 1\),\quad\textrm{ where
}\kappa_0:=\frac{1}{2}\sqrt{1-2/\xi}\ .$$
% -----------
Here $\sigma>0$ is a sufficiently small number.
Skipping the $\delta$-transformation of Step 4, Step 5 leads to a problem with jump matrix close to the identity matrix everywhere
 except for the parts of the contour close to poles $\i\kappa_j$ of the transmission coefficient.
Turning back to the explicitly solvable meromorphic problem with two poles and identity jump matrix, we come to
% ------------ 
\begin{teor}\label{Teor: soliton_ asympt}Under Assumptions 1-4 of Section \ref{sect: Introduction}, the asymptotics of the solution of the Cauchy problem  (\ref{CH})-(\ref{init_cond}) can be described as follows.
% -------------- 
\\
Let $v_j = \dsfrac{2}{1-4\kappa_j^2}, j=1,...,N$ and let $\varepsilon>0$ be
sufficiently small such that $|\kappa_j-\kappa_l|>4\varepsilon$
for any $j\neq l.$ Then in the soliton region $x\geq \(2\(\frac{c}{\omega} + 1\)+\delta\)\omega
t$ one has as $t\to\infty$
\\
\vskip1mm\noindent
\textbf{1.} if $\left|\dsfrac{x}{\omega t}-v_j\right|<\varepsilon$ for
some $j$, then
\\
$$u(x,t)=u_{\kappa_j,\delta,v_j,x_j}^{(sol)}(x,t)+\mathrm{O}(\e^{-C t}),\quad
m(x,t)=m_{\kappa_j,\delta,v_j,x_j}^{(sol)}(x,t)+\mathrm{O}(\e^{-C t}),$$
% -------- 
where $u_j^{(sol)}(x,t)$, $m_j^{(sol)}(x,t)$ are given below by
(\ref{soliton: y}) - (\ref{soliton: m}) and $C>0$ is a positive
constant; 
\\\textbf{2.} if $\left|\dsfrac{x}{\omega
t}-\dsfrac{2}{1-4\kappa_j^2}\right|\geq\varepsilon$ for all $j$,
then

$\hskip2cm u(x,t)=\mathrm{O}(\e^{-C t}),$
$m(x,t)=\mathrm{O}(\e^{-C t}).$
\\
Here $C = C(\varepsilon)$ is a positive constant
% -----------
and $u_{\kappa,\delta,v,x_j}$, $u_{\kappa,\delta,v,x_j}$ are given parametrically by the formulas
\begin{eqnarray}\label{soliton: y}
&&x =  y + \log \dsfrac{1+\delta\e^{-2\kappa\(y-\omega v t\)}
\frac{1+2\kappa}{1-2\kappa}}{1+\delta\e^{-2\kappa\(y-\omega v t\)}
\frac{1-2\kappa}{1+2\kappa}} +x_j
\\&&
\dsfrac{m_{\kappa,\delta,v,x_j}^{(sol)}(x,t)+\omega}{\omega} = \(1
+ \dsfrac{16\kappa^2}{1 - 4 \kappa^2}\dsfrac{\delta\e^{-2\kappa(y
- \omega v t)}}{\(1+\delta\e^{-2\kappa\(y-\omega v t\)}\)^2}\)^2,
\\&&
\dsfrac{u_{\kappa,\delta,v,x_j}^{(sol)}(x,t)}{\omega} =
\dsfrac{\dsfrac{32\kappa^2}{(1-4\kappa^2)^2}\delta\e^{-2\kappa\(y-\omega
v t\)}}{\(1 + \delta\e^{-2\kappa\(y-\omega v t\)}\)^2 +
\dsfrac{16\kappa^2}{1-4\kappa^2} \delta\e^{-2\kappa\(y-\omega v
t\)} }, \label{soliton: m}\end{eqnarray}
with
\begin{eqnarray*}
x_j = 2\sum\limits_{\kappa_l\geq\kappa_0+\varepsilon}\log\frac{1 +
2\kappa_l}{1 - 2 \kappa_l}, \quad \delta =
\frac{\gamma^2}{2\kappa_j \Lambda^2(\i\kappa_j, \xi)},\quad
\kappa=\kappa_j,\quad v_j = \dsfrac{2}{1-4\kappa_j^2}.
\end{eqnarray*}

\end{teor}

\textbf{Acknowledgements.} 
We thank Iryna Egorova, Dmitry Shepelsky, Robert Buckingham and Svetlana Roudenko 
for useful discussions.

The research has been supported by the project "Support of inter-sectoral mobility 
and quality enhancement of research teams at Czech Technical University in Prague", CZ.1.07/2.3.00/30.0034, sponsored by 
European Social Fund in the Czech Republic."


\begin{thebibliography}{}

%\bibitem{Bazargan_2008} Bazargan, Dzh. The direct and inverse scattering problems
%on the whole axis for the one-dimensional Schrodinger equation
%with step-like potential. (Russian) Dopov. Nats. Akad. Nauk Ukr.
%Mat. Prirodozn. Tekh. Nauki 2008, no. 4, 7--11.


\bibitem{BikN2} Bikbaev R F and Novokshenov V Yu 1989 Existence and uniqueness of the solution of the Whitham equation (Russian) \emph{ Asymptotic methods for solving problems in mathematical physics  Akad. Nauk SSSR Ural. Otdel. Bashkir. Nauchn. Tsentr Ufa} 81-95

\bibitem{Bikb1} Bikbaev R F 1989 Structure of a shock wave in the theory of the Korteweg-de Vries equation.  \emph{Phys. Lett. A} \textbf{141/5-6} 289-293

\bibitem{Bikb2} Bikbaev R F and  Sharipov R A 1989 The asymptotic behavior, as $t\to\infty$, of the solution of the Cauchy problem for
the Korteweg-de Vries equation in a class of potentials with
finite-gap behavior as $x\to\pm\infty$.  \emph{ Teoret. Mat.
Fiz.}  \textbf{78/3} 345-356  translation in \emph{Theoret. and
Math. Phys.} \textbf{78/3} 244-252


\bibitem{Bleher11}  Bleher, Pavel M. Lectures on random matrix models: the Riemann-Hilbert approach. 
Random matrices, random processes and integrable systems, 251–349, CRM Ser. Math. Phys., Springer, New York, 2011.

\bibitem{BDIK15} Thomas Bothner, Percy Deift, Alexander Its, Igor Krasovsky. On the asymptotic behavior of a log gas in the bulk scaling limit in the presence of a varying external potential I. 2015, arXiv:1407.2910.



\bibitem{BIK07} Boutet de Monvel A, Its A R and Kotlyarov V P 2007 Long-time asymptotics for the focusing NLS equation with time-periodic
boundary condition. \emph{C. R. Math. Acad. Sci. Paris.} \textbf{
345/11} 615-620


\bibitem{Shep2007} Boutet de Monvel, Anne; Shepelsky, Dmitry. Riemann-Hilbert problem
in the inverse scattering for the Camassa-Holm equation on the
line. Probability, geometry and integrable systems, 53--75, Math.
Sci. Res. Inst. Publ., 55, Cambridge Univ. Press, Cambridge, 2007.


\bibitem{Shep_2008_halh-line} Boutet de Monvel, Anne; Shepelsky, Dmitry. The Camassa-Holm
equation on the half-line: a Riemann-Hilbert approach. J. Geom.
Anal. 18 (2008), no. 2, 285--323.


\bibitem{Shep_2008} Boutet de Monvel, Anne; Shepelsky, Dmitry. Long-time asymptotics
of the Camassa-Holm equation on the line. Integrable systems and
random matrices, 99--116, Contemp. Math., 458, Amer. Math. Soc.,
Providence, RI, 2008.

%\bibitem{Egorova_2008} Boutet de Monvel, Anne; Egorova, Iryna; Teschl,
%Gerald. Inverse scattering theory for one-dimensional Schrodinger
%operators with steplike finite-gap potentials. J. Anal. Math. 106
%(2008), 271--316.


\bibitem{Kostenko_Shepelsky_Teschl_2009} Boutet de Monvel, Anne; Kostenko, Aleksey; Shepelsky, Dmitry;
Teschl, Gerald Long-time asymptotics for the Camassa-Holm
equation. SIAM J. Math. Anal. 41 (2009), no. 4, 1559--1588.

\bibitem{Shep_Its_2010} Boutet de Monvel, Anne; Its, Alexander; Shepelsky,
Dmitry.
Painleve-type asymptotics for the Camassa-Holm equation. SIAM J.
Math. Anal. 42 (2010), no. 4, 1854--1873.


\bibitem{BM14} Buckingham, Robert J.; Miller, Peter D. Large-degree asymptotics of rational Painlevé-II functions: noncritical behaviour. Nonlinearity 27 (2014), no. 10, 2489–2578.

\bibitem{BTVZ07} Buckingham, Robert; Tovbis, Alexander; Venakides, Stephanos; Zhou, Xin The semiclassical focusing nonlinear Schrödinger equation. Recent advances in nonlinear partial differential equations and applications, 47–80, Proc. Sympos. Appl. Math., 65, Amer. Math. Soc., Providence, RI, 2007.

\bibitem{BV} Buckingham R and Venakides S 2007 Long-time asymptotics of the non-linear Schrodinger equation shock problem. \emph{Comm. Pure Appl. Math.} \textbf{60/9} 1349-1414


\bibitem{Buslaev_Fomin} Buslaev, V.; Fomin, V. An inverse scattering problem for the one-dimensional Schrödinger equation on the entire axis. (Russian) Vestnik Leningrad. Univ. 17 (1962), no. 1, 56–64.


\bibitem{Camassa Holm 1993} Camassa, R.; Holm, D. An integrable
shallow water equation with peaked solitons. Phys. Rev. Lett. 71
(1993), 1661-1664.

\bibitem{Camassa Holm Hyman 1994} Camassa, R.; Holm, D.; Hyman, J. A new integrable
shallow water equation. Adv. Appl. Mech. 31 (1994), 1-33.

\bibitem{CV07} Claeys, T.; Vanlessen, M. Universality of a double scaling limit near singular edge points in random matrix models. Comm. Math. Phys. 273 (2007), no. 2, 499–532.


%\bibitem{Constantin 2000}  Constantin, Adrian Existence of permanent and breaking waves for a shallow water 
%equation: a geometric approach. Ann. Inst. Fourier (Grenoble) 50 (2000), no. 2, 321–362.

\bibitem{Cohen_Kappeler} Cohen, Amy; Kappeler, Thomas. Scattering and inverse scattering
for steplike potentials in the Schrodinger equation. Indiana Univ.
Math. J. 34 (1985), no. 1, 127--180.


\bibitem{Constantin_2001} Constantin, Adrian On the scattering problem for
the Camassa-Holm equation. R. Soc. Lond. Proc. Ser. A Math. Phys.
Eng. Sci. 457 (2001), no. 2008, 953--970.

\bibitem{Constantin_Gerdikov_Ivanov} Constantin, Adrian; Gerdjikov, Vladimir S.; Ivanov, Rossen I.
Inverse scattering transform for the Camassa-Holm equation.
Inverse Problems 22 (2006), no. 6, 2197--2207.

\bibitem{Dai} Dai, H.-H. Model equations for nonlinear dispersive
waves in a compressible Mooney -- Rivlin rod. Acta Mechanica 127
(1998), 293-308.

\bibitem{DZ93} Deift, P., Zhou, X., A steepest descent method for oscillatory Riemann-Hilbert
problems. Asymptotics for the MKdV equation, Ann. of Math. (2), 137 (1993),
no.2, 295--368

\bibitem{EGKT} Egorova I, Gladka Z, Kotlyarov V and Teschl G 2012  Long-Time Asymptotics for the Korteweg-de Vries Equation with Steplike Initial
Data. \emph{Nonlinearity } \textbf{26/7} 1839-1864


\bibitem{Egorova_Gladka_Kotlyarov_Teschl} Egorova, Iryna; Gladka, Zoya; Kotlyarov, Volodymyr; Teschl,
Gerald. Long-time asymptotics for the Korteweg--de Vries equation
with step-like initial data. Nonlinearity 26 (2013), no. 7,
1839--1864.

\bibitem{Fokas Fuchssteiner 1981} Fokas, A.; Fuchssteiner, B.
Symplectic structures, their Backlund transformation and
hereditary symmetries. Physica D 4 (1981), 47-66.

\bibitem{Germain Pusateri Rousset} Germain~P., Pusateri~F., and Rousset~F. Asymptotic stability of solitons 
for mKdV. arXiv:1503.093143, 2015.


\bibitem{Grunert_Holden_Raynaud_2011} Grunert~K., Holden~H.,
Raynaud~X. Global conservative solutions of the Camassa -- Holm
equation for initial data with nonvanishing asymptotics.
arXiv:1106.4125v1 [math.AP] 21.06.2011

\bibitem{GP} Gurevich, A.~V.; Pitaevskii, L.~P. Decay of Initial Discontinuity in the Korteweg-de Vries Equation. (Russian)
\emph{JETP Letters} (1973), \textbf{17/5} 193


\bibitem{Kh2} Khruslov E Ya 1976 Asymptotics of the solution of the Cauchy prob\-lem for the Korteweg de Vries equation with initial data of step type. \emph{Matem. Sbornik (New Series)} \textbf{99(141):2} 261-281

\bibitem{KK2} Khruslov E Ya and Kotlyarov V P 1994 Soliton asymptotics of nondecreasing solutions of nonlinear completely integrable evolution equations \emph{    Spectral operator theory and
related topics} Adv. Soviet Math. \textbf{19} Amer. Math. Soc.
Providence, RI 129-180



\bibitem{Kotlyarov_Minakov_2010} Kotlyarov~V., Minakov~A. Riemann
-- Hilbert problem to the modified Korteweg -- de Vries equation:
Long-time dynamics of the step-like initial data. Journal of
Mathematical Physics 51, 093506, 2010.

\bibitem{KM12} 
V.~Kotlyarov and A.~Minakov,
{Step-initial function to the mkdv equation: hyperelliptic
long-time asymptotics of the solution,} Journal of mathematical
physics, analysis, geometry, 2012, 8/1, P.~
38--62.

\bibitem{KM15} 
V.~Kotlyarov and A.~Minakov,
{Modulated elliptic wave and asymptotic solitons in a chock problem to the modified Korteweg-de Vries equation,} 
Journal of Phyics A: Mathematical and Theoretical, \textbf{48}, 2015, 305201, 35 pp.


%\bibitem{Marchenko_1972_1986} Marchenko, Vladimir A. Sturm-Liouville operators and applications.
%Translated from the Russian by A. Iacob. Operator Theory: Advances
%and Applications, 22. Birkhauser Verlag, Basel, 1986. xii+367 pp.


\bibitem{M15} Minakov~A. Riemann--Hilbert problem for the Camassa--Holm equation with step-like initial data, J.~Math.~Anal.~Appl. (2015)


\bibitem{Mizumachi-Tzvetkov}  Mizumachi, Tetsu; Tzvetkov, Nikolay L2-stability of solitary waves for the KdV 
equation via Pego and Weinstein's method. Harmonic analysis and nonlinear 
partial differential equations, 33–63, RIMS Kôkyûroku Bessatsu, B49, Res. Inst. Math. Sci. (RIMS), Kyoto, 2014

\bibitem{Pego-Weinstein}  Pego, Robert L.; Weinstein, Michael I. 
Asymptotic stability of solitary waves. Comm. Math. Phys. 164 (1994), no. 2, 305–349.


\bibitem{Zakharov_Shabat_74_79}A.~B.~Shabat and V.~E.~Zakharov, Integration of nonlinear equations of mathematical phyics by inverse scattering problem method, I,II, Funct. anal. and applications, \textbf{8/3}, 1974, \textbf{13/3},1979. (in Russian).

\bibitem{Tzvetkov}  Tzvetkov, Nikolay On the long time behavior of KdV type equations [after Martel-Merle]. 
Séminaire Bourbaki. Vol. 2003/2004. Astérisque No. 299 (2005), Exp. No. 933, viii, 219–248.

%\textbf{\textcolor{green}{Add history Buck, Deift, Clayes}}

\end{thebibliography}
\end{document}